%% file: ms.tex
\documentclass[journal, onecolumn, 10pt, a4paper]{IEEEtran}

\usepackage{amssymb}
\usepackage{amsmath}
\usepackage{cite}
\usepackage{float}
\usepackage{tikz}
\usepackage{mathtools}
\usepackage{amsthm}
\usepackage{array}
\usepackage{enumitem}

\usepackage{empheq}

\usepackage{multirow}

\usetikzlibrary{patterns,arrows,decorations.pathreplacing}

\newcommand{\Worst}{\textnormal{Worst}}
\newcommand{\XOR}{\textnormal{XOR}}
\newcommand{\p}{\textnormal{P}}

\newcommand{\Sym}{\textnormal{sym}}
\newcommand{\tot}{\textnormal{tot}}
\newcommand{\s}{\textnormal{sec}}

\newcommand{\key}{\textnormal{Key}}
\newcommand{\enc}[2]{\textsf{sec}\big(#1, #2\big)}

\newcommand{\C}{\mathcal{C}}
\newcommand{\D}{\mathcal{D}}
\newcommand{\K}{\mathcal{K}}
\newcommand{\Ka}{\mathsf{K}}
\newcommand{\Ks}{\mathsf{K}_s}
\newcommand{\Kw}{\mathsf{K}_w}
\newcommand{\Ma}{\mathsf{M}}
\newcommand{\Ms}{\mathsf{M}_s}
\newcommand{\Mw}{\mathsf{M}_w}
\newcommand{\Da}{\mathsf{D}}

\newtheorem{theorem}{Theorem}
\newtheorem{lemma}{Lemma}
\newtheorem{proposition}{Proposition}
\newtheorem{corollary}{Corollary}
\newtheorem{definition}{Definition}

\newtheorem{remark}{Remark}

\allowdisplaybreaks[4]

\usepackage{fancyhdr}
 
\begin{document}

\title{Secrecy Capacity-Memory Tradeoff of  Erasure Broadcast Channels}
\author{Sarah Kamel, Mireille Sarkiss, Mich\`ele Wigger, and Ghaya Rekaya-Ben Othman
\thanks{S. Kamel, M.~Wigger, and G. Rekaya-Ben Othman are with   LTCI,  Telecom ParisTech, Universit\'e Paris-Saclay, 75013 Paris. Email: \{sarah.kamel, michele.wigger, ghaya.rekaya\}@telecom-paristech.fr}
\thanks{M. Sarkiss is with CEA, LIST, Communicating Systems Laboratory, BC 173, 91191 Gif-sur-Yvette, France. Email: mireille.sarkiss@cea.fr}
\thanks{Parts of the material in this paper have been presented at the  \emph{IEEE International Conference on Communications (ICC)}, Paris, France, May 2017 \cite{ICC17}, and at the \emph{IEEE Information Theory Workshop (ITW)}, Kaohsiung, Taiwan, Nov. 2017. }}

\maketitle

\begin{abstract}
This paper derives upper and lower bounds on the \emph{secrecy capacity-memory tradeoff} of a wiretap erasure broadcast channel (BC) with $\Kw$ weak receivers and $\Ks$ strong receivers, where weak receivers, respectively strong receivers, have same erasure probabilities and cache sizes. The lower bounds are achieved by schemes that meticulously combine joint cache-channel coding with wiretap coding and key-aided one-time pads. The presented upper bound holds  more generally for arbitrary degraded BCs and arbitrary cache sizes. When only weak receivers have cache memories, upper and lower bounds coincide for small and large cache memories, thus providing the exact secrecy capacity-memory tradeoff for this setup.  The derived bounds  allow to further conclude that the secrecy capacity is positive even when the eavesdropper is stronger than all the legitimate receivers with cache memories. 
	 Moreover, they show that the secrecy capacity-memory tradeoff can be significantly smaller than  its non-secure counterpart, but it grows much faster when cache memories are small. 
	 
	 The paper also presents a lower bound on the \emph{global secrecy capacity-memory tradeoff} where one is allowed to optimize the cache assignment subject to a total cache  budget. It is close to the  best known lower bound without secrecy constraint. For small total cache budget, the global secrecy capacity-memory tradeoff is achieved by assigning all  the available cache memory uniformly over \emph{all} receivers if the eavesdropper is stronger than all legitimate receivers, and it is achieved by assigning the cache memory uniformly only over the \emph{weak} receivers if the eavesdropper is weaker than the strong receivers.
\end{abstract}\bigskip

\section{Introduction} 

	Traffic load in communication systems varies tremendously during the day between busy periods where the network is highly congested causing  packet loss, delivery delays and unsatisfied users and other periods where the network is barely used. Lately, caching has emerged as a promising technique to reduce the network load and latency in such dense wireless networks.  In caching, the communication is divided into two phases: the caching phase and the delivery phase. The caching phase occurs during the off-peak periods of the network, where fragments of popular contents are stored in users' cache memories or on nearby servers. The delivery phase occurs when users request specific files during the peak-traffic periods of the network and they are served partly from their cache memories and partly from the server. The technical challenge in these networks is that in the caching phase the servers do not know exactly which files the receivers will demand during the delivery phase. They are thus obliged to store information about \emph{all possibly requested files} in the receivers' cache memories.
	
	 Maddah-Ali and Niesen showed in their seminal work  \cite{MaddahAli14} that the delivery (high-traffic) communication can benefit from the cache memories more than the obvious local caching gain arising from locally retrieving parts of the requested files. The additional gain, termed \emph{global caching gain}, is obtained through carefully designing the cached contents and applying a new coding scheme, termed \emph{coded caching}, which allows the transmitter to simultaneously serve multiple receivers. In  \cite{MaddahAli14} and in many subsequent works, the delivery phase is modelled as  an error-free broadcast channel and  all receivers have equal cache sizes. 	 
	 Coded caching is straightforwardly extended to noisy BCs by means of a \emph{separate cache-channel coding} architecture where a standard BC code that ignores the cache contents is used to transmit the \emph{delivery-messages}\footnote{Due to the presence of the cache memories the messages conveyed in the delivery phase,  generally differ from the original messages in the library.} produced by the coded caching algorithm to the receivers. Improved global caching gains can be achieved by employing \emph{joint cache-channel coding} where the encoder and the decoder simultaneously adapt to the cache contents and the channel statistics \cite{Wigger15, WiggerJournal,tulino, Gunduz}. 
	
	A different line of works has addressed secrecy issues in cache-aided BCs \cite{Clancy15, Prabhakarany16,SuthanChughKrishnan}, where  delivery communication takes place over a noiseless link. 
	Different secrecy requirements have been studied. 
	In \cite{Clancy15,SuthanChughKrishnan}, the entire library of messages needs to be kept secret from an external eavesdropper that has access to the outputs of the common bit-pipe but not to the cache memories. This is achieved by means of securing the XOR packets produced by coded caching with  secret keys that have been prestored at the receivers during the caching phase \cite{Clancy15}. This approach has subsequently been extended to resolvable networks in \cite{ZewailYener16} and to device-to-device communication models
	\cite{Awan15}.
	In \cite{Prabhakarany16}, each legitimate receiver acts also as eavesdropper and is thus not allowed to learn anything about the files requested by the other receivers. In this case, uncoded fragments of the messages cannot be stored in the users' caches. Instead, random keys and combinations of messages XORed with these random keys are cached. In the delivery phase, messages (or combination of messages) XORed with random keys are transmitted in a way that each message can be decoded only by its intended receiver.
	 Secret communication has also been considered for cache-aided heterogeneous and multiantenna interference networks, see for example  \cite{edge}, \cite{multihop}, and \cite{jamming}. A different angle of attack on this problem is taken in \cite{EngelmannElia}, which presents a privacy-preserving protocol that prevents eavesdroppers from learning which users are requesting which files as well as the statistical popularities of the files. 

	In this paper, we follow the secrecy requirement in \cite{Clancy15} where an external eavesdropper is not allowed to learn anything about \emph{all} the files in the library, but here delivery communication takes place over an erasure BC with one transmitter and $\Ka\geq 2$ receivers.
The eavesdropper  does not have access to the cache memories  but overhears the delivery  communication over the erasure BC.  
	The main interest of this paper is  the \emph{secrecy capacity-memory tradeoff} of such a system, i.e., the largest message rate that allows to find encoding and decoding functions so that the normalized mutual information \begin{equation}\label{eq:mut}\frac{1}{n}I(W_1,\ldots, W_\Da; Z^n)\end{equation} between all the messages $W_1,\ldots, W_\Da$ of the library and the eavesdropper's observations $Z^n$ vanishes asymptotically for increasing blocklengths $n$. In our previous work 
 \cite{WCNC17}, we have addressed the weaker secrecy constraint where  the eavesdropper is not allowed to learn any information about each of the actually demanded messages \emph{individually}. Each of the $\Ka$ mutual informations  $\frac{1}{n}I(W_{d_1};Z^n), \ldots, \frac{1}{n} I(W_{d_\Ka};Z^n)$  thus needs to vanish asymptotically, where  $W_{d_1}, \ldots, W_{d_\Ka}$ denoting the files requested and delivered to Receivers~$1,\ldots, \Ka$.   Our conference publications \cite{WCNC17} and \cite{ICC17} suggest that the ultimate performance limits under these  two secrecy constraints are close.
	
	There are two basic approaches to render cache-aided BCs secure: 1) introduce binning to the non-secure coding schemes so as to  transform them into wiretap codes \cite{elGamalBook}; and 2)  store  secret keys into the cache memories and apply one-time pads \cite{shannon_cipher} to parts of the delivery-messages which can then be sent using  non-secure channel codes. As we will see, in the  BC scenarios where different receivers have different cache sizes, combinations thereof should be employed. Moreover, by applying superposition coding or joint cache-channel coding,  secret keys stored in caches can even be used to secure the communication to other receivers. Under secrecy constraints, further global caching gains are thus possible than the previously  reported gains for non-secure communication \cite{MaddahAli14,WiggerJournal}. 

Based on the described  coding ideas, we propose lower bounds on the secrecy capacity-memory tradeoff of the cache-aided erasure BC in Figure~\ref{fig:systemModel}. The system  consists of $\Kw$ weak receivers $1,\ldots, \Kw$ that have  equal erasure probability $\delta_w \geq 0$ and cache size $\Mw$, and $\Ks$ strong receivers $\Kw+1,\ldots, \Ka$ that have equal erasure probability $\delta_s \geq 0$ and cache size $\Ms$. 
We also provide a general upper bound on the secrecy capacity-memory tradeoff of an arbitrary degraded BC with receivers having arbitrary cache sizes.
Upper and lower bounds match for the  setup in Figure~\ref{fig:systemModel} in special cases, for example, when $\Ms=0$ and $\Mw$ is sufficiently large or small. The proposed bounds moreover allow to conclude the following:
\begin{itemize}
	\item \underline{Secrecy Capacity is Always Positive:} In particular, the secrecy capacity is  positive even when the eavesdropper is stronger than all the legitimate receivers with cache memories.\\
	\item \underline{The Secrecy Constraint Can Significantly Harm Capacity:} The secrecy capacity-memory tradeoff can be significantly smaller than its non-secure counterpart, especially when only weak receivers have cache memories.  One explanation is that when only weak receivers have cache memories, the communication to the strong receivers needs to be either secured through wiretap binning  or other secrecy mechanisms, which both significantly reduce the rate of communication.\\
	\item \underline{Caching \emph{Gains} are More Important under a Secrecy Constraint: } In the regime of small cache memories,    the gains of caching (i.e., the slope of the capacity) are  more pronounced in our system with secrecy constraint  than in the standard non-secure system. Consider for example, $\Ms  =0$ and $\Mw $  sufficiently small. In this regime, when the eavesdropper's erasure probability $\delta_z$ is larger than erasure probability at the strong receivers $\delta_s$, the  slope $\gamma_\s$ of   the secrecy capacity-memory tradeoff satisfies (see Corollary~\ref{cor:1})
	\begin{equation} \gamma_\s =  \frac{\Kw(\delta_z-\delta_s)}{\Kw(\delta_z-\delta_s)+\Ks (\delta_z-\delta_w)^{+}}.
	\end{equation} 
	The slope $\gamma$ of the standard non-secure capacity-memory tradeoff satisfies (see  \cite[Theorem~2]{Wigger15})
	\begin{equation}
\gamma \leq \frac{\Kw}{\Da}.
	\end{equation}
	This latter slope $\gamma$ thus deteriorates with increasing library size $\Da$, which is  not the case for $\gamma_\s$. The main reason for this  behavior is that in a standard system the cache memories are filled with  data, and intuitively each stored bit is useful only under some of the demands. In a secure system, a good option is to store secret keys in the cache memories of the receivers. These secret keys are helpful for all possible demands, and therefore the caching gain does not degrade with the library size $\Da$.   \\
\item \underline{Optimal Cache Assignments for Small Total Cache Budgets:} For small total cache  budgets, the global secrecy capacity-memory tradeoff is achieved by assigning all of the cache memory uniformly  only over the \emph{weak} receivers if the eavesdropper is weaker than strong receivers, and it is achieved by assigning all the cache memory uniformly over \emph{all} receivers, if the eavesdropper is stronger than all receivers.\\
\end{itemize}

	{\textit{Paper Organization:}}
	The remainder of this paper is organized as follows. Section \ref{sec:probDef}  formally defines the  problem. Section \ref{sec:UB} presents a general upper bound on the secrecy capacity-memory tradeoff and specializes it to the specific model of this work. In   Sections \ref{sec:1sided} and \ref{sec:2sided}, we present  our  results for the scenarios when only weak receivers have cache memories and when all receivers have cache memories, and we sketch the  coding schemes  achieving the proposed lower bounds. Section \ref{sec:globalCs} contains our results on the global secrecy capacity-memory tradeoff. Sections \ref{sec:proofLB_Aissgn} and \ref{sec:proofLB}  describe and analyze in detail all the coding schemes  proposed  in this paper. Finally, Section \ref{sec:concl} concludes the paper. \bigskip

\section{Problem Definition}\label{sec:probDef} 

\subsection{Channel Model} 
	We consider a wiretap erasure BC with a single transmitter, $\Ka$ receivers and one eavesdropper, as shown in Figure~\ref{fig:systemModel}. The input alphabet of the BC is
	\begin{equation}
		\mathcal{X}  {=} \{0,1\}
	\end{equation}	 
	and all receivers and the eavesdropper have the same output alphabet
	\begin{equation}
		\mathcal{Z} {=} \mathcal{Y}  {=} \mathcal{X} \cup \Delta.
	\end{equation}
	The output erasure symbol $\Delta$ indicates the loss of a bit at the receiver. 
	Let $\delta_k$ be the erasure probability of Receiver~$k$'s channel, for $k \in \K :=\{1,\ldots,\Ka\}$, and $\delta_z$ be the erasure probability of the eavesdropper's channel.	
	Hence, the marginal transition laws of this BC are described by 
	\begin{equation}
		\mathbb{P}[Y_k=y_k|X=x] = \left\{
		\begin{tabular}{ c l }
  			$1-\delta_k $ & if $~y_k=x$ \\
  			$\delta_k $ & if $~y_k = \Delta  ~~, ~~ \forall ~ k \in \K$ \\
  			$0$ & otherwise  \\
		\end{tabular} \right.
	\end{equation}
	and 
	\begin{equation}
		\mathbb{P}[Z=z|X=x] = \left\{
		\begin{tabular}{ c l }
  			$1-\delta_z $ & if $~z=x$ \\
  			$\delta_z $ & if $~z = \Delta$ \\
  			$0$ & otherwise  \\
		\end{tabular} \right.
	\end{equation}
	with $0 \leq \delta_k,\delta_z \leq 1$.
	
	The $\Ka$ receivers are partitioned into two sets. The first set
	\begin{equation}
		\K_w := \{1,\dots,\Kw\}
	\end{equation}	 
	is formed by $\Kw$ weak receivers which have bad channel conditions. The second set
	\begin{equation}
		\K_s  := \{\Kw+1,\dots,\Ka\}
	\end{equation}	 
	is formed by $\Ks =\Ka-\Kw$ strong receivers which have good channel conditions.

	In other words, we assume that 
	\begin{equation}
		\delta_k = \left\{
		\begin{tabular}{ c l }
  			$\delta_w $ & if $~k \in \K_w$ \\
  			$\delta_s $ & if $~k \in \K_s $ \\
		\end{tabular} \right.
	\end{equation}
	with
	\begin{equation}
		0\leq \delta_s \leq \delta_w \leq 1.
	\end{equation} 
	In a standard wiretap erasure BC, reliable communication with positive  rates is only possible when the eavesdropper's erasure probability $\delta_z$ is larger than the erasure probabilities at all legitimate receivers. Here, any configuration of erasure probabilities is possible, because the legitimate receivers can prestore information in local cache memories and this information is not accessible by the eavesdropper. 
Specifically, we assume that  each  weak receiver has access to a local cache memory of size $n\Mw $ bits and each strong receiver has access to a local cache memory of size $n\Ms  $ bits.

	\begin{figure}
	\begin{center}
		\begin{tikzpicture}[scale=1]
			\node at (6.06,5.8) {Library};
			\node at (6.06,5.3) {$W_1, W_2, \dots, W_\Da$};
			\draw[rounded corners=7pt,thick] (4.47,4.9) rectangle (7.56,6.2);
			\draw[thick] (6.06,4.9) -- (6.06,4.6);

			\node at (4.91,4.25) {$\theta$}; 
			\draw[rounded corners=7pt,thick] (4.56,3.9) rectangle (5.26,4.6);
			\draw[thick] (5.26,4.25) -- (5.56,4.25);
			\node[thick] at (6.06,4.25) {$T_{X}$};
			\draw[thick] (5.56,3.9) rectangle (6.56,4.6);
			\node[right,thick] at (6.06,3.45) {$X^n$};
			\draw[->,thick] (6.06,3.9) -- (6.06,3);
			\draw[thick] (-1.38,2) rectangle (13.5,3);
			\node[thick] at (6.06,2.5) {Erasure Broadcast Channel};

			\draw[->,thick] (0.36,2) -- (0.36,1.1);
			\draw[->,thick] (3.6,2) -- (3.6,1.1);
			\draw[->,thick] (6.6,2) -- (6.6,1.1);
			\draw[->,thick] (9.8,2) -- (9.8,1.1);
			\draw[->,thick] (12.3,2) -- (12.3,1.1);
			\node[right,thick] at (0.36,1.55) {$Y_1^n$};
			\node[right,thick] at (3.6,1.55) {$Y_{\Kw}^n$};
			\node[right,thick] at (6.6,1.55) {$Y_{\Kw+1}^n$};
			\node[right,thick] at (9.8,1.55) {$Y_{\Ka}^n$};
			\node[right,thick] at (12.3,1.55) {$Z^n$};

			\draw[thick] (-0.34,0.4) rectangle (1.06,1.1);
			\node[thick] at (0.36,0.75) {Rx $1$};
			\node[thick] at (1.51,0.75) {$\dots$};
			\draw[thick] (2.9,0.4) rectangle (4.3,1.1);
			\node[thick] at (3.6,0.75) {Rx $\Kw$};
			\draw[thick] (5.8,0.4) rectangle (7.4,1.1);
			\node[thick] at (6.6,0.75) {Rx\! $\Kw\!+\!1$};
			\node[thick] at (7.75,0.75) {$\dots$};
			\draw[thick] (9.1,0.4) rectangle (10.5,1.1);
			\node[thick] at (9.8,0.75) {Rx $\Ka$};
			\draw[thick] (11.1,0.4) rectangle (13.5,1.1);
			\node[thick] at (12.3,0.75) {Eavesdropper};

			\draw[thick] (-0.54,0.75) -- (-0.34,0.75);
			\draw[rounded corners=7pt,thick] (-1.38,0.4) rectangle (-0.54,1.1);
			\node[thick] at (-0.96,0.75) {$n\Mw $};
			\draw[thick] (2.7,0.75) -- (2.9,0.75);
			\draw[rounded corners=7pt,thick] (1.86,0.4) rectangle (2.7,1.1);
			\node[thick] at (2.28,0.75) {$n\Mw $};
			\draw[thick] (5.6,0.75) -- (5.8,0.75);
			\draw[rounded corners=7pt,thick] (4.8,0.4) rectangle (5.6,1.1);
			\node[thick] at (5.2,0.75) {$n\Ms  $};
			\draw[thick] (8.9,0.75) -- (9.1,0.75);
			\draw[rounded corners=7pt,thick] (8.1,0.4) rectangle (8.9,1.1);
			\node[thick] at (8.5,0.75) {$n\Ms  $};	
			
			\draw[->,thick] (0.36,0.4) -- (0.36,-0.1);
			\draw[->,thick] (3.6,0.4) -- (3.6,-0.1);
			\draw[->,thick] (6.6,0.4) -- (6.6,-0.1);
			\draw[->,thick] (9.8,0.4) -- (9.8,-0.1);
			\node[below,thick] at (0.36,-0.1) {$\hat{W}_1$};
			\node[below,thick] at (3.6,-0.1) {$\hat{W}_{\Kw}$};
			\node[below,thick] at (6.6,-0.1) {$\hat{W}_{\Kw+1}$};
			\node[below,thick] at (9.8,-0.1) {$\hat{W}_{\Ka}$};
			
			\draw [thick, decorate,decoration={brace,amplitude=10pt,mirror},xshift=0.4pt,yshift=-0.4pt](-1.28,-0.8) -- (4.3,-0.8) node[midway,yshift=-0.6cm] {Weak receivers};
			\draw [thick, decorate,decoration={brace,amplitude=10pt,mirror},xshift=0.4pt,yshift=-0.4pt](4.9,-0.8) -- (10.5,-0.8) node[midway,yshift=-0.6cm] {Strong receivers};
		\end{tikzpicture} 
		\caption{Erasure BC with $\Ka=\Kw+\Ks $ legitimate receivers and an eavesdropper. The $\Kw$ weaker receivers have cache memories of size $\Mw $ and the $\Ks $ stronger receivers have cache memories of size $\Ms  $. The random variable $\theta$ models a source of randomness locally available at the transmitter.}  \label{fig:systemModel}
	\end{center}
	\end{figure}
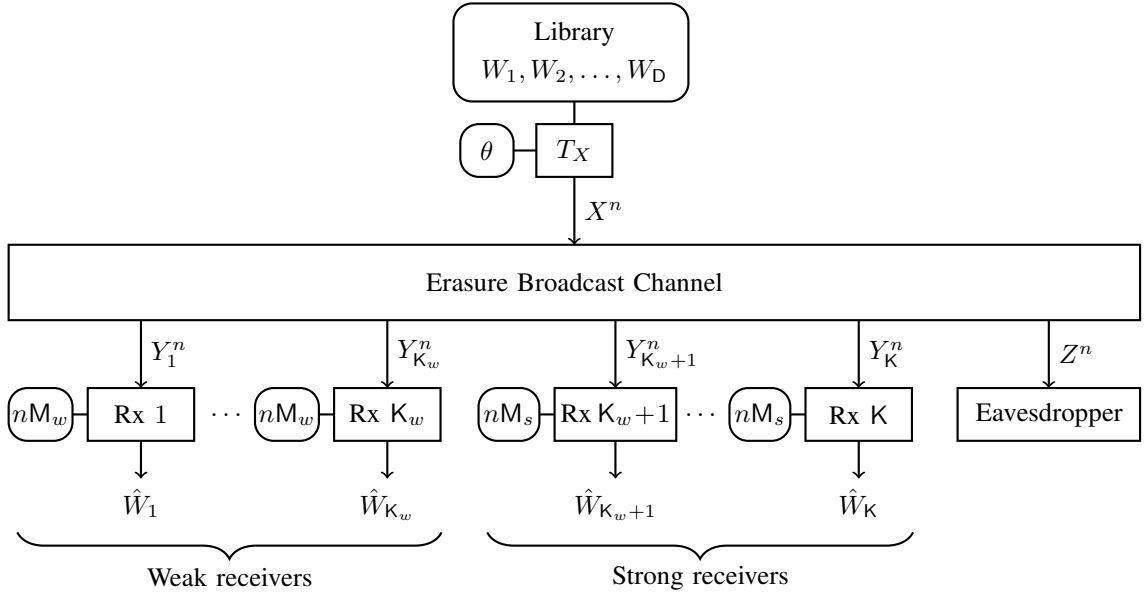

\subsection{Library and Receiver Demands}

	The transmitter can access a library of $\Da > \Ka$ independent messages
	\begin{equation}
		W_1,\dots,W_\Da
	\end{equation}		
	of rate $R \geq 0$ each. So for each $d \in \D$,
	\begin{equation}
		\D := \{1,\ldots, \Da \},
	\end{equation}	
	message~$W_d$ is uniformly distributed over the set $\big\{1,\dots,\lfloor 2^{nR}\rfloor \big\}$, where $n$ is the transmission blocklength.

	Every receiver~$k \in \K$ 
	demands exactly one message $W_{d_k}$ from the library. We denote the demand of Receiver~$k$ by $d_k \in \D$
	and  the \emph{demand vector} of all receivers by
	\begin{equation}
		\mathbf{d} := (d_1,\dots,d_\Ka) \in \D^\Ka.
	\end{equation}

Communication takes place in two phases: a first caching phase where the transmitter sends caching information to be stored in the receivers' cache memories and the subsequent delivery phase where the demanded messages $W_{d_k}$, for $k \in \K$, are conveyed to the receivers.

\subsection{Placement Phase}

	The placement phase takes place during periods of low network-traffic. The transmitter therefore has all the available resources to transmit the  cached contents  in an error-free and secure fashion, and this first phase is only restricted by the storage constraints on the cache memories. 
	However, since the caching phase takes place before the receivers demand their files, the cached content cannot depend on the demand vector $\mathbf{d}$, but only on the library and  local randomness  $\theta$ that is accessible by the transmitter. The cache content $V_k$ stored at receiver $k\in\{1,\ldots,\Ka\}$ is thus of the form
	\begin{equation}\label{eq:caching}
		V_k = g_k \left( W_1,\dots,W_\Da ,\theta \right),
	\end{equation}
	for some caching function
	\begin{equation}\label{eq:cachingFct}
		g_k \colon \left\lbrace 1,\dots,\lfloor 2^{nR}\rfloor \right\rbrace^\Da \times \Theta \rightarrow \mathcal{V}_i
	\end{equation}
 	where for $k \in \K_w$
 	\begin{equation}
 		\mathcal{V}_k := \left\lbrace 1,\dots,\lfloor 2^{n\Mw }\rfloor \right\rbrace, \quad k \in \K_w
 	\end{equation}
 	and for $k \in \K_s $
 	\begin{equation}
 		\mathcal{V}_k := \left\lbrace 1,\dots,\lfloor 2^{n\Ms  }\rfloor \right\rbrace, \quad k \in \K_s .
 	\end{equation}

\subsection{Delivery Phase}
	Prior to the delivery phase, the demand vector $\mathbf{d}$ is learned by the transmitter and the legitimate receivers. The communication of the demand vector requires zero communication rate since it takes only $\Ka \cdot \big\lceil \log(\Da) \big\rceil$ bits to describe $\mathbf{d}$. 
	
	Based on the demand vector $\mathbf{d}$, the transmitter sends 
	\begin{equation}\label{eq:encoding}
		X^n = f_{\mathbf{d}} \left(W_1,\dots,W_{\Da},\theta \right),
	\end{equation} 
	for some function
	\begin{equation} \label{eq:encodingFct}
		f_{\mathbf{d}} \colon \left\lbrace 1,\dots,\lfloor 2^{nR}\rfloor \right\rbrace^{\Da} \times \Theta \rightarrow \mathcal{X}^n.
	\end{equation}
	
	Each Receiver $k \in \K$,  attempts to decode its demanded message $W_{d_k}$ based on its observed outputs $Y_k^n$ and its cache content $V_k$: 
	\begin{equation}\label{eq:decoding}
		\hat{W}_{k} := \varphi_{k,\mathbf{d}}(Y_k^n,V_k), \quad k \in \K,
	\end{equation} 
	for some function 
	\begin{equation}\label{eq:decodingFct}
		\varphi_{k,\mathbf{d}} \colon \mathcal{Y}^n \times \mathcal{V}_k \rightarrow \left\lbrace 1,\dots,\lfloor2^{nR}\rfloor \right\rbrace.
	\end{equation}

\subsection{Secrecy Capacity-Memory Tradeoff}
	A decoding error occurs whenever $\hat{W}_{k} \neq W_{d_k}$, for some $k\in \K$. We consider the worst-case probability of error over all feasible demand vectors
	\begin{equation} \label{eq:PeWorst}
	\p_e^{\Worst} := \max\limits_{\mathbf{d}\in \D^{\Ka}} \p \left[ \bigcup_{k=1}^{\Ka} \left\lbrace \hat{W}_{k} \neq W_{d_k} \right\rbrace \right].
	\end{equation} 

	The communication is considered secure if the eavesdropper's channel outputs $Z^n$ during the delivery phase provide no information about the entire library:
	\begin{equation} \label{eq:jointSecrecy}
		\lim\limits_{n \rightarrow \infty} \frac{1}{n} I \left( W_1,\ldots,W_{\Da};Z^n \right) < \epsilon.
	\end{equation}	
	
	As mentioned previously, the caching phase is supposed to be kept secure from the eavesdropper. This latter thus eavesdrops only on the delivery communication.
	
	\begin{definition}
		A rate-memory triple $(R,\Mw ,\Ms)$ is \emph{securely achievable} if for every $\epsilon>0$ and sufficiently large blocklength $n$, there exist caching, encoding, and decoding functions as in \eqref{eq:cachingFct}, \eqref{eq:encodingFct}, and \eqref{eq:decodingFct} so that
		\begin{equation} \label{eq:Cs_def}
			\p_e^{\Worst} \leq \epsilon  \qquad \mbox{and} \qquad \frac{1}{n} I\left(W_1,\ldots,W_{\Da};Z^n\right) < \epsilon.
		\end{equation}
	\end{definition}
	
	In our previous works \cite{WCNC17} and \cite{ICC17}, we called this secrecy constraint a \emph{joint secrecy constraint} to distinguish it from the \emph{individual secrecy constraint} in \cite{WCNC17} where the second inequality in \eqref{eq:Cs_def} is replaced by $\frac{1}{n} I\left( W_{d_k};Z^n \right) < \epsilon, \forall k \in \K$.

	\begin{definition} \label{def:Cs}
		Given cache memory sizes $\left(\Mw , \Ms \right)$, the \emph{secrecy capacity-memory tradeoff $C_{\s}(\Mw ,\Ms  )$} is the supremum of all rates $R$ so that the triple $\left(R,\Mw ,\Ms \right)$ is securely achievable:
		\begin{equation}
		C_{\s}\left(\Mw ,\Ms  \right) := \sup \left\lbrace R \colon \;\; \left(R, \Mw ,\Ms  \right)  \; \textnormal{securely achievable} \right\rbrace.
		\end{equation}
	\end{definition}
	
	
	\begin{remark} \label{rmk:R0}
		Without cache memories, i.e., $\Mw =\Ms  =0$, the secrecy capacity-memory tradeoff $C_{\s}(\Mw ,\Ms  )$ was determined in \cite{ekrem13}:
		\begin{equation}
			C_{\s}\left(\Mw =0,\Ms  =0\right)  = \left(\sum_{k=1}^{\Ka} \frac{1}{\delta_z-\delta_k}\right)^{-1}.
		\end{equation}
	\end{remark}	\bigskip

	For comparison, we will also be interested in the standard (non-secure) capacity-memory tradeoff $C \left(\Mw,\Ms \right)$ as defined in \cite{WiggerJournal,WiggerYenerJournal}. It is  the largest rate $R$ for given cache sizes $\Mw$ and $\Ms$ that allows to find caching, encoding, and decoding functions as in 
	\eqref{eq:cachingFct}, \eqref{eq:encodingFct}, and \eqref{eq:decodingFct} so that
	\begin{equation}
			\p_e^{\Worst} \leq \epsilon. 
	\end{equation}
	That means, $C \left(\Mw,\Ms \right)$ differs from $C_\s \left(\Mw,\Ms \right)$ in that no secrecy constraint is imposed. 
	
	We start with a small section on preliminaries. 
	
	\subsection{Preliminaries: A Mapping  and A Lemma  }\label{sec:preliminaries}
	
		For fixed positive real numbers $n, R', R_{\key}$, 
		define the mapping
	\begin{IEEEeqnarray}{rCl}
		\textsf{sec}(w, \textnormal{key})\colon \lfloor 2^{nR'} \rfloor \times \lfloor 2^{nR_{\textnormal{Key}}} \rfloor \;&\to &\; \lfloor 2^{nR'} \rfloor \\
		\colon  (w,\key) \;&\mapsto& \;(w+\key)\; \; \textsf{mod} \; \; \lfloor 2^{nR'} \rfloor,
	\end{IEEEeqnarray}
	where $\textsf{mod}$ denotes the modulo operator.
	Notice that from $\textsf{sec}(w,\textnormal{key})$ and $\textnormal{key}$ it is possible to recover $w$. Let $\textsf{sec}^{-1}_{\textnormal{key}}(\cdot)$ denote this inverse mapping so that 
	\begin{equation}
	\textsf{sec}^{-1}_{\textnormal{key}} \left(	\textsf{sec}(w, \textnormal{key}) \right)=w. 
	\end{equation}
	When $R'=R_\key$, we also write $\bigoplus$ instead of $\textsf{sec}$:
	\begin{equation}
	w_1 \bigoplus w_2 := \enc{w_1}{w_2}.
	\end{equation}
	We mostly use $\textsf{sec}$ when one of the arguments is a secret key and we use $\bigoplus$ when both arguments refer to messages.\bigskip
	
	The following lemma is key in the secrecy analysis of our schemes.
	\begin{lemma}\label{lem:secure}
	Consider the rates $\tilde{R} \geq R'>0$ and the random codebook 
		\begin{equation}
		\C=\Big\{x^{n'}(\ell)\colon \; \; \ell\in \big\{1,\ldots,\lfloor 2^{n'\tilde{R}}\rfloor\big\} \Big\}
		\end{equation} with entries drawn i.i.d. according to some distribution $P_X$. Let the message $W'$ be uniform over $\{1,\ldots, \lfloor 2^{nR'}\rfloor\}$. To encode the message $W'=w'$, the transmitter picks uniformly at random a codeword from a predetermined subset $\mathcal{S}(w')\subseteq \C$ and sends this codeword over the channel.  If the cardinalities of the  subsets $\{\mathcal{S}(w')\}$ satisfy
		\begin{equation}
		\log_2|\mathcal{S}(w')| \geq \min\{ R', (1-\delta_z)\}, \quad \forall w'\in\{1,\ldots, \lfloor 2^{n'R'}\rfloor\},
		\end{equation}   then 
		\begin{IEEEeqnarray}{rCl}
			\frac{1}{n'} I\big(W'; Z^{n'}|\mathcal{C}\big) \to 0 \qquad \textnormal{as} \qquad n'\to\infty.
		\end{IEEEeqnarray}
		
	\end{lemma}
	\begin{IEEEproof}
	If $	\log_2|\mathcal{S}(w')| =  R'$, then the message is encoded by means of a one-time pad \cite{shannon_cipher} and $\frac{1}{n}I(W';Z^{n'}|\mathcal{C})=0$. Otherwise, the lemma can be proved for example by following the steps in \cite[Ch. 22.2]{elGamalBook}, see p.~554 and Proof of Lemma~22.1. 
	\end{IEEEproof}
	
	\bigskip

\section{Upper Bounds on Secrecy Capacity-Memory Tradeoff} \label{sec:UB}

We start by presenting an upper bound on the secrecy capacity-memory tradeoff of a general degraded $\Ka$-user BC with arbitrary cache sizes $\Ma_1, \ldots, \Ma_{\Ka}$ at the receivers. Subsequently, we  specialize this bound to the erasure BC studied in this paper where all weak receivers and all strong receivers have equal cache sizes.
  
 
 	Consider an arbitrary degraded $\Ka$-user discrete memoryless BC (not necessarily an erasure BC) with channel transition law 
$\Gamma(y_1, \ldots, y_{\Ka}|x)$. For simplicity, and because  our result depends only on the conditional marginals $\Gamma_1(y_1|x), \ldots, \Gamma_{\Ka}(y_{\Ka}|x)$, we assume that the channel is physically degraded, so the Markov chain
	\begin{equation}
		X \to Y_{\Ka} \to Y_{{\Ka}-1}  \to \ldots \to Y_1 
	\end{equation}
	holds. In the same spirit, we also  assume that the eavesdropper is degraded with respect to some of the legitimate receivers, and all other legitimate receivers are degraded with respect to the eavesdropper.  Three scenarios can be considered:
\begin{enumerate} 
\item[a)] The eavesdropper is degraded with respect to \emph{all} legitimate receivers:
\begin{equation}
		X \to Y_{\Ka} \to Y_{{\Ka}-1} \to \ldots \to Y_1 \to Z.
\end{equation}
\item[b)] \emph{All} legitimate receivers are degraded with respect to the eavesdropper:
\begin{equation}
		X \to Z\to Y_{\Ka} \to Y_{{\Ka}-1} \to \ldots \to Y_1.
\end{equation}
\item[c)] The eavesdropper is degraded with respect to the strongest $\Ka-\ell^*$   legitimate receivers, for some $\ell^*\in\{1,\ldots, \Ka -1\}$, and the remaining legitimate receivers are degraded with respect to the eavesdropper:
	\begin{equation}
		X \to Y_{\Ka} \to Y_{{\Ka}-1} \to \ldots \to Y_{\ell^*+1} \to Z \to Y_{\ell^*}  \to \ldots \to Y_1 
	\end{equation}
\end{enumerate}
Let each Receiver~$k\in\{1,\ldots, {\Ka}\}$ have cache  size $\Ma_k$. The following lemma holds.

\bigskip
	\begin{lemma}[Upper Bound for Arbitrary Degraded BCs and Cache Sizes] \label{lemma:1}
	 If a rate-memory tuple $(R, \Ma_1, \ldots, \Ma_{\Ka})$ is securely achievable, then for each receiver set $\mathcal{S} := \left\lbrace j_1,\ldots,j_{|\mathcal{S}|} \right\rbrace \subseteq \K$, there exist auxiliaries $(U_{1}, U_{2}, \ldots, U_{{|\mathcal{S}|}},Q)$ so that for each realization of $Q=q$ the following Markov chain holds:
	\begin{equation}
		U_{1} \to U_{2} \to \ldots \to U_{{|\mathcal{S}|}} \to X \to ( Y_{j_1}, \ldots, Y_{j_{|\mathcal{S}|}}, Z) ;
	\end{equation}
	and the following $|\mathcal{S}|$ inequalities are satisfied:
 	\begin{subequations} \label{eq:q1}
		\begin{equation}	 
 			R \leq \big[ I(U_{1};Y_{j_1}|Q) - I(U_{1};Z|Q)\big]^{+} + \Ma_{j_1},		
		\end{equation}	
	 and 
  		\begin{equation} \label{eq:twob}
  			k R \leq  \sum_{\ell=1}^k \big[ I(U_{\ell};Y_{j_\ell}|U_{\ell-1},Q) - I(U_{\ell};Z|U_{{\ell-1}},Q) \big]^+ + \sum_{\ell=1}^k \Ma_{j_\ell}, \qquad k\in\{2,\ldots, |\mathcal{S}|\},   
  		\end{equation}
   \end{subequations}
   where $(\cdot)^+: = \max \{0,\cdot\}$.
   \end{lemma}
   \begin{proof}
   See Appendix~\ref{app:proofUB}. 
   \end{proof}
   \bigskip

Turn back to the setup with weak and strong receivers in Figure~\ref{fig:systemModel}. Based on the previous lemma and  the upper bound on the standard (non-secure) capacity-memory tradeoff in \cite[Theorem~5]{ WiggerYenerJournal}, the following Theorem~\ref{thm:UB_assign} presents two upper bounds (Inequalities \eqref{eq:UB1_assign} and \eqref{eq:UB2_assign}) on the secrecy capacity-memory tradeoff for each choice of $k_w \in  \{0,1,\ldots,\Kw\}$ and $k_s \in  \{0,1,\ldots,\Ks \}$. Depending on the cache sizes $\Mw$ and $\Ms$, a different  choice of the parameters and of the bounds \eqref{eq:UB1_assign} or \eqref{eq:UB2_assign}   is  tightest. 

	\begin{theorem}[Upper Bound on $C_{\s}(\Mw ,\Ms  )$]\label{thm:UB_assign}
	For each choice of $k_w \in  \{0,1,\ldots,\Kw\}$ and $k_s \in  \{0,1,\ldots,\Ks \}$, the  secrecy capacity-memory tradeoff is upper bounded in the following two ways:\footnote{For any finite numbers $a,b$, we  define $\min\{a/0, b\} =b$ and $\min\{a/0, b/0\} =\infty$. In the minimization \eqref{eq:UB2_assign}, a minimum over an empty set is defined as $+\infty$.}
		\begin{subequations}
			\begin{enumerate} 
				\item 
		\begin{align} \label{eq:UB1_assign}	 
			C_{\s}(\Mw ,\Ms  ) \leq  \max_{\beta \in [0,1]} \min & \left\lbrace \frac{\beta(\delta_z - \delta_w)^+}{k_w}  + \Mw ,   \frac{\beta (\delta_z - \delta_w)^+ + (1-\beta)(\delta_z-\delta_s)^+}{k_w+k_s} + \frac{k_w \Mw  + k_s \Ms  }{k_w+k_s}   \right\rbrace ;
		\end{align}	
		and
\item 
		\begin{equation}\label{eq:UB2_assign}
			C_{\s}(\Mw ,\Ms  ) \leq\min\left\{ \min\limits_{i = 1,\ldots,k_w} \left\lbrace (1-\delta_w)\beta_i  +\alpha_i \right\rbrace,\;    \min\limits_{j = 1,\ldots,k_s} \left\lbrace (1-\delta_s)\beta_{k_w+j}  +\alpha_{k_w+j} \right\rbrace \right\} ,
		\end{equation}
	for some tuple of nonnegative real numbers $\beta_1, \ldots,\beta_{k_w+k_s} \geq 0$ summing to  $1$,
		  where
		\begin{align*}
			\alpha_i & := \min \left\lbrace \frac{i\Mw }{\Da-i+1}, \frac{1}{k_w+k_s-i+1} \left( \frac{(k_w+k_s)(k_w\Mw +k_s\Ms  )}{\Da} - \sum\limits_{\ell=1}^{i-1}\alpha_\ell \right) \right\rbrace, \qquad\qquad i = 1,\ldots,k_w, \\
			\alpha_{k_w+j} & := \min \left\lbrace \frac{k_w\Mw +j\Ms  }{\Da-k_w-j+1}, \frac{1}{k_s-j+1} \left( \frac{(k_w+k_s)(k_w\Mw +k_s\Ms  )}{\Da} - \sum\limits_{\ell=1}^{k_w+j-1}\alpha_\ell \right) \right\rbrace, \qquad j = 1,\ldots,k_s.
		\end{align*}
	\end{enumerate}	\end{subequations}
	\end{theorem}
	\begin{proof}
		 The first upper bound in~\eqref{eq:UB1_assign} is obtained by specializing Lemma~\ref{lemma:1} to the erasure BC  in Figure~\ref{fig:systemModel}. More specifically,  setting $\beta_j:=I(U_j;X|U_{j-1},Q)$, constraints \eqref{eq:q1} can be rewritten as:
  	\begin{subequations} \label{eq:q2}
  		\begin{align}
  				C_{\s}(\Mw ,\Ms  )& \leq \beta_1 (\delta_z - \delta_{j_1})^+ + \Ma_{j_1}, \\
 					C_{\s}(\Mw ,\Ms  ) & \leq \frac{1}{k} \sum_{\ell=1}^k \big[ \beta_\ell (\delta_z-\delta_{j_\ell})^+ + \Ma_{j_\ell} \big], \qquad \forall k~\in\{1,\ldots, |\mathcal{S}|\}.
		\end{align} 
	\end{subequations}  
	By well known properties on the mutual information, one finds that $\beta_1,\dots, \beta_{|\mathcal{S}|} \geq 0$ and 
	\begin{equation}
	\sum_{k=1}^{|\mathcal{S}|} \beta_k = I(U_1,\ldots,U_{|\mathcal{S}|};X|Q) \leq H(X) \leq  1.
	\end{equation}
Upper bound~\eqref{eq:UB1_assign} is now obtained by specializing \eqref{eq:q2} to one of the subsets
\begin{subequations} 
\begin{equation} \label{eq:set1}
\mathcal{S}=\{1,\ldots, k_w, \Kw+1,\ldots, \Kw +k_s\}, \qquad k_w,k_s >0,
\end{equation} 
or
\begin{equation} \label{eq:set2}
\mathcal{S}=\{1,\ldots, k_w\},
\end{equation}
or
\begin{equation} \label{eq:set3}
\mathcal{S}=\{\Kw+1,\ldots, \Kw +k_s\},
\end{equation}
\end{subequations}
and by noticing that for the subset in \eqref{eq:set1} one can  restrict to
\begin{equation}
\beta_{1} = \beta_2=\ldots = \beta_{k_w} = \frac{\beta}{k_w}
\end{equation}and 
\begin{equation}
\beta_{\Kw+1} = \beta_{\Kw+2}=\ldots = \beta_{\Kw+k_s} = \frac{1-\beta}{k_s},
\end{equation}
for some $\beta \in[0,1]$.
For the subset in \eqref{eq:set2} one can  restrict to
\begin{equation}
\beta_{1} = \beta_2=\ldots = \beta_{k_w} = \frac{1}{k_w},
\end{equation}and the subset in \eqref{eq:set3} one can  restrict to
\begin{equation}
\beta_{\Kw+1} = \beta_{\Kw+2}=\ldots = \beta_{\Kw+k_s} = \frac{1}{k_s}.
\end{equation}

 	Constraint~\eqref{eq:UB2_assign} follows by ignoring the secrecy constraint and specializing \cite[Theorem~5]{ WiggerYenerJournal} to the erasure BC with weak and strong receivers considered in this paper. 
	\end{proof}
	\bigskip

	We simplify  this upper bound for the special cases where only weak receivers have cache memories. More specifically, we replace the upper bound in~\eqref{eq:UB2_assign}, which is obtained from \cite[Theorem~5]{ WiggerYenerJournal},  by a simpler bound that is obtained by specializing the weaker upper bound in \cite{WiggerJournal}.
	\begin{corollary}[Upper Bound on $C_{\s}(\Mw ,\Ms  =0)$]\label{cor:UB}
For each choice of $k_w \in\{0,1,\ldots, \Kw\}$, the secrecy capacity-memory tradeoff $C_{\s}(\Mw ,\Ms  =0)$ 
	is upper bounded in the following two ways:
	\begin{subequations}
			\begin{IEEEeqnarray}{rCl}
			C_{\s} (\Mw ,\Ms  =0)&\leq& \left(\frac{k_w}{1-\delta_w}+ \frac{\Ks }{1-\delta_s}\right)^{-1}\hspace{-2mm}+  \frac{k_w \Mw }{\Da}, \vspace{2mm} \label{eq:UB2}
			\end{IEEEeqnarray}
			and
			\begin{IEEEeqnarray}{rCl}
			C_{\s} (\Mw ,\Ms  =0)&\leq &\max_{\beta\in[0,1]} \min\Bigg\{ \frac{\beta(\delta_z-\delta_w)^+}{k_w}+ \Mw , \frac{\beta(\delta_z-\delta_w)^++(1-\beta)(\delta_z-\delta_s)^+}{k_w+\Ks } + \frac{k_w}{k_w+\Ks } \Mw  \Bigg\}. \label{eq:UB3}\IEEEeqnarraynumspace
			\end{IEEEeqnarray}
	\end{subequations}
	\end{corollary}
	\begin{proof}
	Constraint~\eqref{eq:UB2} follows from \cite{WiggerJournal} and by ignoring the secrecy constraint.
	Constraint \eqref{eq:UB3} is obtained by specializing \eqref{eq:UB1_assign} 
	to the case when $\Ms  =0$. We notice that  in this case the constraint that  \eqref{eq:UB1_assign} 
	generates for $j=\Ks $ is tighter than any constraint that it generates  for $j < \Ks $. Thus, we can remove all constraints for $j < \Ks $ without affecting the result and we retain only the  constraint in \eqref{eq:UB3}.
\end{proof}	\bigskip

\section{Results when Only Weak Receivers have Cache Memories}\label{sec:1sided}

	Consider the special case where only weak receivers have cache memories, i.e.
	\begin{equation}
		\Ms   = 0.
	\end{equation}		
	 In this case, a positive secrecy rate can only be achieved if
	 \begin{equation} \label{eq:dz}
	 	\delta_z < \delta_s.
	 \end{equation}
		In the remainder of this section we  assume that \eqref{eq:dz} holds.

\bigskip

\subsection{Coding Schemes}

 We present four coding schemes in the order of increasing cache requirements.  In the first two schemes, only random keys are placed in the cache memories.  The third and fourth schemes also place  parts of the messages in the cache memories and apply  joint cache-channel coding for the delivery communication  where the decoding operations at the receivers  adapt at the same time to the channel statistics and the cache contents. For simplicity, and because on an erasure BC time-sharing is optimal to send independent messages to the various receivers, in some of our schemes communication is divided into subphases.  When applied to general discrete memoryless BCs, the schemes can be improved by superposing various subphases on each other. \bigskip

\subsubsection{Wiretap and Cached Keys} ~~ \label{sec:wiretap}

\underline{\textit{Placement phase:}} Store an independent secret  key $K_i$ in Receiver~$i$'s cache memory, for $i\in\Kw$.

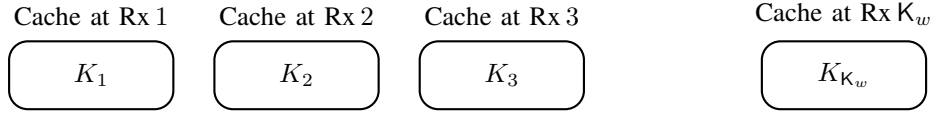
\begin{figure}[H]
	\centering
	\begin{tikzpicture}[scale=0.9]
	\node[above] at (1.2,1.1) {Cache at Rx\! $1$};
	\draw[rounded corners=7pt,thick] (0,0) rectangle (2.4,1);
	\node at (1.2,0.5) {$K_1$};
	
		\node[above] at (4.2,1.1) {Cache at Rx\! $2$};
	\draw[rounded corners=7pt,thick] (3,0) rectangle (5.4,1);
	\node at (4.2,0.5) {$K_2$};
	
		\node[above] at (7.2,1.1) {Cache at Rx\! $3$};
	\draw[rounded corners=7pt,thick] (6,0) rectangle (8.4,1);
	\node at (7.2,0.5) {$K_3$};

		\node[above] at (12.2,1.1) {Cache at Rx\! $\Kw$};
	\draw[rounded corners=7pt,thick] (11,0) rectangle (13.4,1);
	\node at (12.2,0.5) {$K_{\Kw}$};
	\end{tikzpicture}
\end{figure}

\underline{\textit{Delivery phase:}}  Time-sharing is applied over two subphases, where transmission in the first subphase is to all the weak receivers and transmission in the second subphase is to all strong receivers. In Subphase~1, the transmitter uses a standard (non-secure) broadcast code to send the  secured message tuple
\begin{equation}
\textbf{W}_{\sec}:= \Big(\enc{W_{d_1}}{K_1}, \ \enc{W_{d_2}}{K_2},\ \enc{W_{d_3}}{K_3},\  \ldots,\ \enc{W_{d_{\Kw}}}{K_{\Kw}}\Big),
\end{equation}
to weak receivers~$1,\ldots, \Kw$, respectively. With the secret key $K_i$ stored in its cache memory, each weak receiver~$i\in\K_w$ can then recover a guess of its desired message $W_{d_i}$.  In Subphase~2, the transmitter uses a wiretap broadcast code \cite{ekrem13}  to send messages $W_{d_{\Kw+1}},\ldots, W_{d_{\Ka}}$ to the strong receivers $\Kw+1,\ldots, \Ka$, respectively. 

For a detailed analysis, see Section~\ref{sec:wiretap_analysis}.
\bigskip

\subsubsection{Cache-Aided Superposition Jamming}\label{sec:superposition}~~

\underline{\textit{Placement phase:}} As in the previous subsection, store an independent secret  key $K_i$ in Receiver~$i$'s cache memory, for $i\in\Kw$.

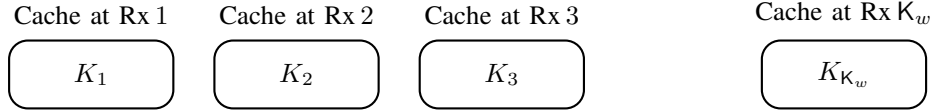
\begin{figure}[H]
	\centering
	\begin{tikzpicture}[scale=0.9]
	\node[above] at (1.2,1.1) {Cache at Rx\! $1$};
	\draw[rounded corners=7pt,thick] (0,0) rectangle (2.4,1);
	\node at (1.2,0.5) {$K_1$};
	
		\node[above] at (4.2,1.1) {Cache at Rx\! $2$};
	\draw[rounded corners=7pt,thick] (3,0) rectangle (5.4,1);
	\node at (4.2,0.5) {$K_2$};
	
		\node[above] at (7.2,1.1) {Cache at Rx\! $3$};
	\draw[rounded corners=7pt,thick] (6,0) rectangle (8.4,1);
	\node at (7.2,0.5) {$K_3$};

		\node[above] at (12.2,1.1) {Cache at Rx\! $\Kw$};
	\draw[rounded corners=7pt,thick] (11,0) rectangle (13.4,1);
	\node at (12.2,0.5) {$K_{\Kw}$};
 	\end{tikzpicture}
\end{figure}
\underline{\textit{Delivery phase:}} The transmitter uses a superposition code to send the  secured message tuple
\begin{equation}\label{eq:cloudc}
\mathbf{W}_{\sec}:=\big(\enc{W_{d_1}}{K_1}, \;\ \enc{W_{d_2}}{K_2}, \;\ \enc{W_{d_3}}{K_3},\  \ldots,\ \enc{W_{d_{\Kw}}}{K_{\Kw}}\big),
\end{equation}
in the cloud center and the non-secure message tuple
\begin{equation}
\mathbf{W}_{\textnormal{sat}}:=\big(W_{d_{\Kw+1}},\ldots, W_{d_{\Ka}}\big)
\end{equation}
in the satellite. If the rate of the cloud messages \eqref{eq:cloudc} does not suffice to secure the satellite message, random binning is added to the satellite codebook. See Figure~\ref{fig:superposition_one}. 
\input{superpos_picture_oneshot.tex}
Weak receivers decode  only the cloud center and strong receivers the cloud center and the satellite codeword. From this decoding operation, each strong receiver $j\in\Ks$ directly obtains a guess of $W_{d_j}$. Each weak receiver~$i\in\Kw$  uses the secret key $K_i$ stored in its cache memory to recover a guess of its intended message $W_{d_i}$.  Section~\ref{sec:superposition_analysis} presents the details of the scheme and its analysis. 
\bigskip

\subsubsection{Secure Cache-Aided Piggyback Coding I}~~\label{sec:piggyback}

 The scheme builds on the nested piggyback coding scheme in \cite{WiggerJournal}, which is rendered secure by applying secret keys to the produced XOR-messages and by introducing wiretap binning. During the placement phase, each of these secret keys is stored in the cache memories of the weak receivers that decode the corresponding XOR-message.  
We outline  the scheme for $\Kw=3$ weak receivers and $\Ks=1$ strong receiver. 

Divide each message $W_d$ into six submessages 
\begin{equation}
W_d = \big(W_{d,\{1\}}^{(A)},W_{d,\{2\}}^{(A)},W_{d,\{3\}}^{(A)},W_{d,\{1,2\}}^{(B)},W_{d,\{1,3\}}^{(B)},W_{d,\{2,3\}}^{(B)}\big), \quad d \in \D,
\end{equation}
where the first three are of equal rate and the latter three are of equal rate. Let  $K_{\{1,2,3\}}$, $K_{\{1,2\}},K_{\{2,3\}}, K_{\{1,3\}}$ be independent secret keys generated at the transmitter.

\underline{\textit{Placement phase:}} Placement is as described in the following table:
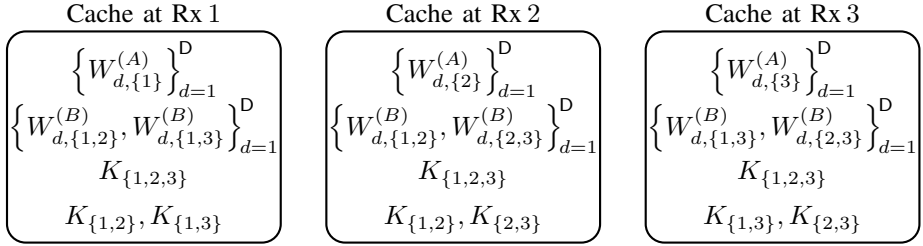
\begin{figure}[H] \centering
	\begin{tikzpicture}	
	\node[above] at (1.4,3.5) {Cache at Rx\! $1$};
	\draw[rounded corners=7pt,thick] (-0.4,0.7) rectangle (3.2,3.5);
	\node at (1.4,3) {$\left\lbrace W_{d,\{1\}}^{(A)}\right\rbrace_{\!d=1}^{\!\Da}$};	
	\node at (1.4,2.25) {$\left\lbrace W_{d,\{1,2\}}^{(B)}, W_{d,\{1,3\}}^{(B)}\right\rbrace_{\!d=1}^{\!\Da}$};
	\node at (1.4,1.6) {$K_{\{1,2,3\}}$};	
	\node at (1.4,1) {$K_{\{1,2\}},K_{\{1,3\}}$};

	\node[above] at (5.6,3.5) {Cache at Rx\! $2$};
	\draw[rounded corners=7pt,thick] (3.8,0.7) rectangle (7.4,3.5);
	\node at (5.6,3) {$\left\lbrace W_{d,\{2\}}^{(A)}\right\rbrace_{\!d=1}^{\!\Da}$};	
	\node at (5.6,2.25) {$\left\lbrace W_{d,\{1,2\}}^{(B)}, W_{d,\{2,3\}}^{(B)}\right\rbrace_{\!d=1}^{\!\Da}$};
	\node at (5.6,1.6) {$K_{\{1,2,3\}}$};	
	\node at (5.6,1) {$K_{\{1,2\}},K_{\{2,3\}}$};

	\node[above] at (9.8,3.5) {Cache at Rx\! $3$};
	\draw[rounded corners=7pt,thick] (8,0.7) rectangle (11.6,3.5);
	\node at (9.8,3) {$\left\lbrace W_{d,\{3\}}^{(A)}\right\rbrace_{\!d=1}^{\!\Da}$};	
	\node at (9.8,2.25) {$\left\lbrace W_{d,\{1,3\}}^{(B)}, W_{d,\{2,3\}}^{(B)}\right\rbrace_{\!d=1}^{\!\Da}$};
	\node at (9.8,1.6) {$K_{\{1,2,3\}}$};	
	\node at (9.8,1) {$K_{\{1,3\}},K_{\{2,3\}}$};

		\end{tikzpicture} 
\end{figure}

\underline{\textit{Delivery  phase:}} Time-sharing is applied over three subphases and Subphase 2 is further divided into 3 periods. 


In Subphase~1, the secured message 
\begin{equation}
\enc{W_{d_1,\{2,3\}}^{(B)} \oplus W_{d_2,\{1,3\}}^{(B)} \oplus W_{d_3,\{1,2\}}^{(B)}}{\ K_{\{1,2,3\}}}
\end{equation} is sent to all three weak receivers using a standard point-to-point code. With their cache contents, each weak receiver~$i$ can decode the submessage of $W_{d_i}$ sent in this subphase. Communication is secured when the  key $K_{\{1,2,3\}}$ is sufficiently long.
In Subphase~3, the non-secure message $W_{d_4}^{(A)}$ is sent to the strong receiver~$4$ using a standard wiretap code. 

In the first period of Subphase~2, the transmitter uses the secure piggyback codebook in Figure~\ref{fig:piggy} to transmit the secure message
\begin{equation}
\mathbf{W}_{\sec,\{1,2\}}^{(A)}=\enc{W_{d_1,\{2\}}^{(A)} \oplus W_{d_2,\{1\}}^{(A)}}{\ K_{\{1,2\}}}\qquad 
\end{equation}
  to Receivers~$1$ and $2$ and
 the non-secure message $W_{d_4,\{1,2\}}^{(B)}$ to Receiver~$4$. It randomly chooses a codeword in the wiretap bin indicated by $\mathbf{W}_{\sec,\{1,2\}}$ and  $W_{d_4,\{1,2\}}^{(B)}$ and sends the chosen codeword over the channel. 
 Weak receivers~$1$ and $2$ have stored $W_{d_4,\{1,2\}}^{(B)}$  in their cache memories and can decode based on a restricted codebook consisting only of the bins in the  column indicated by  $W_{d_4,\{1,2\}}^{(B)}$. Their decoding performance is thus the same as if this message $W_{d_4,\{1,2\}}^{(B)}$ had not been sent at all. The strong receiver~$4$ has no  cache memory and decodes both messages based on the entire codebook. Notice that the secured message $\mathbf{W}_{\sec,\{1,2\}}$ also acts as random wiretap binning to secure message $W_{d_4,\{1,2\}}^{(B)}$ to Receiver~$4$. If this binning suffices to secure $W_{d_4,\{1,2\}}^{(B)}$, then no additional random binning is needed, i.e.,  the magenta bin in Figure~\ref{fig:piggy} can be chosen of size 1.
 
Similar secure piggyback codebooks are also used during the second  and third periods of Subphase 2 to send messages $\enc{W_{d_1,\{3\}}^{(A)} \oplus W_{d_3,\{1\}}^{(A)}}{\ K_{\{1,3\}}}$ and $W_{d_4,\{1,3\}}$ and messages  $\enc{W_{d_2,\{3\}}^{(A)} \oplus W_{d_3,\{2\}}^{(A)}}{\ K_{\{2,3\}}}$ and $W_{d_4,\{2,3\}}$, respectively. 

At the end of the delivery phase, each Receiver $k\in\K$ assembles all the guesses pertaining to its desired message $W_{d_k}$  and (in case of the weak receivers) all the parts of this message stored in its cache memory to form the final guess $\hat{W}_{k}$.

\input{secure_piggyback}
	\bigskip
	
%
%
%
%
%

	\begin{remark}
		The secure piggyback codebook used in Subphase~2 is inspired by the non-secure piggyback coding  and Tuncel coding in \cite{WiggerJournal} and \cite{tuncel06}, and by the secure coding scheme for BCs with complementary side-information  in \cite{schaefer}. In fact, the main difference to the scheme in \cite{schaefer} is that here one of the receivers and the transmitter share a common secret key, 
which allows to reduce the size of the wiretap bins or even eliminate them completely. 
		
		Interestingly, in this  construction,  the secret key $K_{\{1,2\}}, K_{\{1,3\}}, K_{\{2,3\}}$ stored at the weak receivers allow to ``remotely secure" the transmission from the transmitter to the strong receiver~4.
	\end{remark}

	\bigskip
	\subsubsection{Secure Cache-Aided Piggyback Coding II} \label{sec:piggybackII} ~
	
	This scheme is  similar to the scheme in the previous section, but simpler. Divide  each message $W_d$ into 2 submessages $W_{d}=(W_d^{(A)}, W_d^{(B)})$ and let $K_1,\ldots, K_{\Kw}$ be independent secret keys. 
	
	\bigskip
	\underline{\textit{Placement phase:}} Placement is as depicted in the following. In particular, each weak receiver $i\in \Kw$ caches the secret key $K_i$ and  all submessages $W_{d}^{(B)}$, for $d\in\D$.
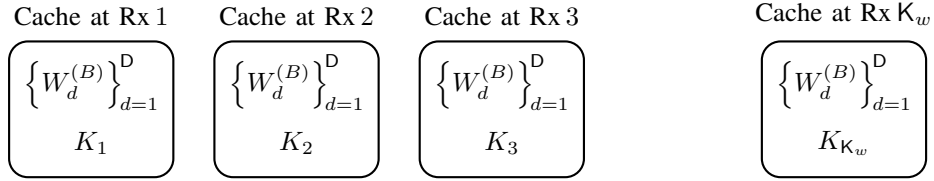
	\begin{figure}[H]
	\centering
	\begin{tikzpicture}[scale=0.9]
	\node[above] at (1.2,1.1) {Cache at Rx\! $1$};
	\draw[rounded corners=7pt,thick] (0,-1) rectangle (2.4,1);
		\node at (1.2, 0.4) {$\left\lbrace W_{d}^{(B)}\right\rbrace_{\!d=1}^{\!\Da}$};	
	\node at (1.2,-.5) {$K_1$};
	
		\node[above] at (4.2,1.1) {Cache at Rx\! $2$};
	\draw[rounded corners=7pt,thick] (3,-1) rectangle (5.4,1);
			\node at (4.2, 0.4) {$\left\lbrace W_{d}^{(B)}\right\rbrace_{\!d=1}^{\!\Da}$};	
	\node at (4.2,-0.5) {$K_2$};
	
		\node[above] at (7.2,1.1) {Cache at Rx\! $3$};
	\draw[rounded corners=7pt,thick] (6,-1) rectangle (8.4,1);
			\node at (7.2, 0.4) {$\left\lbrace W_{d}^{(B)}\right\rbrace_{\!d=1}^{\!\Da}$};	
	\node at (7.2,-0.5) {$K_3$};

		\node[above] at (12.2,1.1) {Cache at Rx\! $\Kw$};
	\draw[rounded corners=7pt,thick] (11,-1) rectangle (13.4,1);
			\node at (12.2, 0.4) {$\left\lbrace W_{d}^{(B)}\right\rbrace_{\!d=1}^{\!\Da}$};	

	\node at (12.2,-0.5) {$K_{\Kw}$};
 	\end{tikzpicture}
\end{figure}
\bigskip

\underline{\textit{Delivery phase:}} Transmission is in two subphases. In Subphase~1, the secure piggyback codebook is used to send the secured message tuple
\begin{equation}
\mathbf{W}_{\sec, w}^{(A)} := \big( \enc{W_{d_1}^{(A)}}{ K_1}, \ \ldots, \ \enc{W_{d_{\Kw}}^{(A)}}{ K_{\Kw}}\big)
\end{equation} 
to all weak receivers and the non-secure message tuple 
\begin{equation}
\mathbf{W}_{ s}^{(B)} := \big(W_{d_{\Kw+1}}^{(B)}, \ \ldots, \ W_{d_{\Ka}}^{(B)}\big)
\end{equation} 
to the strong receivers. The codebook is depicted in Figure~\ref{fig:piggy} where $\mathbf{W}_{\sec,\{1,2\}}^{(A)}$ needs to be replaced by $\mathbf{W}_{\sec, w}^{(A)}$ and $W_{d_4,\{1,2\}}^{(B)}$ by $\mathbf{W}_{ s}^{(B)}$. The weak receivers can reconstruct $\mathbf{W}_{ s}^{(B)}$ from their cache contents, and thus decode their desired message tuple $\mathbf{W}_{\sec, w}^{(A)}$ based on the single column of the codebook indicated by $\mathbf{W}_{ s}^{(B)}$. From this decoded tuple and the secret key $K_i$ stored in its cache memory, each weak receiver~$i\in\K_w$ can then produce a guess of its desired message part $W_{d_i}^{(A)}$. The strong receivers decode both message tuples $\mathbf{W}_{\sec, w}^{(A)} $ and $\mathbf{W}_{ s}^{(B)}$.  Strong receiver~$j\in\K_s$ keeps only its guess of $W_{d_j}^{(B)}$ and discards the rest.

Communication in this first subphase is secured because messages $W_{d_1}^{(A)}, \ldots, W_{d_{\Kw}}^{(A)}$ are perfectly secured by one-time pads and these one-time pads act as random bin indices to protect the messages $W_{d_{\Kw+1}}^{(B)}, \ldots, W_{d_{\Ka}}^{(B)}$ as in wiretap coding. 

In Subphase 2, the message tuple 
\begin{IEEEeqnarray}{rCl}
\mathbf{W}_{s}^{(A)} =  \big(W_{d_{\Kw+1}}^{(A)}, \ \ldots, \ W_{d_{\Ka}}^{(A)}\big)
\end{IEEEeqnarray}
is sent to all the strong receivers using a point-to-point wiretap code. 

The choice of the rates and the lengths of the subphases are explained in Section~\ref{sec:piggyback_analysisII}, where the scheme is also analyzed.


\bigskip
\subsection{Results on the Secrecy Capacity-Memory Tradeoff}
Consider the following $\Kw+4$ rate-memory pairs: \vspace{3mm}
	\begin{subequations}\label{eq:RM}
	\begin{itemize}
	
	\item \ \vspace{-7mm}
	\begin{IEEEeqnarray}{rCl}
		R^{(0)} &:=&  \frac{(\delta_z-\delta_s) (\delta_z-\delta_w)^{+}}{\Kw(\delta_z-\delta_s)+\Ks (\delta_z-\delta_w)^{+}} \label{eq:R0} , \hspace{11.27cm} \\[2.ex]
		\Ma^{(0)} &:=& 0; \\ \nonumber
	\end{IEEEeqnarray}

	\item \	\vspace{-7mm}
	\begin{IEEEeqnarray}{rCl}
		R^{(1)} &:=& \frac{(1-\delta_w)(\delta_z-\delta_s)}{\Ks (1-\delta_w)+\Kw(\delta_z-\delta_s)}, \label{eq:R1}\hspace{11.65cm}\\[1.2ex]
		\Ma^{(1)} &:=& \frac{(\delta_z-\delta_s) \min \big\{1-\delta_z,1-\delta_w \big\}}{\Ks (1-\delta_w)+\Kw(\delta_z-\delta_s)}; \label{eq:M1}\\ \nonumber
	\end{IEEEeqnarray}
	
	\item \	\vspace{-7mm}
	\begin{IEEEeqnarray}{rCl}
		R^{(2)} &:=& \min \left\lbrace \frac{(1-\delta_w)(1-\delta_s)}{\Ks (1-\delta_w)+\Kw(1-\delta_s)}, \frac{(1-\delta_w)(\delta_z-\delta_s)}{\Ks (1-\delta_w)+\Kw(\delta_w-\delta_s)} \right\rbrace, \hspace{6.37cm} \label{eq:R2}\\[1.2ex]
		\Ma^{(2)} &:=& \min \left\lbrace \frac{1-\delta_z}{\Kw},\frac{(1-\delta_w)(\delta_z-\delta_s)}{\Ks (1-\delta_w)+\Kw(\delta_w-\delta_s)} \right\rbrace; \label{eq:M2}\\ \nonumber
	\end{IEEEeqnarray}
	
	\item \quad For $t \in \{1,\ldots,\Kw-1\}$
	\begin{IEEEeqnarray}{rCl}
		R^{(t+2)} &:=& \frac{(t+1)(1-\delta_w)(\delta_z-\delta_s)\big[\Ks t(1-\delta_w)+(\Kw-t+1)\min\left\lbrace \delta_w-\delta_s,\delta_z-\delta_s \right\rbrace \big]}{(\Kw\!-t+1)(\delta_z-\delta_s) \big[\Ks (t+1)(1-\delta_w)\!+\!(\Kw\!-t)\min \left\lbrace \delta_w-\delta_s,\delta_z-\delta_s \right\rbrace \big]\! + \Ks ^2t(t+1)(1-\delta_w)^2},\label{eq:R3}\\[1.2ex]
		\Ma^{(t+2)} &:=&  \frac{\Da\cdot t(t+1)(1-\delta_w)(\delta_z-\delta_s)\big[\Ks (t-1)(1-\delta_w)+(\Kw-t+1) \min\left\lbrace \delta_w-\delta_s,\delta_z-\delta_s \right\rbrace \big]}{\Kw\big[(\Kw\!-t+1)(\delta_z-\delta_s) [\Ks (t+1)(1-\delta_w)\!+\!(\Kw\!-t)\min\left\lbrace \delta_w-\delta_s,\delta_z-\delta_s \right\rbrace ]\! + \Ks ^2t(t+1)(1-\delta_w)^2\big]} ~~  \nonumber \\ 
			+ \;\;\;&& \!\!\!\!\!\!\!\!\!\!\!\!\!  \frac{(t+1)(\Kw-t+1)(\delta_z-\delta_s)\min \left\lbrace 1-\delta_z,1-\delta_w \right\rbrace \big[ \Ks t(1-\delta_w) + (\Kw-t)\min\left\lbrace \delta_w-\delta_s,\delta_z-\delta_s \right\rbrace \big]  }{\Kw\big[(\Kw\!-t+1)(\delta_z-\delta_s) [\Ks (t+1)(1-\delta_w)\!+\!(\Kw\!-t)\min\left\lbrace \delta_w-\delta_s,\delta_z-\delta_s \right\rbrace ]\! + \Ks ^2t(t+1)(1-\delta_w)^2\big] };  \label{eq:M3}\\[2ex] \nonumber
	\end{IEEEeqnarray}
	
	\item \	\vspace{-7mm}
	\begin{IEEEeqnarray}{rCl}
		R^{(\Kw+2)} &:=&  \frac{\delta_z-\delta_s}{\Ks }, \label{eq:R4}\\[1.2ex]
		\Ma^{(\Kw+2)} &:=& \frac{\Da \cdot \Kw(\delta_z-\delta_s)^2 + \Ks (\delta_z-\delta_s)\min \left\lbrace 1-\delta_z,1-\delta_w \right\rbrace}{\Ks \big[\Ks \min\{ 1-\delta_z,1-\delta_w\}+\Kw(\delta_z-\delta_s)\big]}; \label{eq:M4} \hspace{7.2cm}\\ \nonumber
	\end{IEEEeqnarray}
	
	\item \	\vspace{-7mm}
	\begin{IEEEeqnarray}{rCl}
		R^{(\Kw+3)} &:=& \frac{\delta_z-\delta_s}{\Ks },\label{eq:RKw3} \\[1.2ex]
		\Ma^{(\Kw+3)} &:=& \frac{\Da \cdot(\delta_z-\delta_s)}{\Ks }. \label{eq:MKw3}\hspace{13.1cm}\\ \nonumber
	\end{IEEEeqnarray}
	
	\end{itemize}
	\end{subequations}	\bigskip

	\begin{theorem}[Lower Bound on $C_{\s}(\Mw ,\Ms  =0)$]\label{thm:LB}
		\begin{align}
			C_{\s}\left(\Mw ,\Ms  =0\right)\geq \textnormal{upper hull} & \left\lbrace \left(R^{(\ell)},\Ma^{(\ell)}\right) \colon \quad \ell\in\{0,\dots,\Kw+3\} \right\rbrace.
		\end{align}
	\end{theorem}
	\begin{proof}
	
	It suffices to prove  achievability of the $\Kw+4$ rate-memory pairs $\{(R^{(\ell)}, \Ma^{(\ell)}) \colon \; \ell=0, \ldots, \Kw+3\}$. Achievability of the upper convex hull follows by time/memory sharing arguments as in \cite{MaddahAli14}. 
	 The pair $(R^{(0)},\Ma^{(0)})$ is achievable by Remark~\ref{rmk:R0}. The  pair $(R^{(1)},\Ma^{(1)})$ is achieved by the ``wiretap and cached keys" scheme described and analyzed in Sections~\ref{sec:wiretap} and \ref{sec:wiretap_analysis}. The  pair $(R^{(2)},\Ma^{(2)})$ is achieved by the ``cache-aided superposition jamming" scheme described and analyzed in Sections~\ref{sec:superposition} and \ref{sec:superposition_analysis}. The  pairs $(R^{(t+2)},\Ma^{(t+2)})$, for $t\in\{1,\ldots, \Kw-1\}$, are achieved by the ``secure cache-aided piggyback coding I" scheme described and analyzed in Sections~\ref{sec:piggyback} and \ref{sec:piggyback_analysis}. The  pair $(R^{(\Kw+2)},\Ma^{(\Kw+2)})$ is achieved by the ``secure cache-aided piggyback coding II" scheme described and analyzed in Sections~\ref{sec:piggybackII} and \ref{sec:piggyback_analysisII}. The pair $(R^{(\Kw+3)},\Ma^{(\Kw+3)})$ is achieved by  storing the entire library in the cache memory of each weak receiver and by applying a standard wiretap BC code \cite{ekrem13} to send the requested messages to the strong receivers. 	\end{proof}
	
	Interestingly, upper and lower bounds in Corollary~\ref{cor:UB} and Theorem~\ref{thm:LB} match for small and large $\Mw $ irrespective of the number of weak and strong receivers $\Kw$ {and $\Ks $}. In the absence of a secrecy constraint, the best upper and lower bounds for small $\Mw $ match only when $\Kw=1$, irrespective of the value of $\Ks $ \cite{WiggerJournal, Gunduz}.
	
	\begin{corollary} \label{cor:1}
	 	When the cache memory $\Mw $ is small:
	 	\begin{equation} \label{eq:cor1}
	 		C_{\s} \left(\Mw ,\Ms  =0\right) = R^{(0)} + \frac{\Kw(\delta_z-\delta_s)}{\Kw(\delta_z-\delta_s)+\Ks (\delta_z-\delta_w)^{+}} \Mw   , \qquad 0 \leq \Mw  \leq\Ma^{(1)},
	 	\end{equation}
		where $R^{(0)}$ is defined in \eqref{eq:R0} and $\Ma^{(1)}$ is defined in \eqref{eq:M1}.
	\end{corollary}	
\begin{IEEEproof}
Achievability follows from the two achievable rate-memory pairs $(R^{(0)}, \Ma^{(0)})$ and  $(R^{(1)}, \Ma^{(1)})$ in \eqref{eq:R0}--\eqref{eq:M1} and by time/memory-sharing arguments. The converse follows from upper bound \eqref{eq:UB3} in Corollary~\ref{cor:UB} when specialized to $k_w=\Kw$. In fact, for $k_w=\Kw$ and  cache size $\Mw \in[0, \Ma^{(1)}]$, the maximizing $\beta$ is:
\begin{equation}\label{eq:betamax}
\beta= \frac{(\delta_z- \delta_s)-\Ks \Mw}{ \Kw (\delta_z -\delta_s)+\Ks (\delta_z -\delta_w)^+} \Kw.
\end{equation}
This choice of $\beta$ makes the two terms in the minimization \eqref{eq:UB3} equal.
\end{IEEEproof}
\bigskip
	
	Notice that when $\delta_z \leq \delta_w$, then \eqref{eq:cor1} specializes to 
	\begin{equation}
	 C_{\s} \left(\Mw ,\Ms  =0\right) = \Mw .	
	\end{equation} 
	The secrecy capacity thus grows in the same way as the cache size at weak receivers. This is achieved by the scheme we termed cache-aided wiretap coding with secret keys. 
	\bigskip 
	
	Notice that in this case the secrecy capacity-memory tradeoff  $C_{\s} \left(\Mw ,\Ms  =0\right)$ grows much faster in the cache size $\Mw$ than its non-secure counterpart $C(\Mw,\Ms=0)$.   In fact, by the upper bound in  \cite{WiggerJournal}, the maximum   slope of the standard capacity-memory tradeoff 
	\begin{equation}
	\gamma :=  \max_{m \geq 0}\left\{\frac{ \mathsf{d} C(\Mw,\Ms=0)}{\mathsf{d} \Mw}\Bigg|_{\Mw=m}\right\}
	\end{equation} is at most 
	\begin{equation}
	\gamma \leq \frac{\Kw \Mw}{\Da}.
	\end{equation}
	
	By the above Corollary~\ref{cor:1} and the concavity of $C_{\textnormal{sec}}(\Mw,\Ms=0)$ in $\Mw$, the maximum slope of the secrecy capacity-memory tradeoff 
		\begin{equation}
		\gamma_{\textnormal{sec}} :=  \max_{m \geq 0} \left\{\frac{ \mathsf{d} C_{\textnormal{sec}}(\Mw,\Ms=0)}{\mathsf{d}  \Mw}\Bigg|_{\Mw=m}\right\}
		\end{equation} 
		is 
		\begin{equation}\label{eq:gs}
			\gamma_{\textnormal{sec}} = \lim\limits_{m\to 0} \left\{\frac{ \mathsf{d} C_{\textnormal{sec}}(\Mw,\Ms=0)}{\mathsf{d}  \Mw}\Bigg|_{\Mw=m}\right\} = \frac{\Kw(\delta_z-\delta_s)}{\Kw(\delta_z-\delta_s)+\Ks (\delta_z-\delta_w)^{+}} \Mw .
		\end{equation}
		 So in contrast to the maximum slope of the standard capacity-memory tradeoff $\gamma$, the maximum slope of the secrecy capacity-memory tradeoff $\gamma_\s$   does not deteriorate with the size of the library $\Da$. The reason for this discrepancy is that in the setup with secrecy constraint an optimal strategy for small cache memories is to exclusively place  secret keys in the cache memories. In this case, each bit of the cache content is useful irrespective of the specifically demanded files. In a setup without secrecy constraint,  only data is placed in the cache memories. So, at least on an intuitive level,  each bit of cache memory is useful only under some  of the demands. 
		 
		 \bigskip
		
		We turn to the regime of large cache memories.
	\begin{corollary}
		When the cache memory $\Mw $ is large:
		\begin{equation}
			C_{\s} \left(\Mw ,\Ms  =0\right) = \frac{\delta_z-\delta_s}{\Ks }, \quad \Mw  \geq \Ma^{(\Kw+2)},
		\end{equation}
		where $\Ma^{(\Kw+2)}$ is defined in \eqref{eq:M4}.
	\end{corollary} \bigskip
	
	The rate-memory pairs \eqref{eq:R3}--\eqref{eq:M4} are attained by means of joint cache-channel coding where the decoders simultaneously adapt to the cache contents and the channel statistics. To emphasize the strength of the joint coding approach, we characterize the rates that are securely achievable under a separate cache-channel coding approach.
	
	For $t \in \{1,\ldots,\Kw-1\}$, define the following rate-memory pairs:
	\begin{subequations}
	\begin{IEEEeqnarray}{rCl}
		R_{\text{sep}}^{(t)} &:=& \frac{(t+1)(1-\delta_w)(\delta_z-\delta_s)}{\Ks (t+1)(1-\delta_w)+(\Kw-t)(\delta_z-\delta_s)}, \\[1.2ex]
		\Ma_{\text{sep}}^{(t)} &:=& \frac{\Da \cdot t(t+1)(1-\delta_w)(\delta_z-\delta_s)+(t+1)(\Kw-t)(\delta_z-\delta_s)\min\{1-\delta_z,1-\delta_w\} }{\Kw\big[\Ks (t+1)(1-\delta_w)+(\Kw-t)(\delta_z-\delta_s)\big]}.~
	\end{IEEEeqnarray}
	\end{subequations} \bigskip

	\begin{proposition}
	\label{prop:LB_sep}
	Any rate $R>0$ is achievable by means of separate cache-channel coding, if it satisfies
		\begin{align}
			R & <   \textnormal{upper hull} \Big\lbrace \big(R^{(\ell)},\Ma^{(\ell)}\big),\big(R_{\textnormal{sep}}^{(t)} ,\Ma_{\textnormal{sep}}^{(t)}\big) \colon \quad \ell \in \{0,1,2,\Kw+3\} \textnormal{ and } t \in \{1,\ldots,\Kw-1\} \Big\rbrace.
		\end{align}
	\end{proposition}
	\begin{IEEEproof}
The rate-memory tuples $\big\{\left(R^{(\ell)},\Ma^{(\ell)}\right)\colon \ell =0, 1,2, \Kw+3\big\}$ in \eqref{eq:R0}--\eqref{eq:M2} and \eqref{eq:RKw3}--\eqref{eq:MKw3} are achieved by  the schemes in Subsections~\ref{sec:wiretap} and \ref{sec:superposition}. It can be verified that these schemes apply a separate cache-channel coding architecture. Rate-memory pairs $\big\{\big(R_{\textnormal{sep}}^{(t)} ,\Ma_{\textnormal{sep}}^{(t)}\big) \colon   t =1,\ldots,\Kw-1\bigl\}$ are achieved by a scheme that communicates to weak and strong receivers in two independent phases: in the first phase  it combines the Sengupta secure coded caching scheme \cite{Clancy15} with a standard optimal BC code to communicate to the weak receivers, and in the second phase it applies a standard wiretap BC code to  communicate to the strong receivers.	\end{IEEEproof}

		\begin{figure}
		\centering
		\begin{tikzpicture}[scale=2.4]
			\footnotesize
			\draw (0,0.9) -- (5,0.9) -- (5,4.1179) -- (0,4.118) -- (0,0.9);
			\node[below] at (2.5,0.65) {$\Mw $};
			\node[rotate=90] at (-0.6,2.5) {$C_{\text{sec}}(\Mw ,\Ms  =0)$};
			\node[below] at (0,0.85) {$0$};
			\node[left] at (0,0.92) {$0.01$}; 
			
			\draw[densely dotted, gray] (0,1.8) -- (2.17,1.8);		\draw[densely dotted, gray] (4.94,1.8) -- (5,1.8);
			\draw[densely dotted, gray] (0,2.7) -- (5,2.7);		\draw[densely dotted, gray] (0,3.6) -- (5,3.6);
			
			\draw[densely dotted, gray] (1,0.9) -- (1,4.1179);		\draw[densely dotted, gray] (2,0.9) -- (2,4.1179);		
			\draw[densely dotted, gray] (3,0.9) -- (3,1);		\draw[densely dotted, gray] (3,2) -- (3,4.1179);	
			\draw[densely dotted, gray] (4,0.9) -- (4,1);		\draw[densely dotted, gray] (4,2) -- (4,4.1179);
			
%
			
			\draw (0,1.8) -- (0.05,1.8);		\node[left] at (0,1.8) {$0.02$};
			\draw (0,2.7) -- (0.05,2.7);		\node[left] at (0,2.7) {$0.03$};
			\draw (0,3.6) -- (0.05,3.6);		\node[left] at (0,3.6) {$0.04$};
			
			\draw (1,0.9) -- (1,0.95);		\node[below] at (1,0.85) {$0.2$};			
			\draw (2,0.9) -- (2,0.95);		\node[below] at (2,0.85) {$0.4$};
			\draw (3,0.9) -- (3,0.95);		\node[below] at (3,0.85) {$0.6$};
			\draw (4,0.9) -- (4,0.95);		\node[below] at (4,0.85) {$0.8$};
			\draw (5,0.9) -- (5,0.95);		\node[below] at (5,0.85) {$1$};

			\draw[thick] (0,0.0125*90) -- (0.0143*5,0.0214*90) -- (0.0231*5,0.0231*90) -- (0.0719*5,0.0291*90) -- (0.2345*5,0.0316*90) -- (0.4727*5,0.0333*90) -- (1*5,0.0333*90);

			\draw[thick,dashed] (0,0.0125*90) -- (0.0143*5,0.0214*90) -- (0.03*5,0.0312*90) -- (0.042*5,0.0333*90) -- (1*5,0.0333*90);
			
			\draw[thick,blue] (0,0.0125*90) -- (0.0143*5,0.0214*90) -- (0.0231*5,0.0231*90) -- (0.1782*5,0.0273*90) -- (0.372*5,0.03*90) -- (0.5768*5,0.0316*90) -- (0.787*5,0.0326*90) -- (1*5,0.0333*90);
			
			\draw[thick,magenta] (0,0.0263*90) -- (0.0649*5,0.0352*90) -- (0.0875*5,0.0369*90) -- (0.0905*5,0.037*90) -- (0.0907*5,0.037*90) -- (0.0907*5,0.037*90) -- (0.2938*5,0.0416*90) -- (0.5498*5,0.0439*90) -- (0.8255*5,0.0452*90) -- (1*5,0.045755*90);


			\draw[thick,dashed] (2.25,1.85) -- (2.4,1.85);	\node[right] at (2.4,1.85) {UB on $C_{\text{sec}}(\Mw ,0)$};
			\draw[thick] (2.25,1.6) -- (2.4,1.6);				\node[right] at (2.4,1.6) {LB on $C_{\text{sec}}(\Mw ,0)$ using joint coding};
			\draw[thick,blue] (2.25,1.35) -- (2.4,1.35);		\node[right] at (2.4,1.35) {LB on $C_{\text{sec}}(\Mw ,0)$ using separate coding};							
			\draw[thick,magenta] (2.25,1.1) -- (2.4,1.1);	\node[right] at (2.4,1.1) {LB on $C(\Mw ,0)$};
			\draw (2.17,1) rectangle (4.94,2);
			        
			\draw (0,0.0125*90) circle (0.25mm);			\draw (0.0143*5,0.0214*90) circle (0.25mm);
			\draw (0.0231*5,0.0231*90) circle (0.25mm);		\draw (0.0719*5,0.0291*90) circle (0.25mm);
			\draw (0.2345*5,0.0316*90) circle (0.25mm);		\draw (0.4727*5,0.0333*90) circle (0.25mm);
			\draw (1*5,0.0333*90) circle (0.25mm);
			
			\node[right] at (0,0.0125*90) {$(a)$};			\node[right] at (0.0143*5,0.0208*90) {$(b)$};
			\node[right] at (0.0231*5,0.0227*90) {$(c)$};	\node[above] at (0.4727*5,0.0333*90) {$(d)$};
			\node[above] at (0.98*5,0.0333*90) {$(e)$};

		\end{tikzpicture} 
		\caption{Upper and lower bounds on $C_{\text{sec}}(\Mw ,\Ms  =0)$ for $\delta_w = 0.7$, $\delta_s = 0.3$, $\delta_z = 0.8$, $\Da=30$, $\Kw = 5$, and $\Ks  = 15$.} \label{fig:example1}
	\end{figure}
	
	\bigskip
	\subsection{Numerical Comparisons}
	In 
	Figure~\ref{fig:example1}, we compare the presented bounds  at hand of an example with $5$ weak and $15$ strong receivers and where the eavesdropper is degraded with respect to all receivers. The figure shows  the upper and lower bounds on the secrecy capacity-memory tradeoff $C_\s(\Mw,\Ms=0)$  in Corollary~\ref{cor:UB} and Theorem~\ref{thm:LB}. It also shows the rates achieved by the separate cache-channel coding scheme leading to Proposition~\ref{prop:LB_sep}. Finally, the  figure  presents the lower bound on the standard capacity-memory tradeoff in \cite{WiggerJournal}. 
	  
	  The presented lower bound of  Theorem~\ref{thm:LB} (see the black solid line in Figure~\ref{fig:example1}) is piece-wise linear with the end points of the pieces corresponding to the points in \eqref{eq:RM}. The left-most point $(a)$ corresponds to the capacity in the absence of cache memories. The second and third left-most points $(b)$ and $(c)$ are obtained by storing only secret keys in the cache memories. %
	  The right-most point $(e)$ corresponds to the point where the messages can be sent at the same rate as if only strong receivers were present in the system. This performance is trivially achieved by storing all messages in each of the weak receivers' cache memories and holding the delivery communication only to strong receivers. After this point, the capacity cannot be increased further  because strong receivers do not have cache memories. Through coding, the same rate can also be achieved without storing the entire library at each weak receiver, see the second right-most point $(d)$.
	
	For small and large cache memories, our upper and lower bounds are exact. This shows that in the regime of small cache memories, it is optimal to place only secret keys in the weak receivers' cache memories. In this regime, the slope of $C_{\textnormal{sec}}(\Mw ,\Ms  =0)$ in $\Mw $ is steep (see  Corollary~\ref{cor:1} and Equation \eqref{eq:gs}) because the  secret keys stored in the cache memories are always helpful, irrespective of the specific demands $\mathbf{d}$. In particular, the slope is not divided by the library size $\Da$ as is  the case in the traditional caching setup without secrecy constraint. 
	
	In the regime of moderate or large cache memories, the proposed placement strategies also store   information about the messages in the cache memories. In this regime, the slope  $C_{\s}$ is smaller and  proportional to $\frac{1}{\Da}$, because only a fraction of the cache content is effectively helpful for a specific demand $\mathbf{d}$. 
	
	\bigskip
	
	Figure~\ref{fig:ex2} shows the bounds for an example where the eavesdropper is stronger than the weak receivers but not the strong receivers:
	\begin{equation}
	\delta_s < \delta_z < \delta_w.
	\end{equation}It shows that positive rates can be achieved even if $\delta_z \leq \delta_w$, because   messages sent to weak receivers can be specially secured by  means of one-time pads using the secret keys stored in their cache memories.

	\begin{figure}
		\centering
		\begin{tikzpicture}[scale=2.6]
			\footnotesize
			\draw (0,0) -- (3.6,0) -- (3.6,2.1) -- (0,2.1) -- (0,0);
			\node[below] at (1.8,-0.15) {$\Mw $};
			\node[rotate=90] at (-0.5,1) {$C_{\text{sec}}(\Mw ,\Ms  =0)$};
			\node[below] at (0,0) {$0$};
			\node[left] at (0,0) {$0$}; 
			
%
	
			\draw[densely dotted, gray] (0,2) -- (3.6,2);	\draw[densely dotted, gray] (0,1.5) -- (3.6,1.5);	
			\draw[densely dotted, gray] (0,1) -- (3.6,1);	
			\draw[densely dotted, gray] (0,0.5) -- (1.11,0.5);	\draw[densely dotted, gray] (3.52,0.5) -- (3.6,0.5);	
			
			\draw[densely dotted, gray] (0.6,0) -- (0.6,2.1);		
			\draw[densely dotted, gray] (1.2,0) -- (1.2,0.1);		\draw[densely dotted, gray] (1.2,0.75) -- (1.2,2.1);
			\draw[densely dotted, gray] (1.8,0) -- (1.8,0.1);		\draw[densely dotted, gray] (1.8,0.75) -- (1.8,2.1);	
			\draw[densely dotted, gray] (2.4,0) -- (2.4,0.1);		\draw[densely dotted, gray] (2.4,0.75) -- (2.4,2.1);	
			\draw[densely dotted, gray] (3,0) -- (3,0.1);			\draw[densely dotted, gray] (3,0.75) -- (3,2.1);
			
			\draw (0,2) -- (0.05,2);		\node[left] at (0,2) {$0.02$};
			\draw (0,1.5) -- (0.05,1.5);		\node[left] at (0,1.5) {$0.015$};
			\draw (0,1) -- (0.05,1);		\node[left] at (0,1) {$0.01$};			
			\draw (0,0.5) -- (0.05,0.5);		\node[left] at (0,0.5) {$0.005$};	
			
			\draw (0.6,0) -- (0.6,0.05);		\node[below] at (0.6,0) {$0.1$};			
			\draw (1.2,0) -- (1.2,0.05);		\node[below] at (1.2,0) {$0.2$};						
			\draw (1.8,0) -- (1.8,0.05);		\node[below] at (1.8,0) {$0.3$};
			\draw (2.4,0) -- (2.4,0.05);		\node[below] at (2.4,0) {$0.4$};
			\draw (3,0) -- (3,0.05);			\node[below] at (3,0) {$0.5$};
			\draw (3.6,0) -- (3.6,0.05);		\node[below] at (3.6,0) {$0.6$};

			\draw[thick] (0,0) -- (0.0109*6,0.0109*100) -- (0.0133*6,0.0133*100)-- (0.055*6,0.0187*100) -- (0.128*6,0.02*100) -- (3.6,0.02*100);
			\draw[thick,dashed] (0,0) -- (0.02*6,0.02*100) -- (3.6,0.02*100);
			\draw[thick,blue] (0,0) -- (0.0655,1.0909) -- (0.08,1.3333) -- (0.68,1.6667) -- (1.3745,1.8182) -- (2.1029,1.9048) -- (2.8471,1.9608) -- (3.6,2);
			
%
			\draw[thick,dashed] (1.16,0.6) -- (1.3,0.6);	\node[right] at (1.3,0.6) {UB on $C_{\text{sec}}(\Mw ,0)$};
			\draw[thick] (1.16,0.4) -- (1.3,0.4);	\node[right] at (1.3,0.4) {LB on $C_{\text{sec}}(\Mw ,0)$ using joint coding};
			\draw[thick,blue] (1.16,0.2) -- (1.3,0.2);	\node[right] at (1.3,0.2) {LB on $C_{\text{sec}}(\Mw ,0)$ using separate coding};
			\draw (1.11,0.1) rectangle (3.52,0.75);

		\end{tikzpicture} 
		\caption{Upper and lower bounds on $C_{\text{sec}}(\Mw ,\Ms  =0)$ for $\delta_w = 0.8$, $\delta_s = 0.3$, $\delta_z = 0.6$, $\Da=30$, $\Kw = 5$, and $\Ks  = 15$.} \label{fig:ex2}
	\end{figure}
	\bigskip

\section{Results when All Receivers have Cache Memories} \label{sec:2sided}	

We turn to the  case where all receivers have cache memories, so $\Mw , \Ms   >0$. In this case, we do not impose any constraint on the eavesdropper's channel, so $\delta_z$ can be larger or smaller than $\delta_s, \delta_w$. 


\subsection{Coding Schemes}
We present four coding schemes. The first one only stores secret keys in all the cache memories, the second one stores keys in all cache memories and data at weak receivers, and the last two schemes store keys and data at all the receivers.

\subsubsection{Cached Keys}~~ \label{sec:allkeys}

\underline{\textit{Placement phase:}} Store  independent secret keys $K_1,\ldots,K_\Ka$ in the cache memories of Receivers $1,\ldots, \Ka$:
\input{random_keys_all}

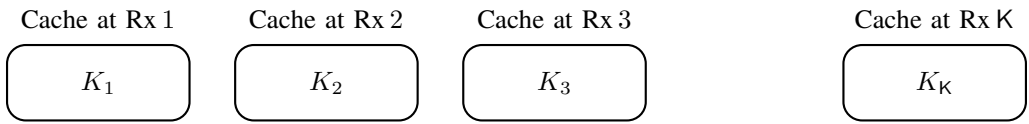
\underline{\textit{Delivery phase:}} Apply a standard (non-secure) broadcast code to send the  secured message tuple
\begin{equation}
\textbf{W}_{\sec}:= \Big(\enc{W_{d_1}}{K_1}, \ \enc{W_{d_2}}{K_2},\ \enc{W_{d_3}}{K_3},\  \ldots,\ \enc{W_{d_{\Ka}}}{K_{\Ka}}\Big),
\end{equation}
to Receivers~$1,\ldots, \Ka$, respectively. With the secret key $K_k$ stored in its cache memory, each  Receiver~$k\in\Ka$  recovers a guess of its desired message $W_{d_k}$. See Section~\ref{sec:allkeys_analysis} on how to choose the parameters of the scheme. 
\bigskip

\subsubsection{Secure Cache-Aided Piggyback Coding with Keys at All Receivers}~~\label{sec:piggyback_allkeys}

The difference between this scheme and  the secure cache-aided piggyback coding scheme of  Sections~\ref{sec:piggyback} and \ref{sec:piggyback_analysis} is that additional secret keys are placed in the cache memories of weak and strong receivers so that communications in Subphases~2 and 3 can entirely  be secured with these keys, i.e., no wiretap binning is required. 

We outline  the scheme for $\Kw=3$ weak receivers and $\Ks=1$ strong receiver. For a detailed description and an analysis in the general case, see Section~\ref{sec:piggyback_allkeys_analysis}.\\
Divide each message $W_d$ into six submessages 
\begin{equation}
W_d = \big(W_{d,\{1\}}^{(A)},W_{d,\{2\}}^{(A)},W_{d,\{3\}}^{(A)},W_{d,\{1,2\}}^{(B)},W_{d,\{1,3\}}^{(B)},W_{d,\{2,3\}}^{(B)}\big), \quad d \in \D,
\end{equation}
where the first three are of equal rate and the latter three are of equal rate.\\

\underline{\textit{Placement phase:}} Placement is as described in the following table, where $K_{\{1,2,3\}}$, $K_{\{1,2\}}$, $K_{\{2,3\}}$, $K_{\{1,3\}}$, $K_{4,\{1,2\}}$, $K_{4,\{2,3\}}$, $K_{4,\{1,3\}}$, $K_{\{4\}}$
denote independent secret keys.

\input{cache_contents_piggyback_all}

\underline{\textit{Delivery  phase:}} Time-sharing is applied over three subphases and Subphase 2 is further divided into 3 periods. 

%

Transmission in Subphase~1 is as described in Section~\ref{sec:piggyback}. 
In Subphase~3, the secured message $\enc{W_{d_4}^{(A)}}{K_{4}}$ is sent to the strong receiver~$4$ using a standard point-to-point code. This transmission is secure if the key $K_4$ is chosen sufficiently long.

In the first period of Subphase~2, the transmitter uses the standard piggyback codebook in Figure~\ref{fig:piggy_standard}  to transmit the secure message
\begin{equation}
\mathbf{W}_{\sec,\{1,2\}}^{(A)}=\enc{W_{d_1,\{2\}}^{(A)} \oplus W_{d_2,\{1\}}^{(A)}}{\ K_{\{1,2\}}}\qquad 
\end{equation}
to Receivers~$1$ and $2$ and the secure message
\begin{equation}
{W}_{\sec,\{4\}}^{(B)} = \enc{W_{d_4,\{1,2\}}^{(B)}}{K_{4,\{1,2\}}}
\end{equation}
to Receiver~$4$. 
Receivers~$1$ and $2$ can reconstruct ${W}_{\sec,\{4\}}^{(A)}$ from their cache contents, and thus decode message $\mathbf{W}_{\sec,\{1,2\}}^{(A)}$ based solely on the column of the codebook that corresponds to $W_{\sec,\{4\}}^{(B)}$.  Receiver~$4$ decodes both messages $\mathbf{W}_{\sec,\{1,2\}}^{(A)}$ and ${W}_{\sec,\{4\}}^{(B)}$. From the decoded secured messages and the keys stored in their cache memories, Receivers $1, 2,$ and $4$ can recover their desired message parts $W_{d_1,\{2\}}^{(A)}, W_{d_2,\{1\}}^{(A)}$ and $W_{d_4, \{1,2\}}^{(B)}$.
\input{piggyback}

In the same way, using a standard piggyback codebook, the secured messages $\enc{W_{d_1,\{3\}}^{(A)} \oplus W_{d_3,\{1\}}^{(A)}}{\ K_{\{1,3\}}}$ and $\enc{W_{d_4,\{1,3\}}^{(B)}}{K_{4,\{1,3\}}}$ are transmitted  in Period 2 to Receivers~$1,3$, and $4$,  and   the secured messages $\enc{W_{d_2,\{3\}}^{(A)} \oplus W_{d_3,\{2\}}^{(A)}}{\ K_{\{2,3\}}}$ and $\enc{W_{d_4,\{2,3\}}^{(B)}}{K_{4,\{2,3\}}}$ are transmitted in Period 3 to Receivers~$2,3$, and $4$.


\bigskip

\subsubsection{Symmetric Secure Piggyback Coding}~~\label{sec:piggyback_symmetric}~~ 

Each message is split into two submessages $W_{d}=(W_d^{(A)}, W_d^{(B)})$, and communication is in three subphases.  
Submessages of $\{W_d^{(A)}\}$ and corresponding secret keys are placed in weak receivers' cache memories according to the  Sengupta et al. 
secure coded caching placement algorithm \cite{Clancy15}. Submessages  of $\{W_d^{(B)}\}$  and corresponding secret keys are placed in strong receivers' cache memories according to the same placement algorithm.  In Subphase~1, the Sengupta et al. delivery scheme for submessages  $\{W_d^{(A)}\}$ is combined with a standard BC code to transmit only to weak receivers and in Subphase 3 it is combined with a standard BC code to transmit only to strong receivers.  Transmission in Subphase 2 is divided into 
into $\Kw\Ks$ periods, each dedicated to a pair of weak and strong receivers $i\in\Kw$ and $j\in\Ks$. A standard  piggyback codebook is used in each of these periods to send secured messages to the corresponding pair of receivers. The secret keys securing these messages have been pre-placed  in the appropriate cache memories. 
\medskip

We now describe the scheme in more detail for the special case $\Kw=3$ and $\Ks=2$, and  for parameters $t_w=2$ and $t_s=1$. The general scheme is described and analyzed in Section~\ref{sec:piggyback_symmetric_analysis}.

Divide each $W_d$ into six submessages 
\begin{equation}
W_d = \big(W_{d,\{1,2\}}^{(A)},W_{d,\{1,3\}}^{(A)},W_{d,\{2,3\}}^{(A)},W_{d,\{4\}}^{(B)},W_{d,\{5\}}^{(B)}\big), \quad d \in \D,
\end{equation}
where the first three are of equal rate and the latter two are of equal rate.\\

\underline{\textit{Placement phase:}} Generate independent secret keys 
$K_{\{1,2,3\}}$, $K_{w,\{1,4\}}$, $K_{w,\{1,5\}}$, $K_{w,\{2,4\}}$, $K_{w,\{2,5\}}$, $K_{w,\{3,4\}}$, $K_{w,\{3,5\}}$, $K_{s,\{1,4\}}$, $K_{s,\{1,5\}}$, $K_{s,\{2,4\}}$, $K_{s,\{2,5\}}$, $K_{s,\{3,4\}}$, $K_{s,\{3,5\}}$, $K_{\{4,5\}}$.

Placement of information in the cache memories is as indicated in the following table.

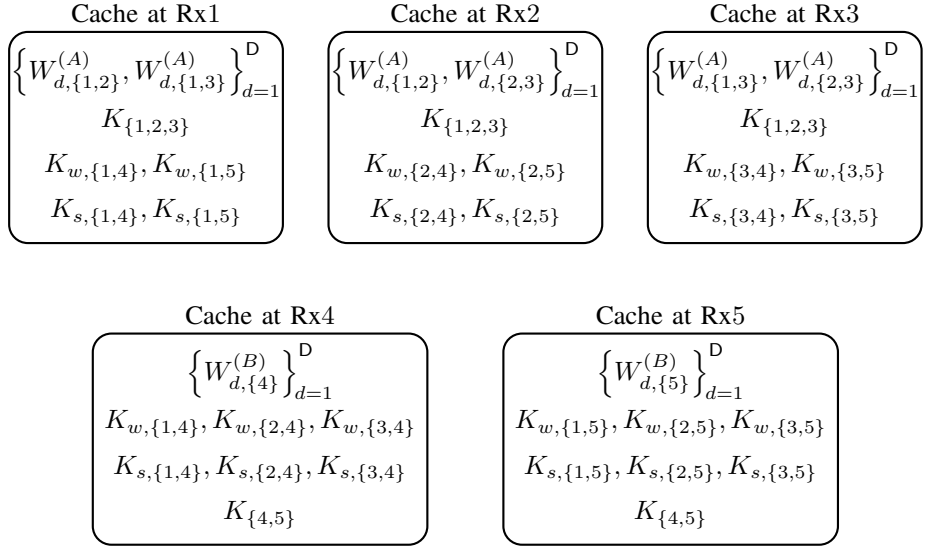
\begin{figure}[H] \centering
	\begin{tikzpicture}
	
	\node[above] at (1.8,6.8) {Cache at Rx$1$};
	\draw[rounded corners=7pt,thick] (0,4) rectangle (3.6,6.8);
	\node at (1.8,6.3) {$\left\lbrace W_{d,\{1,2\}}^{(A)}, W_{d,\{1,3\}}^{(A)}\right\rbrace_{\!d=1}^{\Da}$};
	\node at (1.8,5.6) {$K_{\{1,2,3\}}$};	
	\node at (1.8,5) {$K_{w,\{1,4\}},K_{w,\{1,5\}}$};	
	\node at (1.8,4.4) {$K_{s,\{1,4\}},K_{s,\{1,5\}}$};	
	
	\node[above] at (6,6.8) {Cache at Rx$2$};
	\draw[rounded corners=7pt,thick] (4.2,4) rectangle (7.8,6.8);
	\node at (6,6.3) {$\left\lbrace W_{d,\{1,2\}}^{(A)}, W_{d,\{2,3\}}^{(A)}\right\rbrace_{\!d=1}^{\Da}$};
	\node at (6,5.6) {$K_{\{1,2,3\}}$};	
	\node at (6,5) {$K_{w,\{2,4\}},K_{w,\{2,5\}}$};	
	\node at (6,4.4) {$K_{s,\{2,4\}},K_{s,\{2,5\}}$};	
	
	\node[above] at (10.2,6.8) {Cache at Rx$3$};
	\draw[rounded corners=7pt,thick] (8.4,4) rectangle (12,6.8);	
	\node at (10.2,6.3) {$\left\lbrace W_{d,\{1,3\}}^{(A)}, W_{d,\{2,3\}}^{(A)}\right\rbrace_{\!d=1}^{\Da}$};
	\node at (10.2,5.6) {$K_{\{1,2,3\}}$};	
	\node at (10.2,5) {$K_{w,\{3,4\}},K_{w,\{3,5\}}$};	
	\node at (10.2,4.4) {$K_{s,\{3,4\}},K_{s,\{3,5\}}$};	
	
	\node[above] at (3.3,2.8) {Cache at Rx$4$};
	\draw[rounded corners=7pt,thick] (1.1,0) rectangle (5.5,2.8);
	\node at (3.3,2.3) {$\left\lbrace W_{d,\{4\}}^{(B)} \right\rbrace_{\!d=1}^{\Da}$};	
	\node at (3.3,1.6) {$K_{w,\{1,4\}},K_{w,\{2,4\}},K_{w,\{3,4\}}$};
	\node at (3.3,1) {$K_{s,\{1,4\}},K_{s,\{2,4\}},K_{s,\{3,4\}}$};	
	\node at (3.3,0.4) {$K_{\{4,5\}}$};
	
	\node[above] at (8.7,2.8) {Cache at Rx$5$};
	\draw[rounded corners=7pt,thick] (6.5,0) rectangle (10.9,2.8);
	\node at (8.7,2.3) {$\left\lbrace W_{d,\{5\}}^{(B)} \right\rbrace_{\!d=1}^{\Da}$};	
	\node at (8.7,1.6) {$K_{w,\{1,5\}},K_{w,\{2,5\}},K_{w,\{3,5\}}$};
	\node at (8.7,1) {$K_{s,\{1,5\}},K_{s,\{2,5\}},K_{s,\{3,5\}}$};	
	\node at (8.7,0.4) {$K_{\{4,5\}}$};
	\end{tikzpicture} 
\end{figure}

\underline{\textit{Delivery phase:}} The delivery phase is divided into three subphases, where Subphase~$2$ is further divided into 6 periods. 


Subphase~$1$ is intended only for weak receivers. The transmitter sends the secured message 
\begin{equation}
W_{\sec,\{1,2,3\}}^{(A)} = \enc{W_{d_1,\{2,3\}}^{(A)} \oplus W_{d_2,\{1,3\}}^{(A)} \oplus W_{d_3,\{1,2\}}^{(A)}}{\ K_{\{1,2,3\}}}
\end{equation}
using a capacity-achieving code to the weak receivers $1,2,$ and $3$. With their cache contents, each of these weak receivers can decode its desired submessage. 



%
The first period of Subphase~$2$ is dedicated to Receivers~$1$ and $4$. The transmitter sends the secured messages
\begin{equation}
W_{\sec,1,4}^{(B)}= \textsf{sec}\big(W_{d_1,\{4\}}^{(B)}, K_{w,\{1,4\}} \big) 
 \end{equation}
 and
 \begin{equation}
W_{\sec,4,1}^{(A)}= \textsf{sec}\big(W_{d_4,\{1,2\}}^{(A)}, K_{s,\{1,4\}} \big) 
 \end{equation}
 using a standard piggyback codebook where rows encode $W_{\sec,1,4}^{(B)}$ and columns encode $W_{\sec,4,1}^{(A)}$. (This corresponds to the piggyback codebook in Figure~\ref{fig:piggy_standard} where $\mathbf{W}_{\sec,\{1,2\}}^{(A)}$ 
 needs to be replaced by $W_{\sec,1,4}^{(B)}$ and ${W}_{\sec,\{4\}}^{(B)}$
  by $W_{\sec,4,1}^{(A)}$.) Since Receiver~$1$ can reconstruct $W_{\sec,4,1}^{(A)}$ from its cache content, it decodes $W_{\sec,1,4}^{(B)}$  based on the \emph{single column} of the piggyback codebook indicated by $W_{\sec,4,1}^{(A)}$. Similarly, since Receiver~$4$ can reconstruct $W_{\sec,1,4}^{(B)}$ from its cache content, it decodes $W_{\sec,4,1}^{(A)}$ based on the \emph{single row} corresponding to $W_{\sec,1,4}^{(B)}$. With their decoded secure messages and the secret keys in their cache memories, Receiver~$1$ can recover $W_{d_1,\{4\}}^{(B)}$ and  Receiver~$4$ can recover $W_{d_4,\{1,2\}}^{(A)}$. 
 
   A similar scheme is used in the subsequent periods to convey the parts $W_{d_1}^{(B)}, W_{d_2}^{(B)}$, $W_{d_3}^{(B)}$, $W_{d_4}^{(A)}$, $W_{d_5}^{(A)}$  that are not stored in their cache memories to Receivers $1$--$5$. Table~\ref{tab:subphase2} shows the concerned receivers, the conveyed messages and the keys used in each period of Subphase~$2$.
   
   	{ \renewcommand{\arraystretch}{2} 
	\begin{table}[h]
	\centering
 	\caption{Messages sent and keys used in the six periods of Subphase~$2$ for the example with $\Kw=3$ weak receivers and $\Ks=2$ strong receivers.} \label{tab:subphase2}
		\begin{tabular}{|m{1.9cm}|c|c|c|c|c|c|}
		\hline
			 &  \multicolumn{6}{ c| }{\textbf{Subphase~2}}   \\ \cline{2-7}
			& \textbf{Period~1} & \textbf{Period~2} & \textbf{Period~3}  & \textbf{Period~4} & \textbf{Period~5} &  \textbf{Period~6} \\ \hline \hline
			 \textbf{Receivers} & $1,4$ & $2,4$ & $3,4$ & $1,5$ & $2,5$ & $3,5$ \\ \hline
			\textbf{Keys} & $\!\!K_{w,\{1,4\}},K_{s,\{1,4\}}\!\!$ & $\!\!K_{w,\{2,4\}},K_{s,\{2,4\}}\!\!$ & $\!\!K_{w,\{3,4\}},K_{s,\{3,4\}}\!\!$ &$\!\!K_{w,\{1,5\}},K_{s,\{1,5\}}\!\!$ & $\!\!K_{w,\{2,5\}},K_{s,\{2,5\}}\!\!$ & $\!\!K_{w,\{3,5\}},K_{s,\{3,5\}}\!\!$ \\ \hline
			\textbf{Messages for weak receivers} & $W_{d_1,\{4\}}^{(B)}$ & $W_{d_2,\{4\}}^{(B)}$ & $W_{d_3,\{4\}}^{(B)}$ & $W_{d_1,\{5\}}^{(B)}$ & $W_{d_2,\{5\}}^{(B)}$ & $W_{d_3,\{5\}}^{(B)}$ \\  \hline
			\textbf{Messages for strong receivers} & $W_{d_4,\{1,2\}}^{(A)}$ & $W_{d_4,\{2,3\}}^{(A)}$ & $W_{d_4,\{1,3\}}^{(A)}$ & $W_{d_5,\{1,2\}}^{(A)}$ & $W_{d_5,\{2,3\}}^{(A)}$ & $W_{d_5,\{1,3\}}^{(A)}$ \\ \hline
		\end{tabular} 
	\end{table}
	}
   
\textit{Subphase 3:} The transmitter sends
\begin{equation}
W_{\sec,\{4,5\}}^{(B)} = \enc{W_{d_4,\{5\}}^{(B)} \oplus W_{d_5,\{4\}}^{(B)}}{\ K_{\{4,5\}}}
\end{equation}
using a capacity-achieving code to the strong receivers $4$ and $5$. With their cache contents, each receiver can decode its desired submessage sent in this subphase.

Choosing all keys sufficiently long ensures that the delivery communication satisfies the secrecy constraint~\eqref{eq:jointSecrecy}.
\bigskip

\subsubsection{Secure Generalized Coded Caching}\label{sec:generalized}~~

The scheme is based on the \emph{generalized coded caching} algorithms of \cite{WiggerYenerJournal}, but where the produced zero-padded XORs are secured with independent secret keys and these keys are placed in the cache memories of the receivers that decode  the XORs. Choosing the secret keys sufficient long, ensures than  the secrecy constraint~\eqref{eq:jointSecrecy} is satisfied.

\subsection{Results on the Secrecy Capacity-Memory Tradeoff}

Consider the following $\Ka+\Kw+\Kw \Ks$ rate-memory tuples. Let $\tilde{R}^{(0)}=R^{(0)}$ and $\tilde{\Ma}_{w}^{(0)}=\tilde{\Ma}_s^{(0)}=0$.

	\begin{subequations}\label{eq:RM_Assign}
	\begin{itemize}
	\item Let 
	\begin{IEEEeqnarray}{rCl}
		\tilde{R}^{(1)} &:=& \frac{(1-\delta_s)(1-\delta_w)}{\Kw(1-\delta_s)+\Ks (1-\delta_w)}, \label{eq:R2_Assign}  \hspace{11.2cm}\\
		\tilde{\Ma}_{w}^{(1)} &:=&  \frac{(1-\delta_s)\min \big\{ 1-\delta_z,1-\delta_w\big\}}{\Kw(1-\delta_s)+\Ks (1-\delta_w)}, \label{eq:Mw2_Assign} \\
		\tilde{\Ma}_{s}^{(1)} &:=& \frac{(1-\delta_w)\min \big\{ 1-\delta_z,1-\delta_s\big\}}{\Kw(1-\delta_s)+\Ks (1-\delta_w)}; \label{eq:Ms2_Assign} \\ \nonumber
	\end{IEEEeqnarray}
	\item For $t \in \{1,\ldots,\Kw-1\}$, let
	\begin{IEEEeqnarray}{rCl}
		\tilde{R}^{(t+1)} &:=& \frac{(t+1)(1-\delta_w)(1-\delta_s)\big[\Ks t(1-\delta_w) + (\Kw-t+1)(\delta_w-\delta_s)\big]}{(\Kw-t+1)(1-\delta_s)\big[\Ks (t+1)(1-\delta_w) + (\Kw-t)(\delta_w-\delta_s) \big] + \Ks ^2t(t+1)(1-\delta_w)^2}
, \label{eq:R3_Assign}\\
		\tilde{\Ma}_{w}^{(t+1)} &:=& \frac{\Da \cdot t(t+1)(1-\delta_w)(1-\delta_s)\big[\Ks (t-1)(1-\delta_w)+(\Kw-t+1)(\delta_w-\delta_s)\big]}{\Kw\big[(\Kw-t+1)(1-\delta_s)[\Ks (t+1)(1-\delta_w)+(\Kw-t)(\delta_w-\delta_s)] + \Ks ^2t(t+1)(1-\delta_w)^2\big]} \qquad\qquad \hspace{1mm}  \nonumber\\
			&&+  \frac{\Ks t(t+1)(\Kw-t+1)(1-\delta_w)(1-\delta_s)\min \{1-\delta_z,1-\delta_s\}}{\Kw\big[(\Kw-t+1)(1-\delta_s)[\Ks (t+1)(1-\delta_w)+(\Kw-t)(\delta_w-\delta_s)] + \Ks ^2t(t+1)(1-\delta_w)^2\big]}\nonumber\\
			&& + \frac{(t+1)(\Kw-t)(\Kw-t+1)(1-\delta_s)(\delta_w-\delta_s)\min \left\lbrace 1-\delta_z,1-\delta_w \right\rbrace\big]}{\Kw\big[(\Kw-t+1)(1-\delta_s)[\Ks (t+1)(1-\delta_w)+(\Kw-t)(\delta_w-\delta_s)] + \Ks ^2t(t+1)(1-\delta_w)^2\big]}, \label{eq:Mw3_Assign}   \\ 
		\tilde{\Ma}_{s}^{(t+1)} &:=& \frac{\Ks t(t+1)(1-\delta_w)^2 \min \left\lbrace 1-\delta_z,1-\delta_s \right\rbrace}{(\Kw-t+1)(1-\delta_s)\big[ \Ks (t+1)(1-\delta_w)+(\Kw-t)(\delta_w-\delta_s)\big] + \Ks ^2t(t+1)(1-\delta_w)^2}\nonumber\\
			&&+  \frac{(t+1)(\Kw-t+1)(1-\delta_w)(1-\delta_s)\min \left\lbrace (\delta_w-\delta_z)^+,\delta_w-\delta_s \right\rbrace}{(\Kw-t+1)(1-\delta_s)\big[ \Ks (t+1)(1-\delta_w)+(\Kw-t)(\delta_w-\delta_s)\big] + \Ks ^2t(t+1)(1-\delta_w)^2};  \label{eq:Ms3_Assign}\\ \nonumber
	\end{IEEEeqnarray}
	
	\item For each pair $t_w \in \{1,\ldots,\Kw\}$ and $t_s \in \{1,\ldots,\Ks \}$, let
	\begin{IEEEeqnarray}{rCl}
		\tilde{R}^{(\Kw+(t_w-1)\Ks +t_s)} &:=& \frac{(t_w+1)(t_s+1)(1-\delta_w)(1-\delta_s)\big[\Ks (1-\delta_w) + \Kw(1-\delta_s)\big]}{\Kw(\Kw\!-t_w)(t_s\!+1)(1-\delta_s)^2 \!+\! \Ks (t_w\!+1)(1-\delta_w)\big[ (\Ks \!-t_s)(1-\delta_w) \!+\! \Kw(t_s\!+1)(1-\delta_s) \big]},  \nonumber \\
		\label{eq:R4_Assign}\\
		\tilde{\Ma}_{w}^{(\Kw+(t_w-1)\Ks +t_s)} &:=& \frac{(t_w+1)(t_s+1)(1-\delta_s)^2 \big[ \Da \cdot t_w(1-\delta_w)+ (\Kw-t_w)\min \left\lbrace 1-\delta_z,1-\delta_w \right\rbrace \big]}{\Kw(\Kw\!-t_w)(t_s\!+1)(1-\delta_s)^2 \!+\! \Ks (t_w\!+1)(1-\delta_w)\big[ (\Ks \!-t_s)(1-\delta_w) \!+\! \Kw(t_s\!+1)(1-\delta_s) \big]} \nonumber \\
		&&+ \frac{\Ks (t_w+1)(t_s+1)(1-\delta_w)(1-\delta_s) \min \left\lbrace 1-\delta_z,2-\delta_w-\delta_s \right\rbrace }{\Kw(\Kw\!-t_w)(t_s\!+1)(1-\delta_s)^2 \!+\! \Ks (t_w\!+1)(1-\delta_w)\big[ (\Ks \!-t_s)(1-\delta_w) \!+\! \Kw(t_s\!+1)(1-\delta_s) \big]}, \nonumber \\ \label{eq:Mw4_Assign}   \\ 
	\tilde{\Ma}_{s}^{(\Kw+(t_w-1)\Ks +t_s)} &:=& \frac{(t_w+1)(t_s+1)(1-\delta_w)^2\big[\Da \cdot t_s(1-\delta_s) + (\Ks -t_s) \min\left\lbrace 1-\delta_z,1-\delta_s \right\rbrace\big]}{\Kw(\Kw\!-t_w)(t_s\!+1)(1-\delta_s)^2 \!+\! \Ks (t_w\!+1)(1-\delta_w)\big[ (\Ks \!-t_s)(1-\delta_w) \!+\! \Kw(t_s\!+1)(1-\delta_s) \big]} \nonumber \\
		 &&+  \frac{\Kw(t_w+1)(t_s+1)(1-\delta_w)(1-\delta_s) \min \left\lbrace 1-\delta_z,2-\delta_w-\delta_s \right\rbrace }{\Kw(\Kw\!-t_w)(t_s\!+1)(1-\delta_s)^2 \!+\! \Ks (t_w\!+1)(1-\delta_w)\big[ (\Ks \!-t_s)(1-\delta_w) \!+\! \Kw(t_s\!+1)(1-\delta_s) \big]};  \nonumber \\ \label{eq:Ms4_Assign}   \\ \nonumber
	\end{IEEEeqnarray}
	
	\item For $t \in \{1,\ldots,\Ka-1\}$, let
	\begin{IEEEeqnarray}{rCl} 
		\tilde{R}^{(\Kw+\Kw \Ks +t)} &:=& \frac{ \sum_{t_w=\max\{0,t-\Ks \}}^{\min\{t,\Kw\}} {\Kw \choose t_w}{\Ks \choose t-t_w}(1-\delta_w)^{-t_w}(1-\delta_s)^{t_w}}{\sum_{t_w=\max\{0,t+1-\Ks \}}^{\min\{t+1,\Kw\}} {\Kw \choose t_w}{\Ks  \choose t+1-t_w}(1-\delta_w)^{-t_w}(1-\delta_s)^{t_w-1}}, \label{eq:R5_Assign}\\
		\tilde{\Ma}_{w}^{(\Kw+\Kw \Ks +t)} &:=&  \frac{\Da  \cdot \sum_{t_w=\max\{1,t-\Ks \}}^{\min\{t,\Kw\}} {\Kw-1\choose t_w-1}{\Ks  \choose t-t_w}(1-\delta_w)^{-t_w}(1-\delta_s)^{t_w}}{\sum_{t_w=\max\{0,t+1-\Ks \}}^{\min\{t+1,\Kw\}} {\Kw \choose t_w}{\Ks  \choose t+1-t_w}(1-\delta_w)^{-t_w}(1-\delta_s)^{t_w-1}} \nonumber \\
		 &&+ \min \left\lbrace \frac{(t+1)(1-\delta_z)}{\Ka},  \frac{(1-\delta_w)^{-t_w}(1-\delta_s)^{t_w} \left[ \binom{\Ka-1}{t}(1-\delta_s)- \binom{\Kw-1}{t}(\delta_w-\delta_s) \right]}{\sum_{t_w=\max\{0,t+1-\Ks \}}^{\min\{t+1,\Kw\}} \binom{\Kw}{t_w}\binom{\Ks }{t+1-t_w} (1-\delta_w)^{-t_w}(1-\delta_s)^{t_w} } \right\rbrace,\qquad	\label{eq:Mw5_Assign}  \\
		\tilde{\Ma}_{s}^{(\Kw+\Kw \Ks +t)} &:=&  \frac{\Da \cdot \sum_{t_w=\max\{0,t-\Ks \}}^{\min\{t-1,\Kw\}} {\Kw \choose t_w}{\Ks -1 \choose t-t_w-1}(1-\delta_w)^{-t_w}(1-\delta_s)^{t_w}}{\sum_{t_w=\max\{0,t+1-\Ks \}}^{\min\{t+1,\Kw\}} {\Kw \choose t_w}{\Ks  \choose t+1-t_w}(1-\delta_w)^{-t_w}(1-\delta_s)^{t_w-1}} \nonumber \\
		 &&+ \min \left\lbrace \frac{(t+1)(1-\delta_z)}{\Ka},  \frac{\binom{\Ka-1}{t}(1-\delta_w)^{-t_w}(1-\delta_s)^{t_w+1}}{\sum_{t_w=\max\{0,t+1-\Ks \}}^{\min\{t+1,\Kw\}} \binom{\Kw}{t_w}\binom{\Ks }{t+1-t_w} (1-\delta_w)^{-t_w}(1-\delta_s)^{t_w} } \right\rbrace. 	\label{eq:Ms5_Assign}  \\ \nonumber
	\end{IEEEeqnarray}

	\end{itemize}
	\end{subequations}	
	
	\bigskip

	\begin{theorem}[Lower bound on $C_{s}(\Mw ,\Ms  )$] \label{thm:LB_Assign}
		\begin{align}
			C_{\s}\left(\Mw ,\Ms  \right)\geq \textnormal{upper hull} & \left\lbrace \left(\tilde{R}^{(\ell)},\tilde{\Ma}_{w}^{(\ell)},\tilde{\Ma}_{s}^{(\ell)}\right) \colon \quad \ell\in\{0,1,\dots,\Ka+\Kw+\Kw \Ks -1\} \right\rbrace.
		\end{align}
	\end{theorem}
	\begin{proof}
		By time/memory-sharing arguments, it suffices to prove the achievability of the  rate-memory triples $\big\{(\tilde{R}^{(\ell)},\tilde{\Ma}_{w}^{(\ell)},\tilde{\Ma}_{s}^{(\ell)}):\ \ell\in\{0,1,\dots,\Ka+\Kw+\Kw \Ks -1\}\big\}$. The triple $\{(\tilde{R}^{(1)},\tilde{\Ma}_{w}^{(1)}, \tilde{\Ma}_{s}^{(1)})$ is achieved by the ``cached keys" scheme, see Subsections~\ref{sec:allkeys} and \ref{sec:allkeys_analysis}. The triples $(\tilde{R}^{(\ell)},\tilde{\Ma}_{w}^{(\ell)},\tilde{\Ma}_{s}^{(\ell)})$, for $\ell\in\{2,\dots,\Kw\}\big\}$, are achieved by the ``secure piggyback coding scheme with keys at all receivers", see Subsections~\ref{sec:piggyback_allkeys} and \ref{sec:piggyback_allkeys_analysis}. The triples $(\tilde{R}^{(\ell)},\tilde{\Ma}_{w}^{(\ell)},\tilde{\Ma}_{s}^{(\ell)})$, for $\ell\in\{\Kw+1,\dots,\Kw+\Kw\Ks\}\big\}$, are achieved by the ``symmetric secure piggyback coding" scheme, see  Subsections~\ref{sec:piggyback_symmetric} and \ref{sec:piggyback_symmetric_analysis}. Finally, the triples  $(\tilde{R}^{(\ell)},\tilde{\Ma}_{w}^{(\ell)},\tilde{\Ma}_{s}^{(\ell)})$, for $\ell\in\{\Kw+\Kw\Ks+1,\ldots,   \Ka +\Kw+ \Kw \Ks-1\}$, are achieved by the ``secure generalized coded caching" scheme sketched in Section~\ref{sec:generalized}. 
	\end{proof}
	
	\bigskip
	
	\begin{corollary} \label{cor:2}
		The rate-memory tradeoff $\left(\tilde{R}^{(1)},\tilde{\Ma}_{w}^{(1)},\tilde{\Ma}_{s}^{(1)}\right)$ is optimal, i.e.,
		\begin{equation}
			C_{\s}\left(\Mw =\tilde{\Ma}_{w}^{(1)},\Ms  =\tilde{\Ma}_{s}^{(1)}\right) = \tilde{R}^{(1)}.
		\end{equation}
	\end{corollary}
		\begin{IEEEproof}
			Achievability follows from the two achievable rate-memory triples $(\tilde{R}^{(0)}, \tilde{\Ma}_w^{(0)}=0, \tilde{\Ma}_s^{(0)}=0)$ and  $(\tilde{R}^{(1)}, \tilde{\Ma}_w^{(1)}, \tilde{\Ma}_s^{(1)})$ in \eqref{eq:R2_Assign}--\eqref{eq:Ms2_Assign} and by time/memory-sharing arguments. The converse follows by specializing upper bound \eqref{eq:UB1_assign} in Theorem~\ref{thm:UB_assign}  to $k_w=\Kw$ and $k_s=\Ks$. In fact, for $k_w=\Kw$, $k_s=\Ks$, and  cache sizes $\Mw = \tilde{\Ma}_w^{(1)}$ and $\Ms= \tilde{\Ma}_s^{(1)}$, the maximizing $\beta$ equals
\begin{equation}
\beta = \frac{\Kw(1-\delta_s)}{\Kw(1-\delta_s)+\Ks(1-\delta_w)}
\end{equation}
when $\delta_z<\delta_w$, and it equals $\beta=0$ when $\delta_z \geq \delta_w$.
\end{IEEEproof}
	
\bigskip
Notice that when $\delta_z \leq \delta_s$, then
\begin{equation}\tilde{R}^{(1)} = \tilde{\Ma}_{w}^{(1)} = \tilde{\Ma}_{s}^{(1)}.
\end{equation} This rate is achieved  by XORing the messages with secret keys stored in the cache memories, and by sending the resulting bits using a traditional non-secure code for the erasure BC.

\section{Global Secrecy Capacity-Memory Tradeoff} \label{sec:globalCs}

In the preceding sections, we considered scenarios with unequal cache sizes at the receivers and showed that in these scenarios joint cache-channel coding schemes can significantly improve over the traditional separation-based schemes with their typical uniform cache assignment. In this section, we emphasize the importance of unequal cache sizes that depend on the receivers' channel conditions by focusing on the \emph{global secrecy capacity-memory tradeoff} $C_{\textnormal{sec,glob}}$, which is the largest secrecy capacity-memory tradeoff that is  possible given  a  total cache budget  
	\begin{equation}
	  \Kw \Mw +\Ks \Ms  	\leq \Ma_{\text{tot}}.
	\end{equation}
	We assign the same  cache memory size $\Mw $  to all weak receivers and the same cache memory size $\Ms  $ to all strong receivers. Using simple time/memory-sharing arguments, it can be shown that this assumption is without loss in optimality.
	
	So, the main quantity of interest in this section is 
	\begin{equation}
C_{\textnormal{sec,glob}} (\Ma_{\text{tot}}): = \max_{\substack{\Mw , \Ms   \geq 0 \colon \\ \Kw \Mw +\Ks \Ms  	\leq \Ma_{\text{tot}} } } C_{\textnormal{sec}} (\Mw , \Ms  )
	\end{equation}

\subsection{Results}
Using the achievability results in Theorems~\ref{thm:LB} and \ref{thm:LB_Assign} combined with an appropriate cache assignment and time/memory-sharing arguments, yields the following lower bound on $C_{\textnormal{sec,glob}}(\Ma_{\textnormal{tot}})$.

	\begin{corollary}{(Lower bound on $C_{\textnormal{sec,glob}}(\Ma_{\tot})$)} \label{cor:LB}
		The global secrecy capacity-memory tradeoff $C_{\textnormal{sec,glob}}$ is lower bounded as
		\begin{align}
			C_{\textnormal{sec,glob}}(\Ma_{\tot}) \geq \textnormal{upper hull} & \Big\{\left(R^{(\ell)},\Ma_{\tot} = \Kw \Ma^{(\ell)}\right),~ \left(\tilde{R}^{(\ell')}, \Ma_{\tot} = \Kw \tilde{\Ma}_{w}^{(\ell')}+\Ks  \tilde{\Ma}_{s}^{(\ell')}\right) \colon \nonumber\\ 
			& \qquad\qquad \ell\in\{0,\dots,\Kw+3\} \textnormal{ and } \ell' \in\{1,\dots,\Ka+\Kw+\Kw \Ks-1 \} \Big\}.
		\end{align}
	\end{corollary}
	
	\bigskip
	
	\begin{proposition}
	If the eavesdropper is weaker than the strong receivers, i.e., 
	\begin{equation}
	\delta_z > \delta_s,
	\end{equation} then for small cache sizes: 
	\begin{IEEEeqnarray}{rCl}\label{eq:smallCM}
			C_{\textnormal{sec,glob}}( \Ma_{\tot})&=& 
			R^{(0)} + \frac{(\delta_z-\delta_s)^+}{\Kw(\delta_z-\delta_s)^++\Ks (\delta_z-\delta_w)^{+}} \Ma_{\tot} ,   \qquad 0 \leq \Ma_{\tot}  \leq\Kw \Ma^{(1)}. 
		\end{IEEEeqnarray}
		If the eavesdropper is at least as strong as the strong receivers, i.e., 
		\begin{equation}
		\delta_z \leq \delta_s, \label{eq:evestrong}
		\end{equation}
		then for small cache sizes: 
			\begin{IEEEeqnarray}{rCl}\label{eq:smallCM2}
				C_{\textnormal{sec,glob}}( \Ma_{\tot})&=& \frac{\Ma_{\tot}}{\Ka} , \qquad 0 \leq\Ma_{\tot}  \leq \Ka \cdot 	\frac{(1-\delta_w)(1-\delta_s)}{\Kw (1-\delta_s) +\Ks(1-\delta_w)} .		\end{IEEEeqnarray}
	\end{proposition}	
	\begin{IEEEproof}
		Achievability of \eqref{eq:smallCM} follows from Corollary~\ref{cor:1} and by assigning all available cache  memory uniformly across weak receivers, so $\Mw=\frac{\Ma_{\tot}}{\Kw}$ and $\Ms=0$.  The converse to \eqref{eq:smallCM} can be proved by  specializing Lemma~\ref{lemma:1} to the set of all users  $\mathcal{S}=\K$. For a given choice of the cache assignment $\Mw, \Ms \geq 0$, Lemma~\ref{lemma:1} yields (amongst others) that for some $\beta \in[0,1]$, the secrecy capacity-memory tradeoff is upper bounded as 
				\begin{IEEEeqnarray}{rCl}\label{eq:bound1}
				\Kw \cdot	C_{\textnormal{sec,glob}}( \Ma_{\tot})&\leq & \beta (\delta_z-\delta_w)^++ \Kw \Mw
							\end{IEEEeqnarray}
							and 
							\begin{IEEEeqnarray}{rCl}\label{eq:bound2}
								\Ka\cdot C_{\textnormal{sec,glob}}( \Ma_{\tot})&\leq &\beta (\delta_z-\delta_w)^+ + (1-\beta)(\delta_z-\delta_s)^++ \Ma_{\tot}.
							\end{IEEEeqnarray}	
							Upper bounding now $\Kw\Mw$ by $\Ma_{\tot}$ and combining \eqref{eq:bound1} and \eqref{eq:bound2} into a single bound yields:
														\begin{equation} \label{eq:u2}
								C_{\textnormal{sec,glob}}(\Ma_{\tot})\leq  \max_{\beta \in [0,1]} \min \left\lbrace \frac{\beta(\delta_z - \delta_w)^++ \Ma_{\tot}}{\Kw}  ,  \; \frac{\beta (\delta_z - \delta_w)^+ + (1-\beta)(\delta_z-\delta_s)^+}{\Ka} + \frac{\Ma_{\tot}  }{\Ka}   \right\rbrace.
							\end{equation}
							If $\delta_z >\delta_w$, then the maximizing $\beta$ in the above upper bound is as in \eqref{eq:betamax} when $\Mw$ is replaced by $\Ma_{\tot}$; otherwise the maximizing $\beta$ equals $0$. Plugging these values into \eqref{eq:u2} establishes the desired upper bound. 
							
							Achievability of \eqref{eq:smallCM2} follows from time/memory-sharing arguments and from the achievability of the rate-memory triple $(\tilde{R}^{(1)}, \tilde{M}_w^{(1)}, \tilde{M}_s^{(1)})$ in \eqref{eq:R2_Assign}--\eqref{eq:Ms2_Assign}, which under \eqref{eq:evestrong}  specializes to the rate-memory triple:
							\begin{equation}
							\tilde{R}^{(1)}=\tilde{\Ma}_w^{(1)}= \tilde{\Ma}_s^{(1)} =  	\frac{(1-\delta_w)(1-\delta_s)}{\Kw (1-\delta_s) +\Ks(1-\delta_w)}.
							\end{equation}
					In fact, the rate-memory triples under consideration are achieved by storing an independent secret key at each receiver and securing the messages with these keys by means of one-time pads. This requires a uniform cache  assignment across \emph{all} receivers, i.e., $\Mw=\Ms=\frac{\Ma_{\tot}}{\Ka}$. The converse to \eqref{eq:smallCM2} follows from upper bound \eqref{eq:bound2}, which under \eqref{eq:evestrong} specializes to $C_{\textnormal{sec,glob}}(\Ma_{\tot}) \leq \Ma_{\tot}/\Ka$.	\end{IEEEproof}
\bigskip

	Figure~\ref{fig:cacheContent} illustrates our choices of the cache contents when $\delta_z > \delta_s$ in the order of increasing total cache memory $\Ma_{\tot}$.  When the total cache budget $\Ma_{\tot}$ is small, all of it is assigned to the weak receivers and it is solely used  to store secret keys. For a slightly larger cache budget $\Ma_{\tot}$, it is assigned to all receivers and  secret keys are stored in the cache memories. When the total cache budget $\Ma_{\tot}$ exceeds the size required to store the keys for securing  the entire communication, the additional space is used to store data at weak receivers. Once the cached data renders the weak  receivers equally powerful (in terms of their decoding performance) as  the strong receivers, data is  also stored in strong receivers' cache memories. When $\delta_z \leq \delta_s$, then our choice is similar, but we start with placing secret keys directly at all the receivers.

	\begin{figure}
		\begin{tikzpicture}
			\draw[thick,->] (0,0.2) -- (0,0) -- (12,0);	\node[below] at (0,0) {$0$};
			\node[left] at (0,0.5) {Cache at strong receivers:};
			\node[left] at (0,1.2) {Cache at  weak receivers:};
			
			\draw[thick] (3,0) -- (3,0.2);		\node at (1.5,0.5) {Empty};		\node at (1.5,1.2) {Keys};
			\draw[thick] (6,0) -- (6,0.2);		\node at (4.5,0.5) {Keys};		\node at (4.5,1.2) {Keys};
			\draw[thick] (9,0) -- (9,0.2);		\node at (7.5,0.5) {Keys};		\node at (7.5,1.2) {Keys + Data};
			\node at (10.5,0.5) {Key + Data};		\node at (10.5,1.2) {Keys + Data};
			
			\node[below] at (6,-0.2) {Total cache budget $\Ma_{\tot}$};
			
		\end{tikzpicture} \caption{Cache content at weak and strong receivers in order of increasing total cache budget.} \label{fig:cacheContent} 
	\end{figure}

\subsection{Results under Uniform Cache Assignment}
	For comparison, we also propose a lower bound on $C_{\textnormal{sec,glob}}(\Ma_{\tot})$ when the cache memory is uniformly assigned over \emph{all receivers}, i.e. 
	\begin{equation}
		\Mw=\Ms = \frac{\Ma_{\tot}}{\Ka}.
	\end{equation}
	
	Consider the following $\Ka+2$ rate-memory pairs: \vspace{3mm}
	
	\begin{subequations}\label{eq:RM_Sym}
	\begin{itemize}
	
	\item \ \vspace{-7mm}
	\begin{IEEEeqnarray}{rCl}
		R_{\Sym}^{(0)} &:=&  \frac{(\delta_z-\delta_s)^{+} \cdot (\delta_z-\delta_w)^{+}}{\Kw(\delta_z-\delta_s)+\Ks (\delta_z-\delta_w)}, \hspace{9.35cm} \\
		\Ma_{\Sym}^{(0)} &:=& 0; \\ \nonumber
	\end{IEEEeqnarray}

	\item \	\vspace{-7mm}
	\begin{IEEEeqnarray}{rCl}
		R_{\Sym}^{(1)} &:=& \frac{(1-\delta_w)(1-\delta_s)}{\Kw(1-\delta_s)+\Ks (1-\delta_w)}, \label{eq:R1_Sym} \hspace{9.65cm}\\
		\Ma_{\Sym}^{(1)} &:=& \frac{(1-\delta_s) \min \left\lbrace 1-\delta_z,1-\delta_w \right\rbrace }{\Kw(1-\delta_s)+\Ks (1-\delta_w)}; \label{eq:M1_Sym} \\ \nonumber
	\end{IEEEeqnarray}
	
	\item \quad For $t \in \{1,\ldots,\Ks -1\}$
	\begin{IEEEeqnarray}{rCl}
		R_{\Sym}^{(t+1)} &:=& \frac{\binom{\Ka}{t}(1-\delta_w)(1-\delta_s)}{\binom{\Ka}{t+1}(1-\delta_s)-\binom{\Ks }{t+1}(\delta_w-\delta_s)}, \label{eq:R2_Sym}\\
		\Ma_{\Sym}^{(t+1)} &:=& \frac{\Da \cdot t\binom{\Ka}{t}(1-\delta_w)(1-\delta_s)}{\Ka\big[ \binom{\Ka}{t+1}(1-\delta_s)-\binom{\Ks }{t+1}(\delta_w-\delta_s)\big]}+\frac{(\Ka-t)\binom{\Ka}{t}(1-\delta_s)\min \left\lbrace1-\delta_z,1-\delta_w \right\rbrace}{\Ka\big[ \binom{\Ka}{t+1}(1-\delta_s)-\binom{\Ks }{t+1}(\delta_w-\delta_s)\big]}; \hspace{1.6cm} \label{eq:M2_Sym} \\ \nonumber
	\end{IEEEeqnarray}
	
	\item \quad For $t \in \{\Ks ,\ldots,\Ka\}$
	\begin{IEEEeqnarray}{rCl}
		R_{\Sym}^{(t+1)} &:=& \frac{(t+1)(1-\delta_w)}{(\Ka-t)}, \label{eq:R3_Sym}\\
		\Ma_{\Sym}^{(t+1)} &:=& \frac{\Da \cdot t(t+1)(1-\delta_w)}{\Ka(\Ka-t)}+ \frac{(t+1)}{\Ka} \min \left\lbrace 1-\delta_z,1-\delta_w \right\rbrace. \hspace{5.5cm} \label{eq:M3_Sym} \\ \nonumber
	\end{IEEEeqnarray}
	
	\end{itemize}
	\end{subequations}

	\begin{proposition}[Lower Bound on $C_{\s}(\Mw =\Ma_{\Sym},\Ms  =\Ma_{\Sym})$]\label{thm:LB_Sym}
		\begin{align}
			C_{\s}\left(\Mw =\Ma_{\Sym},\Ms  =\Ma_{\Sym}\right)\geq \textnormal{upper hull} & \left\lbrace \left(R_{\Sym}^{(\ell)},\Ma_{\Sym}^{(\ell)}\right) \colon  \quad \ell\in\{0,\ldots,\Ka+1\} \right\rbrace.
		\end{align}
	\end{proposition}
	\begin{proof}
	
	It suffices to prove the achievability of the $\Ka+2$ rate-memory pairs $\{(R^{(\ell)}_{\Sym}, \Ma^{(\ell)}_{\Sym}) \colon \; \ell=0,\ldots,\Ka+1\}$. The rate-memory pair $(R^{(0)}_{\Sym},\Ma^{(0)}_{\Sym})$ is achievable by Remark~\ref{rmk:R0}. The  rate-memory pair  $(R^{(1)}_{\Sym},\Ma^{(1)}_{\Sym})$ is achievable because, by Theorem~\ref{thm:LB_Assign},   rate $\tilde{R}^{(1)}=R^{(1)}_{\Sym}$ is achievable with cache size $\tilde{\Ma}_w^{(1)}$ at weak receivers and  cache size $\tilde{\Ma}_{s}^{(1)}$  at strong receivers,  and because $ \Ma^{(1)}_{\Sym}= \max \big\{ \tilde{\Ma}_w^{(1)}, \tilde{\Ma}_s^{(1)} \big\}$. For each $t \in \{1,\ldots,\Ka\}$,  the rate-memory pair $(R^{(t+1)}_{\Sym},\Ma^{(t+1)}_{\Sym})$ is achieved by a scheme that  combines the Sengupta et al. secure coded caching scheme \cite{Clancy15} with a standard BC code. Notice that when $t \geq \Ks$, then each of the secured XORs produced by the Sengupta et al. scheme is intended for at least one weak receiver, and communication is limited by this weak receiver. Each secured XOR, which is of rate ${\Ka \choose t}^{-1}R_{\textnormal{sym}}$, thus requires slightly more than ${\Ka \choose t}^{-1}\frac{ nR_{\textnormal{sym}}}{1-\delta_w}$ channel uses to be transmitted reliably. 
	 When $t < \Ks$, then ${\Ks \choose t+1}$ of the secured XOR are only sent to strong receivers and can thus be sent reliably by using slightly more than ${\Ka \choose t}^{-1}\frac{ nR_{\textnormal{sym}}}{1-\delta_s}$ channel uses. The remaining ${\Ka \choose t+1}- {\Ks \choose t+1}$ secured XORs  are sent to at least one weak receiver, and can thus be sent reliably by using slightly more than  ${\Ka \choose t}^{-1}\frac{ nR_{\textnormal{sym}}}{1-\delta_w}$  channel uses. 
	\end{proof}
	
	\bigskip
		\subsection{Numerical Comparison}

	\begin{figure} \centering
		\begin{tikzpicture}[scale=2.5]
			\footnotesize
			\draw (0,0) -- (4.5455,0) -- (4.5455,2.8) -- (0,2.8) -- (0,0);
			\node[below] at (2,-0.2) {$M_{\text{tot}}=\Kw \Mw +\Ks \Ms  $};
			\node[rotate=90] at (-0.5,1.4) {$C_{\text{sec}}(\Ma_\tot)$};
			\node[below] at (0,0) {$0$};	\node[left] at (0,0) {$0$};

			\draw[densely dotted, gray] (0.9091,0) -- (0.9091,2.8); 	\draw[densely dotted, gray] (1.8182,0) -- (1.8182,2.8);
			\draw[densely dotted, gray] (2.7273,0) -- (2.7273,0.085);	\draw[densely dotted, gray] (2.7273,0.95) -- (2.7273,2.8);		
			\draw[densely dotted, gray] (3.6364,0) -- (3.6364,0.085);	\draw[densely dotted, gray] (3.6364,0.95) -- (3.6364,2.8);
			\draw[densely dotted, gray] (0,0.7) -- (2.1,0.7);		\draw[densely dotted, gray] (4.45,0.7) -- (4.5455,0.7);	
			\draw[densely dotted, gray] (0,1.4) -- (4.5455,1.4);	\draw[densely dotted, gray] (0,2.1) -- (4.5455,2.1);

			\draw (0,0.7) -- (0.05,0.7);		\node[left] at (0,0.7) {$0.02$};
			\draw (0,1.4) -- (0.05,1.4);		\node[left] at (0,1.4) {$0.04$};
			\draw (0,2.1) -- (0.05,2.1);		\node[left] at (0,2.1) {$0.06$};
			\node[left] at (0,2.8) {$0.08$};
			
			\draw (0.9091,0) -- (0.9091,0.05);		\node[below] at (0.9091,0) {$5$};
			\draw (1.8182,0) -- (1.8182,0.05);		\node[below] at (1.8182,0) {$10$};
			\draw (2.7273,0) -- (2.7273,0.05);		\node[below] at (2.7273,0) {$15$};			
			\draw (3.6364,0) -- (3.6364,0.05);		\node[below] at (3.6364,0) {$20$};
			\node[below] at (4.5455,0) {$25$};
			
    		\draw[thick,green] (0,0.0046*35) -- (0.16/5.5,0.012*35) -- (0.2/5.5,0.0126*35) -- (1.1765/5.5,0.0229*35) -- (3/5.5,0.0323*35) -- (5.377/5.5,0.0393*35) -- (8.0897/5.5,0.0445*35) -- (10.993/5.5,0.0483*35) -- (14/5.5,0.051*35) -- (17.0608/5.5,0.053*35) -- (20.1471/5.5,0.0546*35) -- (23.2432/5.5,0.0557*35) -- (4.5455,0.0562*35);    		
    		\draw[thick] (0,0.1615) -- (0.0291,0.42) -- (0.0364,0.4421) -- (0.2212,0.8211) -- (0.3114,0.9191) -- (0.5814,1.1923) -- (0.8545,1.435) -- (1.6995,1.9943) -- (2.8857,2.6025) -- (3.36,2.8);
    		\draw[thick,red] (0,0.4421) -- (0.1641,0.8211) -- (0.2388,0.9191) -- (0.435,1.1589) -- (0.7454,1.435) -- (1.554,1.9943) -- (2.7039,2.6025)-- (3.1691,2.8);	
    		\draw[thick,blue] (0,0.1615) -- (0.0459,0.4421) -- (0.2738,0.7742) -- (0.7045,1.1462) -- (1.3633,1.5630) -- (2.282,2.0215) -- (3.4916,2.5206) -- (4.2764,2.8);     		
    		
    		\draw[thick] (2.15,0.8) -- (2.3,0.8);		\node[right] at (2.3,0.8) {LB on $C_{\text{sec,glob}}(\Ma_\tot )$};
			\draw[thick,green] (2.15,0.6) -- (2.3,0.6);	 \node[right] at (2.3,0.6) {LB on $C_{\text{sec}}(\Mw=\Ma_\tot/\Kw, \Ms=0 )$};
			\draw[thick,blue] (2.15,0.4) -- (2.3,0.4);		\node[right] at (2.3,0.4) {LB on $C_{\text{sec}}(\Mw=\Ma_\tot/\Ka, \Ms=\Ma_\tot/\Ka  )$};
			\draw[thick,red] (2.15,0.2) -- (2.3,0.2);		\node[right] at (2.3,0.2) {LB on $C_{\text{glob}}(\Ma_\tot )$};
			\draw (2.1,0.085) rectangle (4.45,0.95);

		\end{tikzpicture} 
		\caption{Lower bounds on $C_{\text{sec,glob}}(\Ma_\tot)$ for $\delta_w = 0.7$, $\delta_s = 0.2$, $\delta_z = 0.8$, $D=50$, $\Kw = 20$, and $\Ks  = 10$.} \label{fig:ex3}
	\end{figure}

	
Figure~\ref{fig:ex3} plots the  lower bound on $C_{\textnormal{sec,glob}}(\Ma_\tot)$ in Corollary~\ref{cor:LB} (black line) for an example with $\Kw=20$ weak receivers an $\Ks=10$ strong receivers. The eavesdropper is degraded with respect to all legitimate receivers.
	For comparison, the figure also shows our lower bound on the  secrecy capacity-memory tradeoff for the scenarios where the available cache memory is uniformly assigned over all  weak receivers , i.e., $\Ma_{w}=\Ma_{\textnormal{tot}}/\Kw$ and $\Ms  =0$, and where it is uniformly assigned over \emph{all receivers}  (blue line), i.e., $\Ma_{w}=\Ms  =\Ma_{\textnormal{tot}}/\Ka$. Finally, the red  line depicts the lower bound on the standard (non-secure) global capacity-memory tradeoff obtained from  \cite{WiggerYenerJournal} and \cite{Gunduz}. One observes that  our lower bound on the global secrecy capacity-memory tradeoff is   close to the currently best known lower bound on the  {non-secure} capacity-memory tradeoff. Moreover, similarly to \cite{WiggerYenerJournal}, for the {global secrecy capacity-memory tradeoff} it is suboptimal to assign the cache memories uniformly across users (unless the eavesdropper is stronger than all other receivers). In particular, as Proposition~\ref{thm:LB_Sym} shows, for small cache memories all of it should be assigned uniformly over the weak receivers only. \bigskip

\section{Coding Schemes when Only Weak Receivers have Cache Memories} \label{sec:proofLB}

\subsection{Wiretap and Cached Keys}\label{sec:wiretap_analysis}
Let Subphase 1 comprise the first $\beta n$ channel uses and Subphase~2 the last $(1-\beta)n$ channel uses, with $\beta$ chosen as  
\begin{equation} \label{eq:R1_alpha}
\beta=\frac{\Kw(\delta_z-\delta_s)}{\Kw(\delta_z-\delta_s)+\Ks (1-\delta_w)}.
\end{equation}
Fix a small $\epsilon>0$. Let the secret keys $K_1,\ldots, K_{\Kw}$ be independent of each other and of rate
\begin{equation}\label{eq:k}
R_{\key}= \min\bigg\{ \frac{\beta(1-\delta_z)}{\Kw}, R\bigg\},
\end{equation}
where the message rate $R$ is chosen as
\begin{equation}
R= \frac{(\delta_z-\delta_s)(1-\delta_w)}{ \Kw(\delta_z-\delta_s)+\Ks(1-\delta_w)}- \epsilon.
\end{equation}
The scheme requires a cache size at weak receivers equal to 
\begin{equation}
\Mw= R_{\key} = \min\bigg\{\frac{(\delta_z-\delta_s)(1-\delta_z)}{\Kw(\delta_z-\delta_s)+\Ks (1-\delta_w)} , \ \frac{(\delta_z-\delta_s)(1-\delta_w)}{ \Kw(\delta_z-\delta_s)+\Ks(1-\delta_w)}- \epsilon  \bigg\}.
\end{equation}

We analyze the scheme when averaged over the choice of the random code construction $\C$. By the joint typicality lemma, the average probability of a decoding error at the weak receivers  tends to 0 as $n\to \infty$, because 
\begin{equation}\label{eq:Rconst1}
R < \frac{\beta (1-\delta_w)}{\Kw}.
\end{equation} 
The probability of error of the wiretap decoding at the strong receivers tends to 0 as $n\to \infty$, because 
\begin{equation}\label{eq:Rconst2}
R < \frac{(1-\beta) (\delta_z-\delta_s)}{\Ks}.
\end{equation} 
Notice that by the choice of $\beta$ in \eqref{eq:R1_alpha},  the two constraints \eqref{eq:Rconst1} and \eqref{eq:Rconst2} coincide.

We conclude this analysis by verifying the secrecy constraint averaged over the choice of the codebooks $\C$. For fixed blocklength $n$:
\begin{align}
\frac{1}{n} I(  W_1, \ldots, W_\Da; Z^n |\C) & =\frac{1}{n} I ( W_{1}, \ldots, W_{\Da}; Z_{1}^{\beta n} |\C) + \frac{1}{n} I (W_{1}, \ldots, W_{\Da}; Z_{\beta n+1}^{ n} | \C)\nonumber \\
& = \frac{1}{n} I ( W_{d_1}, \ldots, W_{d_{\Kw}}; Z_{1}^{\beta n} |\C) + \frac{1}{n} I (W_{d_{\Kw+1}}, \ldots, W_{d_\Ka}; Z_{\beta n+1}^{ n} | \C)\label{eq:sum}
\end{align}
The first term in \eqref{eq:sum} tends to 0 as $n\to \infty$ by  \eqref{eq:k} and by Lemma~\ref{lem:secure}. The second term tends to 0 as $n\to \infty$ because of the properties of a wiretap code.  

From the analysis we conclude that the average over all choices of the codebooks satisfies the desired properties. There must thus exist at least one good codebook. 
Letting $\epsilon\to 0$ concludes the achievability proof of the rate-memory pair $R=R^{(1)}$ and $\Mw=\Ma^{(1)}$ in \eqref{eq:R1} and \eqref{eq:M1}. \bigskip

\subsection{Cached-Aided Superposition-Jamming}\label{sec:superposition_analysis}

Let 
\begin{equation}
\gamma: = \left\lbrace \frac{\Kw(\delta_z-\delta_s)}{\Ks (1-\delta_w)+\Kw(\delta_w-\delta_s)},\frac{\Kw(1-\delta_s)}{\Ks (1-\delta_w)+\Kw(1-\delta_s)} \right\rbrace ,
\end{equation}
and set the rate 
\begin{equation}
R= \frac{\gamma(1-\delta_w)}{\Kw}-\epsilon.
\end{equation}
Let further $p\in[0,1/2]$ be so that its binary entropy-function $h_b(p)=1-\gamma$.\bigskip

\underline{\textit{Placement phase:}} Generate independent random keys $K_1,\ldots, K_w$ of rate 
\begin{equation}\label{eq:kweak}
R_\key =\min\bigg\{ \frac{ 1-\delta_z}{\Kw},R \bigg\}
\end{equation}
Store each key $K_i$ in the cache memory of Receiver~$i$, for $i\in\Kw$.\bigskip

\underline{\textit{Delivery phase:}}
Fix a small $\epsilon>0$, and generate a cloud center codebook \begin{equation}
\C_{\textnormal{center}}=\big\{   u^n(1),\ldots, u^n(\lfloor 2^{nR}\rfloor)\big\}
\end{equation}
by picking all entries of all codewords i.i.d. according to a Bernoulli-$1/2$ distribution. Choose a wiretap binning rate of 
\begin{equation}\label{eq:Rbint}
R_{\textnormal{bin}}= \{(1-\delta_z)- \gamma(1-\delta_w)\}^+.
\end{equation}
Then, superposition on each codeword $u^n(w)$, for $w\in\{1,\ldots, \lfloor 2^{nR}\rfloor \}$, a satellite codebook
\begin{equation}
\C_{\textnormal{sat}} (w) = \big\{   x^n(v,b|w) \colon v\in\big\{1,\ldots, \lfloor 2^{nR}\rfloor\big\}, b \in \big\{1, \ldots \lfloor 2^{nR_{\textnormal{bin}}}\rfloor \big\}\big\}
\end{equation} 
by drawing all entries of all codewords i.i.d. according to a Bernoulli-$p$ distribution.

The transmitter generates the tuples
\begin{equation}\label{eq:cloud}
\mathbf{W}_{\sec}:=\big(\enc{W_{d_1}}{K_1}, \;\ \enc{W_{d_2}}{K_2}, \;\ \enc{W_{d_3}}{K_3},\  \ldots,\ \enc{W_{d_{\Kw}}}{K_{\Kw}}\big),
\end{equation}
and
\begin{equation}
\mathbf{W}_{\textnormal{sat}}:=\big(W_{d_{\Kw+1}},\ldots, W_{d_{\Ka}}\big),
\end{equation} 
and draws $B$ uniformly at random over $\{1, \ldots,\lfloor 2^{nR_{\textnormal{bin}}} \rfloor$\}.
It then sends the codeword 
\begin{equation}
x^{n}\big( \mathbf{W}_{\textnormal{sat}},B \big|\mathbf{W}_{\sec} \big)  
\end{equation}
over the channel.
Each weak receiver~$i\in\K_w$ decodes  the message $\mathbf{W}_{\sec}$ in the cloud center, and with the key $K_i$ retrieved from its cache memory, it produces a guess of its desired message $W_{d_i}$. Each strong receiver~$j\in\K_s$ decodes both messages $\mathbf{W}_{\sec}, \mathbf{W}_{\textnormal{sat}}$ as well as the random binning index $B$ and extracts the guess of its desired message $W_{d_j}$.

\bigskip

\textit{Analysis:} The scheme requires cache memories at the weak receivers of size 
\begin{equation}\label{eq:Mw12}
\Mw= R_\key = \min\bigg\{\frac{ 1-\delta_z}{\Kw}, \ \frac{(\delta_z-\delta_s)(1-\delta_w)}{\Ks (1-\delta_w)+\Kw(\delta_w-\delta_s)} -\epsilon\bigg\}.
\end{equation}
We analyze the scheme on average over the random choice of the codebook $\C$. The average probability of decoding error at a weak receiver tends to 0 as $n\to \infty$, because 
\begin{subequations}
	\begin{align}
	\Kw	R & < 
	\gamma(1-\delta_w).
	\end{align}
	Similarly,  the  average probability of decoding error at a strong receiver tends to 0 as $n\to \infty$, because
	\begin{align}
	R_{\textnormal{bin}}+\Ks R & < (1-\gamma)(1-\delta_s).
	\end{align}
\end{subequations}

We turn to the secrecy analysis of the scheme. We have: 
\begin{align}
\frac{1}{n} I(  W_1, \ldots, W_\Da; Z^n |\C) & = \frac{1}{n} I(  W_{d_1}, \ldots, W_{d_{\Kw}}; Z^n|\C) + \frac{1}{n} I(  W_{d_{\Kw+1}}, \ldots, W_{d_{\Ka}}; Z^n| W_{d_1}, \ldots, W_{d_{\Kw}}, \C). 
\label{eq:sum2}
\end{align}
The first term on the right-hand side of  \eqref{eq:sum2} tends to 0 as $n \to \infty$ because of  Lemma~1, the key-rate~\eqref{eq:kweak}, and because when averaged over codebooks, the satellite codebooks can simply be considered as additional noise. That the second term in \eqref{eq:sum2} tends to 0 as $n\to \infty$, can be shown following the steps in \cite{ekrem13} or by adapting the steps in \cite[Chapter 22]{elGamalBook}. 
Putting these observations together proves that $\frac{1}{n} I(  W_1, \ldots, W_D; Z^n |\C) \to 0$ as $n\to \infty$.  

By standard arguments, one can then conclude that there must exist at least one deterministic  codebook such that the probability  of decoding error vanishes for $n\to \infty$ and the secrecy constraint \eqref{eq:jointSecrecy} constraint holds.

Since $R\to R^{(2)}$  and  $\Mw \to \Ma^{(2)}$ as $\epsilon \to 0$, this establishes achievability of the rate-memory pair $(R^{(2)}, \Mw^{(2)})$ in \eqref{eq:R2} and \eqref{eq:M2}.
\bigskip
\subsection{Secure Cache-Aided Piggyback Coding I} ~~ \label{sec:piggyback_analysis}

Delivery transmission will be partitioned into three subphases, whose lengths are determined by the parameters: \begin{subequations}\label{eq:betas}
	\begin{IEEEeqnarray}{rCl}
		\beta_1& = &\frac{(\Kw-t)(\Kw-t+1)(\delta_z-\delta_s)\min\{\delta_w-\delta_s,\delta_z-\delta_s\}}{(\Kw\!-t+1)(\delta_z-\delta_s) \big[\Ks (t+1)(1-\delta_w)\!+\!(\Kw\!-t)\min \left\lbrace \delta_w-\delta_s,\delta_z-\delta_s \right\rbrace \big]\! + \Ks ^2t(t+1)(1-\delta_w)^2},  \\
		\beta_2& = & \frac{\Ks(\Kw-t+1) (t+1)(1-\delta_w)(\delta_z-\delta_s)}{(\Kw\!-t+1)(\delta_z-\delta_s) \big[\Ks (t+1)(1-\delta_w)\!+\!(\Kw\!-t)\min \left\lbrace \delta_w-\delta_s,\delta_z-\delta_s \right\rbrace \big]\! + \Ks ^2t(t+1)(1-\delta_w)^2}, \\
		\beta_3& = &  \frac{\Ks^2 t(t+1)(1-\delta_w)^2}{(\Kw-t+1)(1-\delta_s)\big[(\Kw-t)(\delta_w-\delta_s) + \Ks (t+1)(1-\delta_w)\big] + \Ks ^2t(t+1)(1-\delta_w)^2}.
	\end{IEEEeqnarray}
	Notice that $\beta_1+\beta_2+\beta_3=1$.
\end{subequations}

\textit{Message splitting:} Fix a small $\epsilon>0$. Divide each message into two independent submessages
\begin{equation}
W_d = \big[ W_d^{(A)}, W_d^{(B)}\big], \quad d \in \D,
\end{equation}
that are of rates
\begin{equation}
R^{(A)} = \frac{\Ks t(t+1)(1-\delta_w)^2(\delta_z-\delta_s)}{(\Kw\!-t+1)(\delta_z-\delta_s) \big[\Ks (t+1)(1-\delta_w)\!+\!(\Kw\!-t)\min \left\lbrace \delta_w-\delta_s,\delta_z-\delta_s \right\rbrace \big]\! + \Ks ^2t(t+1)(1-\delta_w)^2}-\epsilon/2
\end{equation} 
and
\begin{equation}
R^{(B)} = \frac{(\Kw-t+1)(t+1)(1-\delta_w)(\delta_z-\delta_s)\min\{\delta_w-\delta_s,\delta_z-\delta_s\}}{(\Kw\!-t+1)(\delta_z-\delta_s) \big[\Ks (t+1)(1-\delta_w)\!+\!(\Kw\!-t)\min \left\lbrace \delta_w-\delta_s,\delta_z-\delta_s \right\rbrace \big]\! + \Ks ^2t(t+1)(1-\delta_w)^2}-\epsilon/2.
\end{equation} 
 \bigskip

Denote the $\binom{\Kw}{t-1}$ subsets of $\{1,\ldots,\Kw\}$ of size $t-1$ by $G_1^{(t-1)}, \ldots, G_{\binom{\Kw}{t-1}}^{(t-1)}$; the $\binom{\Kw}{t}$ subsets of $\{1,\ldots,\Kw\}$ of size $t$ by $G_1^{(t)}, \ldots, G_{\binom{\Kw}{t}}^{(t)}$; and the $\binom{\Kw}{t+1}$ subsets of $\{1,\ldots,\Kw\}$ of size $t+1$ by $G_1^{(t+1)}, \ldots, G_{\binom{\Kw}{t+1}}^{(t+1)}$.

Divide each message $W_d^{(A)}$ into $\binom{\Kw}{t-1}$ submessages 
\begin{align}
W_d^{(A)} & = \left\lbrace W_{d,G_\ell^{(t-1)}}^{(A)} \colon \quad \ell \in \left\lbrace 1,\ldots,\binom{\Kw}{t-1} \right\rbrace  \right\rbrace,
\end{align} 
of rate
\begin{equation}\label{eq:rateAa}
r^{(A)} = R^{(A)} {\binom{\Kw}{t-1}}^{-1},
\end{equation}
and divide each message $W_d^{(B)}$ into $\binom{\Kw}{t}$ submessages
\begin{align}
W_d^{(B)}  = \left\lbrace W_{d,G_\ell^{(t)}}^{(B)} \colon \quad \ell \in \left\lbrace 1,\ldots,\binom{\Kw}{t} \right\rbrace  \right\rbrace,
\end{align}
of rate 
\begin{equation}\label{eq:rateBa}
r^{(B)} = R^{(B)} {\binom{\Kw}{t}}^{-1}.
\end{equation} 
\bigskip

\textit{Key generation:}
\begin{itemize}
	\item For each $\ell \in \left\lbrace 1,\ldots,\binom{\Kw}{t+1} \right\rbrace$, generate an independent random key $K_{G_\ell^{(t+1)}}$ of rate 
	\begin{equation}
	R_{\key,1} = {\binom{\Kw}{t+1}}^{-1} \cdot \beta_1 \cdot \min \left\lbrace 1-\delta_z,1-\delta_w \right\rbrace.
	\end{equation} 
	\item For each $\ell \in \left\lbrace 1,\ldots,\binom{\Kw}{t} \right\rbrace$, generate an independent random key $K_{G_\ell^{(t)}}$ of rate
	\begin{equation}
	R_{\key,2} = {\binom{\Kw}{t}}^{-1} \cdot \beta_2 \cdot \min \left\lbrace 1-\delta_z,1-\delta_w \right\rbrace.
	\end{equation}
\end{itemize}
Define the binning rate 
	\begin{equation}\label{eq:Rbin}
	R_{\textnormal{bin}} = {\binom{\Kw}{t}}^{-1} \cdot \beta_2 \cdot \min \left\lbrace \max \left\lbrace 0,\delta_w-\delta_z \right\rbrace, \delta_w-\delta_s \right\rbrace.
	\end{equation}
\bigskip

\underline{\textit{Placement phase:}}
Placement is as indicated in the following table.

\begin{figure}[H] \centering
	\begin{tikzpicture}
	\node [above] at (3.5,3.6) {Cache at weak receiver~$i$};
	\draw[rounded corners=7pt,thick] (-0.5,1) rectangle (7.5,3.6);
	\node at (3.5,2.8) {$\left\lbrace \left\lbrace W_{d,G_\ell^{(t-1)}}^{(A)} \right\rbrace_{\ell \colon i \in G_\ell^{(t-1)}}, \ \left\lbrace W_{d,G_\ell^{(t)}}^{(B)} \right\rbrace_{\ell \colon i \in G_\ell^{(t)}} \right\rbrace_{\!d=1}^{\!\Da}$};	
	\node at (3.5,1.5) {$\left\lbrace K_{G_\ell^{(t+1)}} \right\rbrace_{\ell \colon i \in G_\ell^{(t+1)}}, \ \left\lbrace K_{G_\ell^{(t)}} \right\rbrace_{\ell \colon i \in G_\ell^{(t)}}$};	
		\end{tikzpicture} 
\end{figure}
\bigskip

\underline{\textit{Delivery phase:}} Is divided into three subphases of lengths $\beta_1n$, $\beta_2n$, and $\beta_3n$. \bigskip

\textit{Subphase 1:} This subphase is dedicated  to the transmission of the  parts  of $W_{d_1}^{(B)}, \ldots, W_{d_{\Kw}}^{(B)}$ that are not stored in the cache memories of the respective weak receivers.
For each $\ell \in\{1,\ldots, {\Kw \choose t+1}\}$, the transmitter first calculates the XOR-message 
\begin{equation}
{W}_{\XOR,G_\ell^{(t+1)}}^{(B)}:=\bigoplus\limits_{i \in G_\ell^{(t+1)}} W_{d_i,G_\ell^{(t+1)}\setminus \{i\}}^{(B)},
\end{equation}
and its secured version
\begin{equation}
W_{\sec,G_\ell^{(t+1)}}^{(B)}= \enc{{W}_{\XOR,G_\ell^{(t+1)}}^{(B)}}{ \ K_{G_\ell^{(t+1)}}}.
\end{equation}
It then sends the secured message tuple 
\begin{equation}\label{eq:secureda}
\mathbf{W}_{\sec, w}^{(B)}:=(W_{\sec,G_1^{(t+1)}}^{(B)} , \ldots, W_{\sec,G_{\Kw \choose t+1}^{(t+1)}}^{(B)} )
\end{equation}
using a  capacity achieving point-to-point code to all  weak receivers. After decoding the  message tuple \eqref{eq:secureda}, each weak receiver~$i$ 
retrieves 
the secret key $K_{G_\ell^{(t+1)}}$ from its cache memory, for $\ell \in\{1,\ldots, {\Kw \choose t+1}\}$ so that $i\in G_\ell^{(t+1)}$, and produces
$\hat{W}_{\XOR,G_\ell^{(t+1)}}^{(B)}$. Then, it also retrieves the submessages  
\begin{equation}
\Big\{W_{d_k,G_\ell^{(t+1)}\setminus \{k\}}^{(B)}\Big\}_{k \in G_\ell^{(t+1)}\backslash \{i\} }
\end{equation}
from its cache memory, and guesses the desired submessage as:
\begin{equation}\label{eq:dec11}
\hat{W}_{d_k,G_\ell^{(t+1)}\setminus \{i\}}^{(B)} = \hat{W}_{\XOR,G_\ell^{(t+1)}}^{(B)} \bigoplus_{ k \in G_\ell^{(t+1)}\backslash \{i\} }W_{d_k,G_\ell^{(t+1)}\setminus \{k\}}^{(B)}.
\end{equation} 
At the end of  Subphase~1, each weak receiver $i\in\Kw$ assembles the guesses produced in \eqref{eq:dec11} with the parts of $W_{d_i}^{(B)}$ it has stored in its cache memory to form a guess $\hat{W}_{d_i}^{(B)}$.
\bigskip

\textit{Subphase 2:} This subphase is dedicated to the transmission of the parts of $W_{d_i}^{(A)}$ that are not stored in weak receiver $i$'s cache memory, for $i\in\K_w$, and messages $W_{d_j}^{(B)}$, for $j\in \K_s $.
Time-sharing is applied over $\binom{\Kw}{t}$ periods, each of length 
$$n_2={\binom{\Kw}{t}}^{-1}\cdot\beta_2 n.$$
The periods are labeled   $G_1^{(t)},\ldots,G_{\binom{\Kw}{t} }^{(t)}$. For Period $G_{\ell}^{(t)}$,  $\ell \in \big\{1,\ldots, {\Kw \choose t}\big\}$, the transmitter 
calculates the XOR-message
\begin{subequations} \label{eq:msg}
	\begin{equation} 
	W_{\XOR,G_\ell^{(t)}}^{(A)} : = \bigoplus\limits_{i \in G_\ell^{(t)}} W_{d_i,G_\ell^{(t)}\setminus \{i\}}^{(A)} 
	\end{equation}
	and its secured version 
	\begin{equation}
	W_{\sec,G_{\ell}^{(t)}}^{(A)} :=\enc{	W_{\XOR,G_\ell^{(t)}}^{(A)}}{K_{G_\ell^{(t)}}}
	\end{equation}
	 It also 
	forms the tuple of non-secured messages 
	\begin{equation} 
\mathbf{W}_{s,G_\ell^{(t)}}^{(B)}:=\Big(	W_{d_{\Kw+1},G_\ell^{(t)}}^{(B)}, \ldots,W_{d_\Ka,G_\ell^{(t)}}^{(B)}\Big)
	\end{equation}
\end{subequations}
which it sends to the strong receivers. To this end, generate for each period $G_\ell^{(t)}$ a secure piggyback codebook \cite{schaefer}, see Figure~\ref{fig:piggy}, \begin{equation}
\C_{\textnormal{spg},G_\ell^{(t)}} =\Big\{ x_{G_\ell^{(t)}}^{n_2}\big( \ell_{\textnormal{row}} ; \ \ell_{\textnormal{col}};\  b\big)\colon  \; \ell_{\textnormal{row}}  \in \big\{1,\ldots, \big\lfloor 2^{nr^{(A)}}\big\rfloor \big\},\ \ell_{\textnormal{col}}  \in \big\{1,\ldots, \big\lfloor 2^{nr^{(B)}} \big\rfloor\big\} ,\ b\in \big\{1,\ldots, \lfloor 2^{nR_{\textnormal{bin}}} \rfloor\big\}\Big\},
\end{equation} 
by drawing each entry of each codeword i.i.d. according to a Bernoulli-$1/2$ distribution. The  rates $r^{(A)}$, $r^{(B)}$, and $R_{\textnormal{bin}}$ are defined in \eqref{eq:rateAa}, \eqref{eq:rateBa}, and \eqref{eq:Rbin}.

The transmitter draws an index $B$ uniformly at random over $\{1,\ldots, \lfloor 2^{nR_{\textnormal{bin}}}\rfloor\}$, and  sends the codeword
\begin{equation}
x_{G_\ell^{(t)}}^{n_2}\big( 	{W}_{\sec,G_\ell^{(t)}}^{(A)} ; \ \mathbf{W}_{s,G_\ell^{(t)}}^{(B)}; \ B \big)
\end{equation} 
over the channel.

We now describe the decoding, starting with the decoding at the weak receivers. For each  $i\in \Kw$ and each $\ell\in\big\{1,\ldots,{\Kw \choose t}\big\}$ so that  $i \in G_\ell^{(t)}$,  Receiver~$i$ decodes message $W_{d_{i},G_\ell^{(t)}\setminus \{i\}}^{(A)}$ sent in Period $G_\ell^{(t)}$ by performing the following steps:
	\begin{enumerate}
		\item It   retrieves the secret key $K_{G_\ell^{(t)}}$   and the messages $W_{d_{\Kw+1},G_\ell^{(t)}}^{(B)}, \ldots, W_{d_{\Ka},G_\ell^{(t)}}^{(B)} $ from its cache memory.\\
		\item  It   extracts the subcodebook 
		\begin{IEEEeqnarray}{rCl}\label{eq:subcode}
			\tilde{\C}_{\textnormal{spg},G_\ell^{(t)}} \Big( \mathbf{W}_{s,G_\ell^{(t)}}^{(B)}  \Big) := \Big\{ x^{n_2}_{G_\ell^{(t)}} \Big(\ell_{\textnormal{row}}; \ \mathbf{W}_{s,G_\ell^{(t)}}^{(B)};\ b \Big)  \colon \;\quad \ell_{\textnormal{row}} \in \big\{1, \ldots,\big\lfloor 2^{n_2 r^{(A)}} \big\rfloor\big\},\; b\in \big\{1,\ldots, \lfloor 2^{nR_{\textnormal{bin}}} \rfloor\big\}\Big\}.
		\end{IEEEeqnarray} 
		\vspace{1mm}
		
		\item Based on the reduced codebook	$\tilde{\C}_{\textnormal{spg},G_\ell^{(t)}} \Big( \mathbf{W}_{s,G_\ell^{(t)}}^{(B)}  \Big)$, it decodes the secured message 
		$W_{\sec,G_{\ell}^{(t)}}^{(A)}$ from its channel outputs in this period $G_\ell^{(t)}$.\\[1ex]
		\item It applies the inverse mapping  $\textsf{sec}^{-1}_{K_{G_\ell^{(t)}}}$ to the  message decoded in step 3) to obtain the guess  $\hat{W}_{\XOR,G_\ell^{(t)}}^{(A)}$.\\[1ex]
		\item Finally, it   produces the guess
		\begin{equation}\label{eq:dec22}
		\hat{W}_{d_{i}, G_\ell^{(t)}\backslash \{i\}}^{(A)} = \hat{W}_{\XOR,G_\ell^{(t)}}^{(A)}   \bigoplus\limits_{k\in G_\ell^{(t)} \backslash \{i\}} W_{d_k,G_\ell^{(t)}\setminus \{k\}}^{(A)}.
		\end{equation}
	\end{enumerate}

Strong receivers treat the secure piggyback codebook as a simple wiretap codebook and decode all transmitted messages as well as the bin indices. 

At the end of Subphase 2, each weak receiver $i\in \K_w$ assembles 
the parts of submessage $W_{d_i}^{(A)}$ that it decoded or that are stored in its cache memory to form the guess $\hat{W}_{d_i}^{(A)}$. Similarly, each strong receiver~$j\in\K_s$ assembles the decoded parts of $W_{d_j}^{(B)}$ to form the guess $\hat{W}_{d_j}^{(B)}$.
\bigskip

\textit{Subphase 3:} A wiretap code  is used to send the message tuple 
\begin{equation}
W_{d_{\Kw+1}}^{(A)},\ldots, W_{d_{\Ka}}^{(A)}
\end{equation}  to all the strong receivers~$\Kw+1,\ldots, \Ka$. 

After the last subphase, each Receiver~$k\in\K$ assembles its guesses produced for the two submessages $W_{d_k}^{(A)}$ and $W_{d_k}^{(B)}$ to the final guess $\hat{W}_{k}$.
\bigskip

\underline{\textit{Analysis:}} Analysis is performed averaged over the random choice of the codebooks. We first verify that in each of the three subphases the probability of decoding error tends to 0 as $n\to \infty$. Only weak receivers decode during the first subphase. Probability of decoding error vanishes asymptotically as $n\to\infty$, because 
\begin{equation}
\frac{\binom{\Kw}{t+1}{\binom{\Kw}{t}}^{-1}R^{(B)}}{(1-\delta_w)} = \frac{\frac{\Kw-t}{t+1}	R^{(B)}}{(1-\delta_w)} < \beta_1.
\end{equation}  
Only strong receivers decode during the  third subphase. Probability of decoding error in Subphase 3 vanishes asymptotically, because 
\begin{equation}
\frac{\Ks  R^{(A)}}{(\delta_z-\delta_s)} < \beta_3.
\end{equation}
Probability of decoding error at the weak receivers in Subphase 2 vanishes asymptotically, because
\begin{equation}
\frac{\frac{\Kw-t+1}{t}R^{(A)}}{(1-\delta_w)} < \beta_2.
\end{equation}
Probability of decoding error at the strong receivers in Subphase 2 vanishes asymptotically, because
\begin{equation} 
\frac{\frac{\Kw-t+1}{t}R^{(A)}+\Ks R^{(B)} + {\Kw \choose t} R_{\textnormal{bin}}}{(1-\delta_s)} < \beta_2.
\end{equation} 
Whenever the decodings in all subphases are successful, all receivers  correctly guess their desired messages. Since the decoding error for each subphase vanishes asymptotically, we conclude that also the average overall probability of error vanishes.

We  verify the secrecy constraint averaged over the choice of the code construction. For fixed blocklength~$n$:
\begin{align}
\frac{1}{n} I \left(  W_1, \ldots, W_\Da; Z^n|\C \right) & = \frac{1}{n} I\big( W_{d_1}^{(B)}, \ldots, W_{d_{\Kw}}^{(B)}; Z_{1}^{\beta_1 n} \big|\C\big) + \frac{1}{n} I\big( W_{d_1}^{(A)}, \ldots, W_{d_{\Kw}}^{(A)},W_{d_{\Kw+1}}^{(B)}, \ldots, W_{d_\Ka}^{(B)}; Z_{\beta_1n+1}^{(\beta_1+\beta_2) n}\big|\C\big)\nonumber\\
& \qquad\qquad + \frac{1}{n} I\big( W_{d_{\Kw+1}}^{(A)}, \ldots, W_{d_\Ka}^{(A)}; Z_{(\beta_1+\beta_2)n+1}^{n}\big|\C \big),
\end{align} 
because of the independence of the communications in the three subphases. By Lemma~\ref{lem:secure}, all of the three summands can be bounded by $\epsilon/3$ for sufficiently large $n$. This holds in particular for the  second summand because from the eavesdropper's point of view, the structure of the piggyback codebook is meaningless and for each message tuple $\big({W}_{\sec,G_\ell^{(t)}}^{(A)} ; \ \mathbf{W}_{s,G_\ell^{(t)}}^{(B)}\big)$, the codeword to transmit is chosen uniformly at random (depending on $B$ and $K_{G_\ell^{(t)}}$) over a set of $2^{n_2 r_s}$ i.i.d. random codewords where $r_s= R_{\key,2}+R_{\textnormal{bin}} = \min\{1-\delta_z, 1-\delta_s\}$. 

Since averaged over the random codebooks, the probability of decoding error and the average mutual information $\frac{1}{n} I \left(  W_1, \ldots, W_\Da; Z^n|\C \right)$ both vanish as $n\to \infty$, there must exist at least one choice of the codebooks so that for this choice $\p_e^{\Worst}$ and  $\frac{1}{n} I \left(  W_1, \ldots, W_\Da; Z^n\right)$ both vanish asymptotically. 
\bigskip

Notice that for each $t\in\{1,\ldots, \Kw-1\}$, the rate of communication satisfies
\begin{equation}
R=R^{(A)}+R^{{(B)}} =R^{(t+2)}-\epsilon,
\end{equation}
and weak receivers require a cache memory of 
\begin{equation}
\Mw= \Da \frac{(t-1)}{\Kw}R^{(A)} + \Da \frac{t}{\Kw} R^{(B)} + {\Kw-1 \choose t} R_{\key,1} + {\Kw-1 \choose t-1} R_{\key,2}= \Ma^{(t+2)} - \Da \cdot \frac{t-1/2}{\Kw}\epsilon.
\end{equation}
Letting $\epsilon \to 0$ thus concludes the achievability proof of rate-memory pairs $(R^{(t+2)}, \Mw^{(t+2)})$, for $t=1,\ldots, \Kw-1$, as defined in \eqref{eq:R3}--\eqref{eq:M3}.

\bigskip

\subsection{Secure Cache-Aided Piggyback Coding II}~~ \label{sec:piggyback_analysisII}
Let 
\begin{equation}
\beta = \frac{\Kw(\delta_z-\delta_s)}{\Ks\min\{ 1-\delta_z,1-\delta_w\}+\Kw(\delta_z-\delta_s)}
\end{equation}
Subphase~1 is of length $\beta n$ and Subphase~2 of length $(1-\beta) n$. 

Submessages $\{W_d^{(A)}\}$ are of rate
\begin{equation}
R^{(A)} = \frac{(\delta_z-\delta_s)\min\{ 1-\delta_z,1-\delta_w\}}{\Ks\min\{ 1-\delta_z,1-\delta_w\}+\Kw(\delta_z-\delta_s)} -\epsilon/2 ,
\end{equation}
and submessages $\{W_d^{(B)}\}$ are of rate
\begin{equation}
R^{(B)} = \frac{\Kw(\delta_z-\delta_s)^2}{\Ks[\Ks \min\{ 1-\delta_z,1-\delta_w\} +\Kw(\delta_z-\delta_s)]} -\epsilon/2.
\end{equation}
The key rate is chosen as 
\begin{equation}
R_{\key} =R^{(A)}.
\end{equation}

We analyze the scheme averaged over the random choice of the codebooks $\C$. 
Probability of decoding error at the weak receivers in Subphase 1 vanishes asymptotically as $n\to \infty$, because
\begin{equation}
\frac{\Kw R^{(A)}}{(1-\delta_w)} < \beta.
\end{equation}
Probability of decoding error at the strong receivers in Subphase 1 vanishes asymptotically, because
\begin{equation} 
\frac{\Kw R^{(A)}+\Ks R^{(B)} }{(1-\delta_s)} < \beta.
\end{equation} 
Probability of decoding error in Subphase 2 vanishes asymptotically, because 
\begin{equation}
\frac{\Ks  R^{(A)}}{(\delta_z-\delta_s)} < (1-\beta).
\end{equation}
As a consequence, also the overall probability of error (averaged over the random choice of the codebooks $\mathcal{C}$) vanishes as $n \to \infty$. Following a similar secrecy analysis as in the previous Subsection~\ref{sec:piggyback_analysis}, it can be shown that also the averaged mutual information $\frac{1}{n}I(W_1,\ldots, W_{\Da}; Z^n|\mathcal{C})$  vanishes as $n\to \infty$. This implies that there must exist at least one choice of the codebooks so that for this choice $\p_e^{\Worst}\to 0$ and $\frac{1}{n}I(W_1,\ldots, W_{\Da}; Z^n)\to 0$ as $n\to \infty$. 

Notice now that 
\begin{equation}
R=R^{(A)}+R^{(B)} = \frac{\delta_z-\delta_s}{\Ks} -\epsilon.
\end{equation}
Moreover, the required cache size at each weak receiver is
\begin{IEEEeqnarray}{rCl}
\Mw= \Da R^{(B)} + R_{\key} = \Mw^{(\Kw+2)}- \Da \epsilon/2.
\end{IEEEeqnarray}
Taking $\epsilon \to 0$, this concludes the proof of  achievability of the rate-memory pair $R^{(\Kw+2)}, \Mw^{(\Kw+2)}$ in \eqref{eq:R4} and \eqref{eq:M4}.

\section{Coding Schemes when All Receivers have Cache Memories} \label{sec:proofLB_Aissgn}

\subsection{Cached Keys}~~ \label{sec:allkeys_analysis}

Fix a small $\epsilon>0$ and let 
\begin{equation} \label{eq:def_beta}
\beta=\frac{\Kw(1-\delta_s)}{\Kw(1-\delta_s)+\Ks (1-\delta_w)},
\end{equation}
and the message rate be 
\begin{equation}
R = \frac{\beta(1-\delta_w)}{\Kw} -\epsilon.
\end{equation}
Notice that by the choice of the parameter $\beta$, also
\begin{equation}
R = \frac{(1-\beta)(1-\delta_s)}{\Ks}-\epsilon.
\end{equation}
Choose  key rates
\begin{IEEEeqnarray}{rCl}
	R_{\key,w}  &=& \frac{\beta}{\Kw} \cdot  \min \big\{ 1-\delta_z,1-\delta_w\big\},\label{eq:Mwcd1} \\
	R_{\key,s}   & =& \frac{1-\beta}{\Ks } \cdot  \min \big\{ 1-\delta_z,1-\delta_s\big\},\label{eq:Mwcd2}
\end{IEEEeqnarray}
and let secret keys $K_1,\ldots, K_{\Kw}$ be of rate $R_{\key,w}$ and secret keys $K_{\Kw+1},\ldots, K_{\Ka}$ be of rate $R_{\key,s}$.

The cache requirement at weak receivers is 
\begin{equation}
\Mw= R_{\key,w} = \frac{(1-\delta_s)}{\Kw(1-\delta_s)+\Ks (1-\delta_w)} \min \big\{ 1-\delta_z,1-\delta_w\big\}
\end{equation}
and the
cache requirement at strong receivers is 
\begin{equation}
\Ms= R_{\key,s} = \frac{(1-\delta_w)}{\Kw(1-\delta_s)+\Ks (1-\delta_w)} \min \big\{ 1-\delta_z,1-\delta_s\big\}
\end{equation}

Analysis of the probability of error and of the secrecy constraint \eqref{eq:jointSecrecy} are standard and omitted. Letting  $\epsilon \to 0$ proves achievability of  the rate-memory par $(\tilde{R}^{(1)}, \tilde{\Ma}_w^{(1)}, \tilde{\Ma}_s^{(1)})$.

\bigskip
\subsection{Secure Cache-Aided Piggyback Coding with Keys at All Receivers}\label{sec:piggyback_allkeys_analysis}

The scheme follows closely the secure cache-aided piggyback coding I in Subsection~\ref{sec:piggyback_analysis}, but for a different choice of rates and time-sharings and with additional secret keys for Subphases 2 and 3 which render wiretap binning useless. 

The time-sharing parameters used here are :
\begin{subequations}\label{eq:betasa}
	\begin{IEEEeqnarray}{rCl}
		\beta_1& = &\frac{(\Kw-t+1)(\Kw-t)(1-\delta_s)(\delta_w-\delta_s)}{(\Kw-t+1)(1-\delta_s)\big[(\Kw-t)(\delta_w-\delta_s) + \Ks (t+1)(1-\delta_w)\big] + \Ks ^2t(t+1)(1-\delta_w)^2}.  \\
		\beta_2& = & \frac{\Ks(\Kw-t+1) (t+1)(1-\delta_w)(1-\delta_s)}{(\Kw-t+1)(1-\delta_s)\big[(\Kw-t)(\delta_w-\delta_s) + \Ks (t+1)(1-\delta_w)\big] + \Ks ^2t(t+1)(1-\delta_w)^2} \\
		\beta_3& = &  \frac{\Ks^2 t(t+1)(1-\delta_w)^2}{(\Kw-t+1)(1-\delta_s)\big[(\Kw-t)(\delta_w-\delta_s) + \Ks (t+1)(1-\delta_w)\big] + \Ks ^2t(t+1)(1-\delta_w)^2}.
	\end{IEEEeqnarray}
\end{subequations}
and the rates are 
\begin{equation} 
R^{(A)} = \frac{\Ks t(t+1)(1-\delta_w)^2(1-\delta_s)}{(\Kw-t+1)(1-\delta_s)\big[(\Kw-t)(\delta_w-\delta_s) + \Ks (t+1)(1-\delta_w)\big] + \Ks ^2t(t+1)(1-\delta_w)^2}-\epsilon/2
\end{equation} 
and
\begin{equation}
R^{(B)} = \frac{(\Kw-t+1)(t+1)(1-\delta_w)(1-\delta_s)(\delta_w-\delta_s)}{(\Kw-t+1)(1-\delta_s)\big[(\Kw-t)(\delta_w-\delta_s) + \Ks (t+1)(1-\delta_w)\big] + \Ks ^2t(t+1)(1-\delta_w)^2}-\epsilon/2.
\end{equation} 

The following keys are generated:
\begin{itemize}
	\item For each $\ell \in \left\lbrace 1,\ldots,\binom{\Kw}{t+1} \right\rbrace$, generate an independent secret key $K_{G_\ell^{(t+1)}}$ of rate 
	\begin{equation} 
	R_{\key,1} = {\binom{\Kw}{t+1}}^{-1} \cdot \beta_1 \cdot \min \left\lbrace 1-\delta_z,1-\delta_w \right\rbrace.
	\end{equation} 
	\item For each $\ell \in \left\lbrace 1,\ldots,\binom{\Kw}{t} \right\rbrace$, generate an independent secret key $K_{G_\ell^{(t)}}$ of rate
	\begin{equation}
	R_{\key,2} = {\binom{\Kw}{t}}^{-1} \cdot \beta_2 \cdot \min \left\lbrace 1-\delta_z,1-\delta_w \right\rbrace
	\end{equation}
	and  $\Ks$ independent secret  keys $K_{\Kw+1,G_\ell^{(t)}}, \ldots, K_{\Ka,G_\ell^{(t)}}$ of rate 
	\begin{equation}
	R_{\key,3} = {\binom{\Kw}{t}}^{-1} \cdot \frac{\beta_2}{\Ks } \cdot \min \left\lbrace \max \left\lbrace 0,\delta_w-\delta_z \right\rbrace, \delta_w-\delta_s \right\rbrace.
	\end{equation}
	\item For $j \in \K_s $, generate an independent secret key $K_{j}$ of rate 
	\begin{equation} \label{eq:key4}
	R_{\key,4} = \frac{\beta_3}{\Ks } \cdot \min \left\lbrace 1-\delta_z,1-\delta_s \right\rbrace.
	\end{equation}
\end{itemize}
\bigskip

\underline{\textit{Placement phase:}}
For each weak receiver $i \in \K_w$, cache  
\begin{align}\label{eq:cacheWeak1}
V_i & = \left\lbrace W_{d,G_\ell^{(t-1)}}^{(A)} \colon d \in \D \mbox{ and } i \in G_\ell^{(t-1)}  \right\rbrace ~\bigcup~ \left\lbrace W_{d,G_\ell^{(t)}}^{(B)} \colon d \in \D \mbox{ and } i \in G_\ell^{(t)}  \right\rbrace ~\bigcup~ \left\lbrace K_{G_\ell^{(t+1)}} \colon i \in G_\ell^{(t+1)}   \right\rbrace \nonumber \\		
& \qquad\qquad   \bigcup~ \left\lbrace K_{G_\ell^{(t)}} \colon i \in G_\ell^{(t)} \right\rbrace ~\bigcup~ \left\lbrace K_{j,G_\ell^{(t)}} \colon i \in G_\ell^{(t)} \mbox{ and } j \in \K_s  \right\rbrace.
\end{align}

For each strong receiver $j \in \K_s $, cache 
\begin{equation}\label{eq:cacheStrong1}
V_j = K_j  ~\bigcup~ \left\lbrace K_{j,G_\ell^{(t)}} \colon \ell \in \left\lbrace 1,\ldots,\binom{\Kw}{t} \right\rbrace \right\rbrace.
\end{equation}

\bigskip

\underline{\textit{Delivery phase:}} Is divided into three subphases of lengths $n_1=\beta_1n$, $n_2=\beta_2n$, and $n_3=\beta_3n$, where $\beta_1,\beta_2,$ and $\beta_3$ are defined in \eqref{eq:betasa}. \bigskip

\textit{Subphase 1:} Is as described in Subsection~\ref{sec:piggyback_analysis}. \bigskip

\textit{Subphase 2:}  Similar to Subsection~\ref{sec:piggyback_analysis}, but with extra keys. As in 
 Subsection~\ref{sec:piggyback_analysis}, this subphase is split into ${\Kw \choose t}$ equally-long periods that we label by all $t$-user subsets of $\K_w$ i.e., by  $G_1^{(t)},\ldots, G_{{\Kw \choose t}}^{(t)}$. For the transmission in Period $G_{\ell}^{(t)}$, for $\ell\in\{1,\ldots, {\Kw \choose t}\}$,  a standard piggyback codebook  
 \begin{equation}
\C_{\textnormal{pg},G_\ell^{(t)}} =\Big\{ x_{G_\ell^{(t)}}^{n_2}\big( \ell_{\textnormal{row}} ; \ \ell_{\textnormal{col}}\big)\colon \;  \ell_{\textnormal{row}}  \in \big\{1,\ldots, \big\lfloor 2^{nr^{(A)}}\big\rfloor \big\},\ell_{\textnormal{col}}  \in \big\{1,\ldots, \big\lfloor 2^{nr^{(B)}} \big\rfloor\big\} \Big\},
\end{equation} 
is generated by drawing all entries i.i.d.  Bernoulli-$1/2$. The codebooks are revealed to the transmitter and all receivers, and the  rates $r^{(A)}$ and $r^{(B)}$ are chosen as 
\begin{IEEEeqnarray}{rCl}
	r^{(A)} &=&R^{(A)} {\Kw \choose t-1}^{-1}\\
	r^{(B)}&=&R^{(B)} {\Kw \choose t}^{-1}.
	\end{IEEEeqnarray}
	
At the beginning of Period $G_{\ell}^{(t)}$, the transmitter computes
	\begin{equation} 
	W_{\sec,w,G_\ell^{(t)}}^{(A)} = \enc{\bigoplus\limits_{i \in G_\ell^{(t)}} W_{d_i,G_\ell^{(t)}\setminus \{i\}}^{(A)}}{\ K_{G_\ell^{(t)}}}.
	\end{equation} 
and 	\begin{equation} 
\textbf{W}_{\sec,s,G_\ell^{(t)}}^{(B)}:=\Big(	\textsf{sec}\big(W_{d_{\Kw+1},G_\ell^{(t)}}^{(B)}, \ K_{{\Kw+1},G_\ell^{(t)} }\big),\;  \ldots, \; \textsf{sec}\big(W_{d_{\Ka},G_\ell^{(t)}}^{(B)}, 
	\ K_{{\Ka},G_\ell^{(t)}}\big)\Big),
	\end{equation} 
	and sends the codeword
\begin{equation}
x_{G_\ell^{(t)}}^{n_2}\big(W_{\sec,w,G_\ell^{(t)}}^{(A)} ; \ \mathbf{W}_{\sec,s,G_\ell^{(t)}}^{(B)} \ \big)
\end{equation} 
over the channel. Decoding of Period $G_\ell^{(t)}$ is performed as follows.
All weak receivers $i\in G_\ell^{(t)}$ compute the message tuple to the strong receivers $\mathbf{W}_{\sec,s,G_\ell^{(t)}}^{(B)}$ and perform their decoding steps as described in Subsection~\ref{sec:piggyback_analysis}. 
Strong receivers decode both secured message tuples transmitted in this period. Any given strong receiver~$j\in\K_s$ can then recover its desired message
\begin{equation}
W_{d_j,G_{\ell}^{(t)}}^{(B)} 
\end{equation}
using the secret key $K_{j, G_{\ell}^{(t)}}$ 
stored in its cache memory.

At the end of the subphase, each weak receiver $i\in \K_w$ assembles 
the parts of submessage $W_{d_i}^{(A)}$ that it decoded or it has stored in its cache memory, and forms the guess $\hat{W}_{d_i}^{(A)}$. Similarly, each strong receiver $j\in\K_s$ assembles the  decoded parts and  forms $\hat{W}_{d_j}^{(B)}$.
\bigskip

\textit{Subphase 3:} A capacity-achieving code is used to send the secured message tuple
\begin{equation}
\textsf{sec}\big(W_{d_{\Kw+1}}^{(A)}, K_{\Kw+1}\big),\ldots, \textsf{sec}\big(W_{d_{\Ka}}^{(A)},  K_{\Ka}\big)
\end{equation}  to the strong receivers~$\Kw+1,\ldots, \Ka$. Each strong receiver $j\in\K_s$ obtains its desired guess $\hat{W}_{d_j}^{(A)}$ with the help of the secret key $K_j$ stored in its cache memory.\bigskip

At the end of the  entire delivery phase, each Receiver~$k\in\K$ assembles its guesses of $W_{d_k}^{(A)}$ and $W_{d_k}^{(B)}$ to produce the final guess $\hat{W}_{k}$.
\bigskip

\underline{\textit{Analysis:}} The scheme is analyzed by averaging over the random code constructions $\C$. We first verify that in each of the three subphases the average probability of decoding error tends to 0 as $n\to \infty$. Only weak receivers decode during the first subphase. Probability of decoding error in Subphase 1 vanishes asymptotically as $n \to \infty$, because 
\begin{equation}
\frac{\binom{\Kw}{t+1}{\binom{\Kw}{t}}^{-1}R^{(B)}}{(1-\delta_w)} = \frac{\frac{\Kw-t}{t+1}	R^{(B)}}{(1-\delta_w)} < \beta_1.
\end{equation}  
Only strong receivers decode during the  third subphase. Probability of decoding error in Subphase 3 vanishes asymptotically, because 
\begin{equation}
\frac{\Ks  R^{(A)}}{(1-\delta_s)} < \beta_3.
\end{equation}
In Subphase~2, probability of decoding error at weak receivers vanishes asymptotically, because
\begin{equation}
\frac{\frac{\Kw-t+1}{t}R^{(A)}}{(1-\delta_w)} < \beta_2.
\end{equation}
Probability of decoding error at strong receivers in Subphase 2 vanishes asymptotically, because
\begin{equation} 
\frac{\frac{\Kw-t+1}{t}R^{(A)}+\Ks R^{(B)}}{(1-\delta_s)} < \beta_2.
\end{equation} 
Whenever the decodings in all subphases are successful, all receivers  guess their desired messages correctly. As a consequence, also the average of the overall probability of error  vanishes as $n\to \infty$.

We  analyze the secrecy constraint. For fixed blocklength $n$:
\begin{align}
\frac{1}{n} I \big(  W_1, \ldots, W_\Da; Z^n \big |\mathcal{C} \big) & = \frac{1}{n} I\big( W_{d_1}^{(B)}, \ldots, W_{d_{\Kw}}^{(B)}; Z_{1}^{\beta_1 n}  \big |\mathcal{C} \big) + \frac{1}{n} I \big( W_{d_1}^{(A)}, \ldots, W_{d_{\Kw}}^{(A)},W_{d_{\Kw+1}}^{(B)}, \ldots, W_{d_\Ka}^{(B)}; Z_{\beta_1n+1}^{(\beta_1+\beta_2) n} \big |\mathcal{C} \big)\nonumber\\
& \qquad\qquad + \frac{1}{n} I \big( W_{d_{\Kw+1}}^{(A)}, \ldots, W_{d_\Ka}^{(A)}; Z_{(\beta_1+\beta_2)n+1}^{n}  \big |\mathcal{C} \big),
\end{align} 
because of the independence of the communications in the three subphases.  The sizes of the secret keys have been chosen so that the codewords to be sent in each of the three subphases  are chosen uniformly at random over a subset of the random codebooks that are  equal to the minimum between $(1-\delta_z)$ and the rates of communication. By  Lemma~\ref{lem:secure} this proves that $\frac{1}{n} I \big(  W_1, \ldots, W_\Da; Z^n \big |\mathcal{C} \big)$ tends to $0$ as $n\to \infty$. 
\medskip

We can conclude that there must exist at least one choice of the codebooks so that for this choice both $\p_e^{\Worst} \to 0$ and $\frac{1}{n} I \big(  W_1, \ldots, W_\Da; Z^n  \big) $ vanish asymptotically. 

For each choice of the parameter $t\in\{1,\ldots, \Kw-1\}$:
\begin{equation}
R=R^{(A)}+R^{(B)}=\tilde{R}^{(t+1)}-\epsilon.
\end{equation}
Moreover, each weak receiver requires a cache size of
\begin{align} \label{eq:pair4_Mw}
\Mw & = \frac{\Da\big[tR^{(A)} +(t-1)R^{(B)}\big] }{\Kw} + {\Kw-1 \choose t}R_{\key,1} + {\Kw -1 \choose t-1} R_{\key,2} + {\Kw -1 \choose t-1} \Ks R_{\key,3}=
\tilde{\Ma}_w^{(t+1)} - \frac{ \Da (t-1/2)}{\Kw} \epsilon,
\end{align}
and each strong receiver a cache size of 
\begin{equation} \label{eq:pair4_Ms}
\Ms = R_{\key,4} + {\Kw \choose t} R_{\key,3} 
= \tilde{\Ma}_s^{(t+1)}.
\end{equation}
Taking $\epsilon \to 0$ thus proves achievability of the rate-memory triples $( \tilde{R}^{(t+1)},\tilde{\Ma}_w^{(t+1)},\tilde{\Ma}_s^{(t+1)})$, for $t\in\{1,\ldots, \Kw-1\}$, in \eqref{eq:R3_Assign}--\eqref{eq:Ms3_Assign}.

\bigskip

	\subsection{Symmetric Secure Cache-Aided Piggyback Coding} ~~ \label{sec:piggyback_symmetric_analysis}
		
	\indent Fix $\epsilon>0$ and define the time-sharing parameters 	\begin{subequations}
		\begin{IEEEeqnarray}{rCl}
			\beta_1 &:= &\frac{\Kw(\Kw-t_w)(t_s+1)(1-\delta_s)^2}{\Kw(\Kw-t_w)(t_s+1)(1-\delta_s)^2 + \Ks (t_w+1)(1-\delta_w)\big[ (\Ks -t_s)(1-\delta_w) + \Kw(t_s+1)(1-\delta_s) \big]} \\
			\beta_2 & :=&\frac{\Kw\Ks(t_w+1)(t_s+1)(1-\delta_w)(1-\delta_s)}{\Kw(\Kw-t_w)(t_s+1)(1-\delta_s)^2 + \Ks (t_w+1)(1-\delta_w)\big[ (\Ks -t_s)(1-\delta_w) + \Kw(t_s+1)(1-\delta_s) \big]}\\
			\beta_3 & :=& \frac{\Ks (\Ks-t_s)(t_w+1)(1-\delta_w)^2}{\Kw(\Kw-t_w)(t_s+1)(1-\delta_s)^2 + \Ks (t_w+1)(1-\delta_w)\big[ (\Ks -t_s)(1-\delta_w) + \Kw(t_s+1)(1-\delta_s) \big]}. 
		\end{IEEEeqnarray}
	\end{subequations}
	Notice that $\beta_1+\beta_2+\beta_3=1$. 
	
	Choose two positive integers $t_w \in \{1,\ldots,\Kw\}$ and $t_s \in \{1,\ldots,\Ks \}$.
		
	\textit{Message splitting:} Divide each message into two submessages,
	\begin{equation}
	W_d = \big[ W_d^{(A)}, W_d^{(B)}\big], \qquad d \in \D,
	\end{equation}
	so that the submessages are of rates
	\begin{equation}
	R^{(A)} =  \frac{\Kw(t_w+1)(t_s+1)(1-\delta_w)(1-\delta_s)^2}{\Kw(\Kw-t_w)(t_s+1)(1-\delta_s)^2 + \Ks (t_w+1)(1-\delta_w)\big[ (\Ks -t_s)(1-\delta_w) + \Kw(t_s+1)(1-\delta_s) \big]} -\epsilon/2
	\end{equation}
	and 
		\begin{equation} \label{eq:Rtwts}
	R^{(B)} = \frac{\Ks (t_w+1)(t_s+1)(1-\delta_w)^2(1-\delta_s)}{\Kw(\Kw-t_w)(t_s+1)(1-\delta_s)^2 + \Ks (t_w+1)(1-\delta_w)\big[ (\Ks -t_s)(1-\delta_w) + \Kw(t_s+1)(1-\delta_s) \big]}-\epsilon/2.
	\end{equation}
	\vspace{2mm}
	
	Denote the $\binom{\Kw}{t_w}$ subsets of $\{1,\ldots,\Kw\}$ of size $t_w$ by $G_1^{(t_w)}, \ldots, G_{\binom{\Kw}{t_w}}^{(t_w)}$ and the $\binom{\Ks }{t_s}$ subsets of $\{1,\ldots,\Ks \}$ of size $t_s$ by $G_1^{(t_s)}, \ldots, G_{\binom{\Ks }{t_s}}^{(t_s)}$.
	Divide every message $W_d^{(A)}$ into $\binom{\Kw}{t_w}$ submessages and every message $W_d^{(B)}$ into $\binom{\Ks }{t_s}$ submessages:
	\begin{subequations}
		\begin{align}
		W_d^{(A)} & = \left\lbrace W_{d,G_\ell^{(t_w)}}^{(A)} : \quad \ell \in \left\lbrace 1,\ldots,\binom{\Kw}{t_w} \right\rbrace  \right\rbrace, \\
		W_d^{(B)} &  = \left\lbrace W_{d,G_\ell^{(t_s)}}^{(B)} : \quad \ell \in \left\lbrace 1,\ldots,\binom{\Ks }{t_s} \right\rbrace  \right\rbrace.
		\end{align}
	\end{subequations}
	
	\textit{Key generation:}
	\begin{itemize}
		\item For each $\ell \in \left\lbrace 1,\ldots,\binom{\Kw}{t_w+1} \right\rbrace$, generate an independent random key $K_{G_\ell^{(t_w+1)}}$ of rate 
		$$R_{\key,1} = {\binom{\Kw}{t_w+1}}^{-1} \cdot \beta_1 \min \left\lbrace 1-\delta_z,1-\delta_w \right\rbrace.$$
		\item For each $\ell \in \left\lbrace 1,\ldots,\binom{\Ks }{t_s+1} \right\rbrace$, generate an independent random key $K_{G_\ell^{(t_s+1)}}$ of rate 
		$$R_{\key,2} = {\binom{\Ks }{t_s+1}}^{-1} \cdot \beta_3 \min \left\lbrace 1-\delta_z,1-\delta_s \right\rbrace.$$ 
		\item  For each $i \in \K_w$ and $j \in \K_s $, generate  an independent random key $K_{w,\{i,j\}}$ of rate
		$$R_{\key,3} = \frac{\beta_2}{\Kw \Ks } \min \left\lbrace 1-\delta_z,1-\delta_w\right\rbrace,$$
		and  an independent random key $K_{s,\{i,j\}}$ of rate
		$$ R_{\key,4} = \frac{\beta_2}{\Kw \Ks } \min \left\lbrace (\delta_w-\delta_z)^+,1-\delta_s\right\rbrace.$$ 
	\end{itemize} \bigskip
	
	\underline{\textit{Placement phase:}}
	Place the cache  contents as shown in the following table. 
	\begin{figure}[h!] \centering
		\begin{tikzpicture}
		\node at (2,4.4) {Cache at weak receiver~$i$};
		\draw[rounded corners=7pt,thick] (0,0) rectangle (4,4);
		\node at (2,3.3) {$\left\lbrace \left\lbrace W_{d,G_\ell^{(t_w)}}^{(A)} \right\rbrace_{i \in G_\ell^{(t_w)}} \right\rbrace_{\!d=1}^{\!\Da}$};	
		\node at (2,2) {$\left\lbrace K_{G_\ell^{(t_w+1)}} \right\rbrace_{i \in G_\ell^{(t_w+1)}}$};	
		\node at (2,1.1) {$  K_{w,\{i,\Kw+1\}}, \ldots K_{w,\{i,\Ka\}}$};	
		\node at (2,0.4) {$  K_{s,\{i,\Kw+1\}}, \ldots K_{s,\{i,\Ka\}}$};	
		
		\node at (8,4.4) {Cache at strong receiver~$j$};
		\draw[rounded corners=7pt,thick] (6,0) rectangle (10,4);
		\node at (8,3.3) {$\left\lbrace \left\lbrace W_{d,G_\ell^{(t_s)}}^{(B)} \right\rbrace_{j \in G_\ell^{(t_s)}} \right\rbrace_{\!d=1}^{\!\Da}$};	
		\node at (8,2) {$\left\lbrace K_{G_\ell^{(t_s+1)}} \right\rbrace_{j \in G_\ell^{(t_s+1)}}$};	
		\node at (8,1.1) {$  K_{w,\{1,j\}}, \ldots K_{w,\{\Kw,j\}}$};	
		\node at (8,0.4) {$  K_{s,\{1,j\}}, \ldots K_{s,\{\Kw,j\}}$};	
		
		\end{tikzpicture} 
	\end{figure}

	\bigskip
	
\underline{\textit{Delivery phase:}} The delivery phase is divided into three subphases of lengths $\beta_1n$, $\beta_2n$ and $\beta_3n$.\\

\textit{Subphase 1:} This phase conveys to each weak receiver $i \in \K_w$, the parts of $W_{d_i}^{(A)}$ that are not stored in its cache memory.

Time-sharing   is applied  over $\binom{\Kw}{t_w+1}$ equally-long periods of length $n_1=\beta_1 n {\Kw \choose t_w+1}^{-1}$. Label the periods $G_1^{(t+1)}, \ldots, G_{{\Kw \choose t+1}}^{(t+1)}$. 

 In Period $G_\ell^{(t+1)}$, for  $\ell \in \big\lbrace 1, \ldots,\binom{\Kw}{t_w+1} \big\rbrace$, the secured XOR-message
		\begin{equation}\label{eq:secXOR}
		\textsf{sec}\bigg( 
		\bigoplus\limits_{i \in G_\ell^{(t_w+1)}} W_{d_i,G_\ell^{(t_w+1)}\setminus \{i\}}^{(A)},\; K_{G_\ell^{(t_w+1)}}\bigg)
		\end{equation} 
		is sent to 
		 the subset of receivers $G_\ell^{(t_w+1)}$. Each  receiver $i\in G_\ell^{(t_w+1)}$ retrieves the content
		 \begin{equation}\label{eq:cach}
		 K_{G_\ell^{(t_w+1)}}  \; \textnormal{and} \; \Big\{W_{d_k,G_\ell^{(t_w+1)}\backslash\{k\}}^{(A)} \Big\}_{k\in G_\ell^{(t_w+1)}\backslash\{i\}}
		 \end{equation}
		 from its cache memory. It then decodes the secured message in \eqref{eq:secXOR}, and with the retrieved cache content \eqref{eq:cach} it   
		  recovers the desired message $W_{d_i,G_\ell^{(t_w+1)}\backslash\{i\}}^{(A)}$.

		\bigskip
		
		 \textit{Subphase 2:} Submessages $W_{d_1}^{(B)}, \ldots, W_{d_{\Kw}}^{(B)}$ are sent to weak receivers~$1, \ldots, \Kw$ and submessages $W_{d_{\Kw+1}}^{(A)}, \ldots, W_{d_{\Ka}}^{(A)}$ to strong receivers~$\Kw+1,\ldots,\Ka$.
		 
		 Time-sharing is applied over $\Kw\cdot \Ks$ periods, each of length $n_2=\frac{\beta_2 n}{\Kw\Ks}$. 
		 The periods are labeled $\{i,j\}$, for $i\in\K_w$ and $j\in\K_s$. Divide each submessage $W_{d}^{(B)}$ into $\Ks $ parts 
		 \begin{subequations}
		 	\begin{equation}
		 	W_{d}^{(B)} = \big( W_{d,\Kw+1}^{(B)}, \ldots,W_{d,\Ka}^{(B)} \big)
		 	\end{equation}
		 so that each part is of  equal rate $r^{(B)}=R^{(B)}/\Ks$  and  for all $j\in\K_s$ and $d\in\D$, part $W_{d,j}^{(B)}$ is stored in strong receiver~$j$'s cache memory. 
		 	Similarly, divide each submessage $W_{d}^{(A)}$ into $\Kw$ parts 
		 	\begin{equation}
		 	W_{d}^{(A)}  =\big( W_{d,1}^{(A)}, \ldots,W_{d,\Kw}^{(A)} \big), 
		 	\end{equation}
		 	of equal rate  $r^{(A)}=R^{(A)}/\Kw$ so that for each   $i\in\K_w$ and $d\in\D$, part $W_{d,i}^{(A)}$ is stored in weak receiver~$i$'s cache memory. 
		 \end{subequations}

	 We describe the encoding and decoding operations in a period $\{i,j\}$, for $i\in\Kw$ and $j\in\Ks$. In this period, the transmitter sends the codeword
		 \begin{equation}
		 x_{\{i,j\}}^{n_2}\Big( 	\textsf{sec}\big(W_{d_i,j}^{(B)}, K_{w,\{i,j\}} \big) ; \ 	\textsf{sec}\big(W_{d_j,i}^{(A)}, K_{s,\{i,j\}} \big)\Big)
		 \end{equation} 
		 over the channel. At the end of the period, Receiver~$i$ first retrieves  the cache content 
		 \begin{equation}
	K_{s,\{i,j\}}, \ W_{d_j,i}^{(A)},
		 \end{equation}
 and   forms the  subcodebook 
		 	\begin{equation}
		 	\C_{\textnormal{pg},\{i,j\}}^{\textnormal{row}}\big(	\textsf{sec}\big(W_{d_j,i}^{(A)}, K_{s,\{i,j\}} \big) \big) =\Big\{ x_{\{i,j\}}^{n_2}\big( \ell_{\textnormal{row}} ; \ \textsf{sec}\big(W_{d_j,i}^{(A)}, K_{s,\{i,j\}} \big) \big)\colon \; \ell_{\textnormal{row}}  \in \big\{1,\ldots, \big\lfloor 2^{nr^{(B)}} \big\rfloor\big\} \Big\}. 
		 	\end{equation} 
		 Based on this subcodebook, it  then decodes the secured message  $\textsf{sec}\big(W_{d_i,j}^{(B)}, K_{w,\{i,j\}} \big),$ and  recovers its desired message $W_{d_i,j}^{(B)}$ with the secret key $ K_{w,\{i,j\}}$ stored in its cache memory.
 Receiver~$j$ proceeds analogously, except that it retrieves the cache content
		 		 \begin{equation}
		 		 K_{w,\{i,j\}}, \ W_{d_i,j}^{(B)},
		 		 \end{equation}
		 		 and forms  the subcodebook 
		 	\begin{equation}
		 	\C_{\textnormal{pg},\{i,j\}}^{\textnormal{col}}\big( \textsf{sec}\big(W_{d_i,j}^{(B)}, K_{w,\{i,j\}}\big) \big) =\Big\{ x_{\{i,j\}}^{n_2}\big( \textsf{sec}\big(W_{d_i,j}^{(B)}, K_{w,\{i,j\}} \big) ;\  \ell_{\textnormal{col}}\big)\colon \; \ell_{\textnormal{col}}  \in \big\{1,\ldots, \big\lfloor 2^{nr^{(A)}} \big\rfloor\big\} \Big\}.
		 	\end{equation} 
		It then decodes the secured message  $\textsf{sec}\big(W_{d_j,i}^{(A)}, K_{s,\{i,j\}} \big),$ and  recovers $W_{d_j,i}^{(A)}$ using  the secret key $ K_{s,\{i,j\}}$ from its cache memory.

	\bigskip
		
		 \textit{Subphase 3:} 
		 This subphase conveys to each strong receiver $j \in \K_s$, the parts of  submessage $W_{d_j}^{(B)}$ that are not stored in its cache memory. Encoding/decoding operations are obtained from the encoding/decoding operations for Subphase~1 if in the latter the subscript $w$ is replaced by the subscript $s$ and the superscript $(A)$ is replaced by the superscript $(B)$.

\bigskip
	
\underline{	\textit{Analysis:}} We analyze the average (over the random choice of the codebooks) probability of decoding error and the average leakage $\frac{1}{n} I(W_1,\ldots, W_{\Da};Z^n| \mathcal{C})$. 

 We start by  showing that the probability of decoding error vanishes in each of the three subphases. Only weak receivers perform decoding operations in the first subphase. Probability of error in this subphase thus tends to 0 as $n\to \infty$, because 
	\begin{equation}  \label{eq:pair5_phase1}
	\frac{\frac{\Kw-t_w}{t_w+1}	R^{(A)}}{(1-\delta_w)} < \beta_1
	\end{equation}	
In the second subphase, both weak and strong receivers perform decoding operations. The probability of decoding error at weak receivers tends to 0 as $n\to\infty$, because 
	\begin{equation} \label{eq:21}
	 \frac{\Kw R^{(B)}}{(1-\delta_w)} < \beta_2
	\end{equation}
	The probability of decoding error at strong receivers tends to 0 as $n\to\infty$, because 
		\begin{equation} \label{eq:22}
	\frac{\Ks R^{(A)}}{(1-\delta_s)} < \beta_2.
		\end{equation}
			(Notice that by our choice of the subrates $R^{(A)}$ and $R^{(B)}$, the two decoding constraints of Subphase~$2$, \eqref{eq:21} and \eqref{eq:22} are equally strong.)
Only strong receivers perform decoding operations in the third subphase. The probability of error of these decoding operations tends to 0 as $n\to \infty$, because 
	\begin{equation}
	\frac{\frac{\Ks -t_s}{t_s+1} R^{(B)}}{(1-\delta_s)} < \beta_3
	\end{equation}
Since all receivers correctly recover their demanded messages when all decoding operations tend to 0, when averaged over the random code construction, the total probability of error tends to 0 as $n\to \infty$. 

Communication is secured because the secret keys have been chosen sufficiently long. In fact, following the arguments given in the previous Subsection~\ref{sec:piggyback_allkeys_analysis}, it can be shown that the average leakage term $\frac{1}{n}I(W_1,\ldots, W_\Da;Z^n|\mathcal{C})$ tends to 0 as $n\to\infty$. By standard arguments, it can then be concluded that there must exist a choice of the codebooks so that for this choice the probability of error $\p_e^{\Worst}$ and the leakage $\frac{1}{n}I(W_1,\ldots, W_\Da;Z^n)$ both vanish asymptotically as $n\to \infty$. 
	
	The presented scheme requires weak receivers to have a cache of size
	\begin{IEEEeqnarray}{rCl} \label{eq:pair5_Mw}
	\Mw  &= &\frac{\Da t_w}{\Kw} R^{(A)} + \frac{(t_w+1)\beta_1 \min \left\lbrace 1-\delta_z,1-\delta_w \right\rbrace}{\Kw} + \frac{\beta_2\min \left\lbrace1-\delta_z,2-\delta_w-\delta_s\right\rbrace}{\Kw} \nonumber \\
	&=& \tilde{\Ma}_w^{(\Kw+t_w (\Ks+1)+t_s)}-  \frac{\Da t_w}{2\Kw} \epsilon,
	\end{IEEEeqnarray}
	and strong receivers a cache of size 
	\begin{IEEEeqnarray}{rCl} \label{eq:pair5_Ms}
	\Ms   &=& \frac{\Da t_s}{\Ks } R^{(B)} + \frac{(t_s+1)\beta_3 \min \left\lbrace 1-\delta_z,1-\delta_s \right\rbrace}{\Ks } + \frac{\beta_2\min \left\lbrace1-\delta_z,2-\delta_w-\delta_s\right\rbrace}{\Ks }\nonumber \\
	&=&\tilde{\Ma}_s^{(\Kw+t_w (\Ks+1)+t_s)}- \frac{\Da t_s}{2\Ks }\epsilon.
	\end{IEEEeqnarray} 
	The rate of the messages is 
	\begin{equation}
	R=R^{(A)}+R^{(B)}= \tilde{R}^{(\Kw+t_w (\Ks+1)+t_s)} - \epsilon. 
	\end{equation} 
	
	Thus, letting $\epsilon \to 0$ establishes achievability of the desired rate-memory triples  in \eqref{eq:R4_Assign}--\eqref{eq:Ms4_Assign}:
	\begin{equation}
	\Big(\tilde{R}^{(\Kw+(t_w-1)\Ks +t_s)}, \ \tilde{\Ma}_w^{(\Kw+(t_w-1)\Ks +t_s)}, \ \tilde{\Ma}_s^{(\Kw+(t_w-1)\Ks +t_s)}\Big).
	\end{equation}

\section{Summary} \label{sec:concl}

We have studied secrecy of cache-aided wiretap erasure BCs with $\Kw$ weak receivers, $\Ks $ strong receivers and one eavesdropper. We have provided a general upper bound on the secrecy capacity-memory tradeoff for the case when receivers have arbitrary erasure probabilities and arbitrary cache sizes. We have also proposed lower bounds on the secrecy capacity-memory tradeoff for different cache sizes. For some cache arrangements, e.g., for zero cache sizes at strong receivers $\Ms  =0$, our upper and lower bounds coincide for small cache sizes. For $\Ms  =0$, they also match for large cache sizes. These bounds show that the secrecy constraint can induce a significant loss in capacity compared to the standard non-secure system, especially when $\Ms  =0$. They also exhibit that in a secure system, the caching gain with small cache memories is much more important than its non-secure counterpart. This is due to the fact that secret keys can be stored in the caches, which are more useful than cached data. For larger cache sizes, data has to be stored as well and the caching gains of the secure system are similar to the gains in a standard system. We also present a lower bound on the capacity of a scenario where the cache assignment across receivers can be optimized subject to a total cache budget $\Ma_{\tot}$. The lower bound is exact for small cache budgets $\Ma_{\tot}$.

\appendices

\section{Proof of Lemma~\ref{lemma:1}} \label{app:proofUB}
	
	For each blocklength $n$, we fix caching, encoding and decoding functions as in~\eqref{eq:cachingFct}, \eqref{eq:encodingFct} and \eqref{eq:decodingFct} so that both the probability of worst-case error and the secrecy leakage satisfy:
	\[
		\p_e^{\textnormal{Worst}} \xrightarrow[n \to \infty]{} 0 \qquad \text{and} \qquad \frac{1}{n}I(W_1, \ldots, W_D; Z^n) \xrightarrow[n \to \infty]{} 0.
	\]

	We only prove the lemma for the set $\mathcal{S} = \{1,\ldots,k\}$. For other sets 	$\mathcal{S} \subseteq \{1,\ldots,\Ka\}$ the proof is similar.
		
	By Fano's inequality and because conditioning can only reduce entropy, there exists a sequence of real numbers $\{\epsilon_n\}_{n=1}^\infty$ with 
	\[
		\frac{\epsilon_n}{n} \xrightarrow[n \to \infty]{} 0,
	\]
	such that 
	\begin{equation} 
		\begin{dcases}
			& H(W_{d_1} | Y_1^n,V_1)\quad \leq \quad \frac{\epsilon_n}{2\Ka}, \nonumber \\
			& H(W_{d_2} | Y_2^n,V_1,V_2,W_{d_1}) \quad \leq \quad \frac{\epsilon_n}{2\Ka}, \nonumber\\
			 & \qquad \qquad \vdots \nonumber\\
		 & H(W_{d_k} | Y_k^n,V_1,\ldots,V_k,W_{d_1},\ldots, W_{d_{k-1}}) \quad \leq \quad \frac{\epsilon_n}{2\Ka}.
  		\end{dcases} 
	\end{equation}

Thus,
	\begin{align}
		n R & =  H(W_{d_1}) \nonumber \\
		& = H(W_{d_1}|Z^n) + I(W_{d_1};Z^n) \nonumber \\
		& \leq  H(W_{d_1}|Z^n)  + \frac{\epsilon_n}{2}  \nonumber \\
		& \leq  I(W_{d_1}; Y_1^n, V_1) - I(W_{d_1};Z^n) + H(W_{d_1}|Y_1^n,V_1)+ \frac{\epsilon_n}{2} \nonumber \\
		& \leq   I(W_{d_1}; Y_1^n, V_1) - I(W_{d_1};Z^n) + \epsilon_n \nonumber \\
		& \leq  I(W_{d_1}; Y_1^n| V_1) - I(W_{d_1};Z^n|V_1) + I(W_{d_1}; V_1|Z^n) +\epsilon_n \nonumber \\ 
		& \stackrel{(a)}{=} \sum_{i=1}^n \Big[ I(W_{d_1}; Y_{1,i}| V_1, Y_1^{i-1}) - I(W_{d_1};Z_i|V_1, Z_{i+1}^n) \Big] + n\Ma_1+ \epsilon_n \nonumber \\
		& \stackrel{(b)}{=} \sum_{i=1}^n \Big[ I(W_{d_1}; Y_{1,i}| V_1, Y_1^{i-1}) - I(W_{d_1};Z_i|V_1, Z_{i+1}^n) \Big] + n\Ma_1+ \epsilon_n \nonumber \\
		& \qquad \qquad + \sum_{i=1}^n \Big[ I(Z_{i+1}^n; Y_{1,i}| Y_1^{i-1},V_1,W_{d_1}) - I( Y_1^{i-1};Z_i|Z_{i+1}^n,V_1,W_{d_1}) \Big] \nonumber \\
		& = \sum_{i=1}^n \Big[ I(W_{d_1},Z_{i+1}^n; Y_{1,i}| V_1, Y_1^{i-1}) - I(W_{d_1}, Y_1^{i-1};Z_i|V_1, Z_{i+1}^n) \Big] + n\Ma_1+ \epsilon_n \nonumber \\
		& \stackrel{(c)}{=} \sum_{i=1}^n \Big[ I(W_{d_1},Z_{i+1}^n; Y_{1,i}| V_1, Y_1^{i-1}) - I(W_{d_1}, Y_1^{i-1};Z_i|V_1, Z_{i+1}^n) \Big] +  n\Ma_1 + \epsilon_n \nonumber \\
		&  \qquad\qquad - \sum_{i=1}^n \Big[ I(Z_{i+1}^n; Y_{1,i}| Y_1^{i-1},V_1) - I( Y_1^{i-1};Z_i|Z_{i+1}^n,V_1) \Big] \nonumber \\
		& = \sum_{i=1}^n \Big[ I(W_{d_1}; Y_{1,i}| V_1, Y_1^{i-1}, Z_{i+1}^n) - I(W_{d_1};Z_i|V_1, Y_1^{i-1} ,Z_{i+1}^n) \Big] + n\Ma_1 + \epsilon_n \nonumber \\
		& \stackrel{(d)}{\leq} \sum_{i=1}^n \Big[ I(W_{d_1}; Y_{1,i}| V_1, Y_1^{i-1},Z_{i+1}^n) -  I(W_{d_1};Z_i|V_1, Y_1^{i-1},Z_{i+1}^n) \Big]^{+} + n\Ma_1 + \epsilon_n \nonumber \\	
		& \qquad\qquad + \sum_{i=1}^n \Big[ I(V_1,Y_1^{i-1},Z_{i+1}^n; Y_{1,i}) - I(V_1, Y_1^{i-1},Z_{i+1}^n; Z_{i})\Big]^{+}  \nonumber \\	
		& \stackrel{(e)}{=} \sum_{i=1}^n \Big[ I(W_{d_1},V_1,Y_1^{i-1},Z_{i+1}^n; Y_{1,i}) - I(W_{d_1},V_1, Y_1^{i-1},Z_{i+1}^n;Z_i) \Big]^{+} + n\Ma_1 + \epsilon_n \nonumber \\	
		& \stackrel{(f)}{\leq}  \sum_{i=1}^n \Big[ I(W_{d_1},V_1,Y_1^{i-1},Z_{i+1}^n; Y_{1,i}) - I(W_{d_1},V_1, Y_1^{i-1},Z_{i+1}^n;Z_i) \Big]^{+} + n\Ma_1 + \epsilon_n \nonumber \\	
		&  \qquad\qquad + \sum_{i=1}^n \Big[ I(Y_{2}^{i-1}; Y_{1,i}|W_{d_1},V_1,Y_1^{i-1},Z_{i+1}^n) - I(Y_{2}^{i-1}; Z_i|W_{d_1},V_1, Y_1^{i-1},Z_{i+1}^n) \Big]^{+} \nonumber \\
		& \stackrel{(g)}{=} \sum_{i=1}^n \Big[ I(W_{d_1},V_1,Y_1^{i-1},Y_{2}^{i-1},Z_{i+1}^n; Y_{1,i}) - I(W_{d_1},V_1, Y_1^{i-1},Y_{2}^{i-1},Z_{i+1}^n;Z_i) \Big]^{+}  + n\Ma_1 + \epsilon_n \nonumber \\	
		& \stackrel{(h)}{=} \sum_{i=1}^n \Big[ I(W_{d_1},V_1,Y_2^{i-1},Z_{i+1}^n; Y_{1,i}) - I(W_{d_1},V_1, Y_2^{i-1},Z_{i+1}^n;Z_i) \Big]^{+}  + n\Ma_1 + \epsilon_n ,
\end{align}
	where $(a)$ holds because $I(W_{d_1}; V_1|Z^n)$ is limited by the entropy of $V_1$ which cannot exceed $n\Ma_1$; $(b)$ and $(c)$ follow by the chain rule of mutual information and by applying Csiszar's sum-identity \cite[pp. 25]{elGamalBook}; $(d)$  holds because  for all values of $x$ we have: $x \leq x^+$ and $0 \leq x^+$; $(e)$ and $(g)$ hold because if the eavesdropper is degraded with respect to Receiver~$1$, then all $[\cdot]^+$ terms are positive and if Receiver~$1$ is degraded with respect to the eavesdropper then all these terms are zero; $(f)$  holds because  $0 \leq x^+$ for all values of $x$;  and $(g)$ holds because Receiver~$1$ is degraded with respect to Receiver~$2$ and thus the following Markov chain holds: 
	\[
		(V_1, W_{d_1}, Z_{i+1}^n,Y_{1,i}, Z_i) \to Y_2^{i-1} \to Y_1^{i-1}.
	\] 
	Let $Q$ be a random variable uniform over $\{1,\dots,n\}$ and independent of all the previously defined random variables. We define the following random variables:
	\begin{align}
		U_1 & := \left( W_{d_1},V_1,Y_2^{Q-1},Z_{Q+1}^n \right), \\
		Y_1 & := Y_{1,Q}, \\
		Z & := Z_Q. 
	\end{align}
	Dividing by $n$, we can now rewrite the above inequality as
	\begin{align} \label{eq:condOverQ}
		R & \leq \sum_{q=1}^n \Pr\{Q=q\} \Big[ I(W_{d_1},V_1,Y_2^{q-1},Z_{q+1}^n; Y_{1,q}|Q=q) -  I(W_{d_1},V_1,Y_2^{q-1},Z_{q+1}^n; Z_q|Q=q)  \Big]^{+} + \Ma_1 + \frac{\epsilon_n}{n} \nonumber\\
			& \leq \big[ I(U_1;Y_1|Q) -I(U_1;Z|Q)\big]^{+} + \Ma_1 + \frac{\epsilon_n}{n}.
	\end{align}

We now derive   a similar bound as before, but involving Receivers $1,\ldots,k$. Consider
	\begin{align}
		kR & \leq \frac{1}{n} H(W_{d_1}, \ldots ,W_{d_k}) \nonumber \\
		& = \frac{1}{n}  H(W_{d_1}, \ldots, W_{d_k}|Z^n)+ \frac{1}{n}  I(W_{d_1}, \ldots, W_{d_k};Z^n) \nonumber \\
  	& \leq \frac{1}{n}  H(W_{d_1}, \ldots, W_{d_k}|Z^n) + \frac{\epsilon_n}{2n}  \nonumber \\
  	& \stackrel{(a)}{\leq} \frac{1}{n} \big[ H(W_{d_1}) + H(W_{d_2}|W_{d_1}) +  \ldots +  H(W_{d_k}|W_{d_{k-1}},\ldots,W_{d_1}) -   I(W_{d_1},\ldots, W_{d_k};Z^n) \big]  + \frac{\epsilon_n}{2n} \nonumber \\
  	& \stackrel{(b)}{\leq} \frac{1}{n} \Big[ I(W_{d_1};Y_1^n, V_1)+I(W_{d_2}; Y_2^n, V_1, V_2|W_{d_1})+  \ldots + I(W_{d_k}; Y_k^n, V_1, \ldots, V_k|W_{d_1},\ldots, W_{d_{k-1}})\nonumber \\
  	&  \qquad\qquad - I(W_{d_1},\ldots, W_{d_k};Z^n) \Big]  + \frac{\epsilon_n}{n} \nonumber \\
  &	\stackrel{(c)}{=} \frac{1}{n} \Big[ I(W_{d_1}; Y_1^n,V_1) - I(W_{d_1}; Z^n) \Big] \nonumber \\
  &  \qquad \qquad + \frac{1}{n} \sum_{\ell=2}^k \Big[I(W_{d_\ell}; Y_\ell^n, V_1, \ldots, V_\ell|W_{d_1}, \ldots, W_{d_{\ell-1}}) - I(W_{d_\ell};Z^n| W_{d_1}, \ldots, W_{d_{\ell-1}})\Big] +\frac{\epsilon_n}{n}\nonumber\\
    & = \frac{1}{n} \Big[ I(W_{d_1}; Y_1^n|V_1) - I(W_{d_1}; Z^n|V_1)+I(W_{d_1};V_1|Z^n) \Big] \nonumber \\
    & \quad  + \frac{1}{n} \sum_{\ell=2}^k \Big[I(W_{d_\ell}; Y_\ell^n|V_1, \ldots, V_\ell,W_{d_1}, \ldots, W_{d_{\ell-1}})  - I(W_{d_\ell};Z^n|V_1,\ldots, V_\ell, W_{d_1}, \ldots, W_{d_{\ell-1}})  \nonumber \\
     & \qquad\qquad  + I(W_{d_\ell}; V_1, \ldots, V_\ell| W_{d_1}, \ldots, W_{d_{\ell-1}}, Z^n)\Big] + \frac{\epsilon_n}{n}, \label{eq:constraint_sumk}
	\end{align}
	where (a) follows from the chain rule of mutual information; (b) follows from Fano's inequality; and (c) follows from the chain rule of mutual information. 
	
	\noindent
	In a similar way to \eqref{eq:condOverQ}, we can prove that 
  	\begin{equation}\label{eq:constraint_k1}
  		\frac{1}{n} \Big[ I(W_{d_1}; Y_1^n|V_1) - I(W_{d_1}; Z^n|V_1) \Big] \leq \big[ I(U_1;Y_1|Q)- I(U_1;Z|Q) \big]^{+}.
  	\end{equation}
  	
  	\noindent
	Then, we prove that for each $\ell \in \{2,\ldots,k\}$, the following set of inequalities holds:
  	\begin{align} \label{eq:constraint_k2}
		I(W_{d_\ell} & ; Y_\ell^n|V_1, \ldots, V_\ell,W_{d_1}, \ldots, W_{d_{\ell-1}}) - I(W_{d_\ell};Z^n|V_1,\ldots, V_\ell, W_{d_1}, \ldots, W_{\ell-1})\nonumber \\
		& \stackrel{(a)}{=} \sum_{i=1}^n\big[ I(W_{d_\ell}; Y_{\ell,i}|V_1, \ldots, V_\ell,W_{d_1}, \ldots, W_{d_{\ell-1}}, Y_\ell^{i-1}, Z_{i+1}^n) \nonumber\\
		&\qquad \qquad  - I(W_{d_\ell};Z_i|V_1,\ldots, V_\ell, W_{d_1}, \ldots, W_{d_{\ell-1}}, Y_\ell^{i-1},Z_{i+1}^n)\big]
\nonumber\\
		& \stackrel{(b)}{=} \sum_{i=1}^n\big[ I(W_{d_\ell}; Y_{\ell,i}|V_1, \ldots, V_\ell,W_{d_1}, \ldots, W_{d_{\ell-1}}, Y_2^{i-1}, \ldots, Y_\ell^{i-1}, Z_{i+1}^n) \nonumber \\
		&\qquad\qquad   - I(W_{d_\ell};Z_i|V_1,\ldots, V_\ell, W_{d_1}, \ldots, W_{d_{\ell-1}}, Y_2^{i-1}, \ldots, Y_\ell^{i-1},Z_{i+1}^n)\big] \nonumber \\
		& \stackrel{(c)}{\leq} \sum_{i=1}^n\big[ I(W_{d_\ell}; Y_{\ell,i}|V_1, \ldots, V_\ell,W_{d_1}, \ldots, W_{d_{\ell-1}},  Y_2^{i-1}, \ldots,Y_\ell^{i-1}, Z_{i+1}^n) \nonumber \\
		&\qquad \qquad  - I(W_{d_\ell};Z_i|V_1,\ldots, V_\ell, W_{d_1}, \ldots, W_{d_{\ell-1}},  Y_2^{i-1}, \ldots,Y_\ell^{i-1},Z_{i+1}^n)\big]^{+} \nonumber \\
		&  \qquad +\sum_{i=1}^n \big[ I(V_\ell; Y_{\ell,i}|V_1, \ldots, V_{\ell-1},W_{d_1}, \ldots, W_{d_{\ell-1}},  Y_2^{i-1}, \ldots,Y_\ell^{i-1}, Z_{i+1}^n)  \nonumber \\
		&\qquad \qquad\qquad - I(V_\ell; Z_{i}|V_1, \ldots, V_{\ell-1},W_{d_1}, \ldots, W_{d_{\ell-1}},  Y_2^{i-1}, \ldots,Y_\ell^{i-1}, Z_{i+1}^n)  \big]^{+} \nonumber \\
  		& \stackrel{(d)}{\leq} \sum_{i=1}^n\big[ I(W_{d_\ell}, V_\ell; Y_{\ell,i}|V_1, \ldots, V_{\ell-1},W_{d_1}, \ldots, W_{d_{\ell-1}},  Y_2^{i-1}, \ldots,Y_\ell^{i-1}, Z_{i+1}^n) \nonumber \\
  		&\qquad \qquad  - I(W_{d_\ell}, V_\ell;Z_i|V_1,\ldots, V_{\ell-1}, W_{d_1}, \ldots, W_{d_{\ell-1}},  Y_2^{i-1}, \ldots,Y_\ell^{i-1},Z_{i+1}^n)\big]^{+} \nonumber \\
  		&  \qquad +\sum_{i=1}^n \big[ I(Y_{\ell+1}^{i-1}; Y_{\ell,i}|V_1, \ldots, V_{\ell},W_{d_1}, \ldots, W_{d_\ell},  Y_2^{i-1}, \ldots,Y_\ell^{i-1}, Z_{i+1}^n)  \nonumber \\
  		&\qquad \qquad \qquad- I(Y_{\ell+1}^{i-1}; Z_{i}|V_1, \ldots, V_{\ell},W_{d_1}, \ldots, W_{d_\ell},  Y_2^{i-1}, \ldots,Y_\ell^{i-1}, Z_{i+1}^n)  \big]^{+} \nonumber \\
  		&=\sum_{i=1}^n\big[ I(W_{d_\ell}, V_\ell, Y_{\ell+1}^{i-1}; Y_{\ell,i}|V_1, \ldots, V_{\ell-1},W_{d_1}, \ldots, W_{d_{\ell-1}},  Y_2^{i-1}, \ldots,Y_\ell^{i-1}, Z_{i+1}^n) \nonumber \\
  		&\qquad \qquad  - I(W_{d_\ell}, V_\ell, Y_{\ell+1}^{i-1};Z_i|V_1,\ldots, V_{\ell-1}, W_{d_1}, \ldots, W_{d_{\ell-1}}, Y_2^{i-1}, \ldots, Y_\ell^{i-1},Z_{i+1}^n)\big]^{+},
  	\end{align}
	where $(a)$ follows from the chain rule of mutual information and by applying Csiszar's sum-identity; $(b)$ because Receivers $1, \ldots, \ell-1$ are degraded with respect to Receiver~$\ell$, and so the following Markov chain  holds: 
 	\begin{equation}
 		(W_{d_\ell}, Y_{\ell,i},V_1, \ldots, V_k,W_{d_1}, \ldots, W_{d_{\ell-1}},Z_{i+1}^n) \to Y_{\ell}^{i-1} \to (Y_1^{i-1}, \ldots, Y_{\ell-1}^{i-1});
 	\end{equation}
	 $(c)$  holds because  for all values of $x$ we have: $x \leq x^+$ and $0 \leq x^+$; $(d)$ holds because if the eavesdropper is degraded with respect to Receiver~$\ell$, then all $[\cdot]^+$ terms are positive and if Receiver~$\ell$ is degraded with respect to the eavesdropper than all these terms are zero. 
 	
	We define for each $k\in\{2,\ldots, \Ka\}$ the random variables
 	\begin{align}
 		Y_k & := Y_{k,Q}\\
 		U_k & := (W_{d_k},V_k, Y_{k+1}^{Q-1}, U_{k-1}).
 	\end{align}
	Dividing by $n$, we can rewrite constraint \eqref{eq:constraint_k2} as
 	\begin{align}\label{eq:constraint_k3}
		\frac{1}{n} & \Big[ I(W_\ell; Y_\ell^n|V_1, \ldots, V_\ell,W_{d_1}, \ldots, W_{d_{\ell-1}}) - I(W_\ell;Z^n|V_1,\ldots, V_\ell, W_{d_1}, \ldots, W_{d_{\ell-1}}) \Big] \nonumber \\
		& \qquad \leq  \sum_{\ell=1}^k \big[ I(U_\ell;Y_\ell|U_{\ell-1},Q) - I( U_\ell;Z|U_{\ell-1},Q) \big]^+.
 	\end{align}

	Finally, we bound the following sum: 
 	\begin{align}
		I(W_{d_1} & ;V_1|Z^n) + \sum_{\ell=2}^kI(W_{d_\ell}; V_1, \ldots, V_\ell| W_{d_1}, \ldots, W_{d_\ell-1}, Z^n) \nonumber \\
		& \leq I(W_{d_1};V_1\ldots, V_k|Z^n) + \sum_{\ell=2}^kI(W_{d_\ell}; V_1, \ldots, V_k| W_{d_1}, \ldots, W_{d_\ell-1}, Z^n) \nonumber \\
		& = I(W_{d_1}, \ldots, W_{d_k};  V_1, \ldots, V_k|Z^n)\nonumber \\
		& \leq n\sum_{\ell=1}^k \Ma_\ell. \label{eq:constraint_cache}
 	\end{align}
 	
 	Taking into consideration constraints~\eqref{eq:constraint_k1}, \eqref{eq:constraint_k3} and \eqref{eq:constraint_cache}, we can rewrite constraint~\eqref{eq:constraint_sumk} as:
 	\begin{align}\label{eq:constraint_sumK}
 	kR	& \leq \sum_{\ell=1}^k \Big[ I(U_\ell;Y_\ell|U_{\ell-1},Q) - I( U_\ell;Z|U_{\ell-1},Q) \Big]^+ + \sum_{\ell=1}^k \Ma_\ell + \frac{\epsilon_n}{n},
 	\end{align}
 	where $U_0$ is a constant.
 	
	Letting $n \to \infty$, from constraints~\eqref{eq:condOverQ} and \eqref{eq:constraint_sumK}, we conclude that Lemma~\ref{lemma:1} holds.

\end{document}

%% file: superpos_picture_oneshot.tex
\begin{figure}[h!] \centering 
	\begin{tikzpicture} [scale=0.85]
	\draw[thick] (0,1) rectangle (1,5);
	\node[below] at (0.5,-0.5) {\textbf{Cloud center}};
	\node[below] at (5,-0.5) {\textbf{Satellite}};
	
	\draw (0,4) -- (1,4);	\draw (0,3) -- (1,3);	\draw (0,2) -- (1,2);	
	
	\draw[ultra thick,blue] (0,3) rectangle (1,4);
	\node[left,blue] at (0,3.5) {$\mathbf{W}_{\sec}=2$};
	
	\draw[dashed,thick,->] (1,4.5) -- (1.8,5.25) -- (2.5,5.25); 	
	\draw[dashed,thick,->] (1,3.5) -- (1.8,3.75) -- (2.5,3.75); 				
	\draw[dashed,thick,->] (1,2.5) -- (1.8,2.25) -- (2.5,2.25); 				
	\draw[dashed,thick,->] (1,1.5) -- (1.8,0.75) -- (2.5,0.75); 
	
	\draw[thick] (2.5,0.25) rectangle (7.5,1.25);
	\draw[thick] (2.5,1.75) rectangle (7.5,2.75);
	\draw[thick] (2.5,3.25) rectangle (7.5,4.25);
	\draw[thick] (2.5,4.75) rectangle (7.5,5.75);				
	
	\draw[rounded corners=8pt,ultra thick,blue] (6.5,0) rectangle(7.5,6);
	\node[above,blue] at (7,6) {$ \mathbf{W}_{\textnormal{sat}}= 5$};
	\draw[ultra thick,magenta] (6.5,3.25) rectangle (7.5,4.25);
	\node[right,magenta] at (7.5,4) {bin corresponding to};
	 \node[right,magenta] at (7.5,3.45) {$(\mathbf{W}_{\sec} = 2, \mathbf{W}_{\textnormal{sat}} = 5)$};
	
	\draw[thick] (3.5,0.25) -- (3.5,1.25);	\draw[thick] (4.5,0.25) -- (4.5,1.25);
	\draw[thick] (5.5,0.25) -- (5.5,1.25);	\draw[thick] (6.5,0.25) -- (6.5,1.25);
	\draw[thick] (3.5,1.75) -- (3.5,2.75);	\draw[thick] (4.5,1.75) -- (4.5,2.75);
	\draw[thick] (5.5,1.75) -- (5.5,2.75);	\draw[thick] (6.5,1.75) -- (6.5,2.75);
	\draw[thick] (3.5,3.25) -- (3.5,4.25);	\draw[thick] (4.5,3.25) -- (4.5,4.25);
	\draw[thick] (5.5,3.25) -- (5.5,4.25);	\draw[thick] (6.5,3.25) -- (6.5,4.25);
	\draw[thick] (3.5,4.75) -- (3.5,5.75);	\draw[thick] (4.5,4.75) -- (4.5,5.75);
	\draw[thick] (5.5,4.75) -- (5.5,5.75);	\draw[thick] (6.5,4.75) -- (6.5,5.75);	
	
	\draw[fill](0.5,1.5)circle(1mm);
	\draw[fill](0.5,2.5)circle(1mm);	
	\draw[fill](0.5,3.5)circle(1mm);
	\draw[fill](0.5,4.5)circle(1mm);	
	
	\draw[fill](2.75,0.5)circle(0.35mm);	\draw[fill](2.75,1)circle(0.35mm);
	\draw[fill](3.25,0.5)circle(0.35mm);	\draw[fill](3.25,1)circle(0.35mm);
	\draw[fill](3,0.75)circle(0.35mm);
	\draw[fill](2.75,2)circle(0.35mm);		\draw[fill](2.75,2.5)circle(0.35mm);
	\draw[fill](3.25,2)circle(0.35mm);		\draw[fill](3.25,2.5)circle(0.35mm);
	\draw[fill](3,2.25)circle(0.35mm);
	\draw[fill](2.75,3.5)circle(0.35mm);	\draw[fill](2.75,4)circle(0.35mm);
	\draw[fill](3.25,3.5)circle(0.35mm);	\draw[fill](3.25,4)circle(0.35mm);
	\draw[fill](3,3.75)circle(0.35mm);
	\draw[fill](2.75,5)circle(0.35mm);		\draw[fill](2.75,5.5)circle(0.35mm);
	\draw[fill](3.25,5)circle(0.35mm);		\draw[fill](3.25,5.5)circle(0.35mm);
	\draw[fill](3,5.25)circle(0.35mm);
	
	\draw[fill](3.75,0.5)circle(0.35mm);	\draw[fill](3.75,1)circle(0.35mm);
	\draw[fill](4.25,0.5)circle(0.35mm);	\draw[fill](4.25,1)circle(0.35mm);
	\draw[fill](4,0.75)circle(0.35mm);
	\draw[fill](3.75,2)circle(0.35mm);		\draw[fill](3.75,2.5)circle(0.35mm);
	\draw[fill](4.25,2)circle(0.35mm);		\draw[fill](4.25,2.5)circle(0.35mm);
	\draw[fill](4,2.25)circle(0.35mm);
	\draw[fill](3.75,3.5)circle(0.35mm);	\draw[fill](3.75,4)circle(0.35mm);
	\draw[fill](4.25,3.5)circle(0.35mm);	\draw[fill](4.25,4)circle(0.35mm);
	\draw[fill](4,3.75)circle(0.35mm);
	\draw[fill](3.75,5)circle(0.35mm);		\draw[fill](3.75,5.5)circle(0.35mm);
	\draw[fill](4.25,5)circle(0.35mm);		\draw[fill](4.25,5.5)circle(0.35mm);
	\draw[fill](4,5.25)circle(0.35mm);
	
	\draw[fill](4.75,0.5)circle(0.35mm);	\draw[fill](4.75,1)circle(0.35mm);
	\draw[fill](5.25,0.5)circle(0.35mm);	\draw[fill](5.25,1)circle(0.35mm);
	\draw[fill](5,0.75)circle(0.35mm);
	\draw[fill](4.75,2)circle(0.35mm);		\draw[fill](4.75,2.5)circle(0.35mm);
	\draw[fill](5.25,2)circle(0.35mm);		\draw[fill](5.25,2.5)circle(0.35mm);
	\draw[fill](5,2.25)circle(0.35mm);
	\draw[fill](4.75,3.5)circle(0.35mm);	\draw[fill](4.75,4)circle(0.35mm);
	\draw[fill](5.25,3.5)circle(0.35mm);	\draw[fill](5.25,4)circle(0.35mm);
	\draw[fill](5,3.75)circle(0.35mm);
	\draw[fill](4.75,5)circle(0.35mm);		\draw[fill](4.75,5.5)circle(0.35mm);
	\draw[fill](5.25,5)circle(0.35mm);		\draw[fill](5.25,5.5)circle(0.35mm);
	\draw[fill](5,5.25)circle(0.35mm);
	
	\draw[fill](5.75,0.5)circle(0.35mm);	\draw[fill](5.75,1)circle(0.35mm);
	\draw[fill](6.25,0.5)circle(0.35mm);	\draw[fill](6.25,1)circle(0.35mm);
	\draw[fill](6,0.75)circle(0.35mm);
	\draw[fill](5.75,2)circle(0.35mm);		\draw[fill](5.75,2.5)circle(0.35mm);
	\draw[fill](6.25,2)circle(0.35mm);		\draw[fill](6.25,2.5)circle(0.35mm);
	\draw[fill](6,2.25)circle(0.35mm);
	\draw[fill](5.75,3.5)circle(0.35mm);	\draw[fill](5.75,4)circle(0.35mm);
	\draw[fill](6.25,3.5)circle(0.35mm);	\draw[fill](6.25,4)circle(0.35mm);
	\draw[fill](6,3.75)circle(0.35mm);
	\draw[fill](5.75,5)circle(0.35mm);		\draw[fill](5.75,5.5)circle(0.35mm);
	\draw[fill](6.25,5)circle(0.35mm);		\draw[fill](6.25,5.5)circle(0.35mm);
	\draw[fill](6,5.25)circle(0.35mm);
	
	\draw[fill](6.75,0.5)circle(0.35mm);	\draw[fill](6.75,1)circle(0.35mm);
	\draw[fill](7.25,0.5)circle(0.35mm);	\draw[fill](7.25,1)circle(0.35mm);
	\draw[fill](7,0.75)circle(0.35mm);
	\draw[fill](6.75,2)circle(0.35mm);		\draw[fill](6.75,2.5)circle(0.35mm);
	\draw[fill](7.25,2)circle(0.35mm);		\draw[fill](7.25,2.5)circle(0.35mm);
	\draw[fill](7,2.25)circle(0.35mm);
	\draw[fill](6.75,3.5)circle(0.35mm);	\draw[fill](6.75,4)circle(0.35mm);
	\draw[fill](7.25,3.5)circle(0.35mm);	\draw[fill](7.25,4)circle(0.35mm);
	\draw[fill](7,3.75)circle(0.35mm);
	\draw[fill](6.75,5)circle(0.35mm);		\draw[fill](6.75,5.5)circle(0.35mm);
	\draw[fill](7.25,5)circle(0.35mm);		\draw[fill](7.25,5.5)circle(0.35mm);
	\draw[fill](7,5.25)circle(0.35mm);
	\end{tikzpicture}
	\caption{Superposition codebook with random binning in the satellites.} \label{fig:superposition_one} 
\end{figure}

%% file: secure_piggyback.tex
\begin{figure}[h!] \begin{center}
	\begin{tikzpicture} 
	\draw[thick] (0,0) rectangle (8,6);
	
	\draw (0,5) -- (8,5);	\draw (0,4) -- (8,4);	\draw (0,3) -- (8,3);
	\draw (0,2) -- (8,2);	\draw (0,1) -- (8,1);
	
	\draw (1,0) -- (1,6);	\draw (2,0) -- (2,6);	\draw (3,0) -- (3,6);	\draw (4,0) -- (4,6);
	\draw (5,0) -- (5,6);	\draw (6,0) -- (6,6);	\draw (7,0) -- (7,6);
	
	\draw[very thick,blue] (6,0) rectangle (7,6);
	\draw[very thick,blue] (0,4) rectangle (8,5);	
	\node[left,blue] at (0,4.5) {$\mathbf{W}_{\sec,\{1,2\}}^{(A)}$};
	\node[above,blue] at (6.5,6) {$W_{d_4,\{1,2\}}^{(B)}$};
\draw[ultra thick,magenta] (6,4) rectangle (7,5);
		\draw[magenta,thick,->] (7,4.65) -- (7.5,5.25) -- (8.3,5.25);
		\node[right,magenta] at (8.3,5.2) {{Wiretap bin corresponding}};
		\node[right,magenta] at (8.3,4.65) {to $\Big(\mathbf{W}_{\sec,\{1,2\}}^{(A)}, W_{d_4,\{1,2\}}^{(B)}\Big)$};
	
	\draw[fill](0.2,0.2)circle(0.2mm);	\draw[fill](0.4,0.2)circle(0.2mm);
	\draw[fill](0.6,0.2)circle(0.2mm);	\draw[fill](0.8,0.2)circle(0.2mm);
	\draw[fill](0.2,0.6)circle(0.2mm);	\draw[fill](0.4,0.6)circle(0.2mm);
	\draw[fill](0.6,0.6)circle(0.2mm);	\draw[fill](0.8,0.6)circle(0.2mm);
	\draw[fill](0.25,0.4)circle(0.2mm);	\draw[fill](0.5,0.4)circle(0.2mm);	\draw[fill](0.75,0.4)circle(0.2mm);	
	\draw[fill](0.25,0.8)circle(0.2mm);	\draw[fill](0.5,0.8)circle(0.2mm);	\draw[fill](0.75,0.8)circle(0.2mm);			
	\draw[fill](0.2,1.2)circle(0.2mm);	\draw[fill](0.4,1.2)circle(0.2mm);
	\draw[fill](0.6,1.2)circle(0.2mm);	\draw[fill](0.8,1.2)circle(0.2mm);
	\draw[fill](0.2,1.6)circle(0.2mm);	\draw[fill](0.4,1.6)circle(0.2mm);
	\draw[fill](0.6,1.6)circle(0.2mm);	\draw[fill](0.8,1.6)circle(0.2mm);
	\draw[fill](0.25,1.4)circle(0.2mm);	\draw[fill](0.5,1.4)circle(0.2mm);	\draw[fill](0.75,1.4)circle(0.2mm);	
	\draw[fill](0.25,1.8)circle(0.2mm);	\draw[fill](0.5,1.8)circle(0.2mm);	\draw[fill](0.75,1.8)circle(0.2mm);			
	\draw[fill](0.2,2.2)circle(0.2mm);	\draw[fill](0.4,2.2)circle(0.2mm);
	\draw[fill](0.6,2.2)circle(0.2mm);	\draw[fill](0.8,2.2)circle(0.2mm);
	\draw[fill](0.2,2.6)circle(0.2mm);	\draw[fill](0.4,2.6)circle(0.2mm);
	\draw[fill](0.6,2.6)circle(0.2mm);	\draw[fill](0.8,2.6)circle(0.2mm);
	\draw[fill](0.25,2.4)circle(0.2mm);	\draw[fill](0.5,2.4)circle(0.2mm);	\draw[fill](0.75,2.4)circle(0.2mm);	
	\draw[fill](0.25,2.8)circle(0.2mm);	\draw[fill](0.5,2.8)circle(0.2mm);	\draw[fill](0.75,2.8)circle(0.2mm);		
	\draw[fill](0.2,3.2)circle(0.2mm);	\draw[fill](0.4,3.2)circle(0.2mm);
	\draw[fill](0.6,3.2)circle(0.2mm);	\draw[fill](0.8,3.2)circle(0.2mm);
	\draw[fill](0.2,3.6)circle(0.2mm);	\draw[fill](0.4,3.6)circle(0.2mm);
	\draw[fill](0.6,3.6)circle(0.2mm);	\draw[fill](0.8,3.6)circle(0.2mm);
	\draw[fill](0.25,3.4)circle(0.2mm);	\draw[fill](0.5,3.4)circle(0.2mm);	\draw[fill](0.75,3.4)circle(0.2mm);	
	\draw[fill](0.25,3.8)circle(0.2mm);	\draw[fill](0.5,3.8)circle(0.2mm);	\draw[fill](0.75,3.8)circle(0.2mm);		
	\draw[fill](0.2,4.2)circle(0.2mm);	\draw[fill](0.4,4.2)circle(0.2mm);
	\draw[fill](0.6,4.2)circle(0.2mm);	\draw[fill](0.8,4.2)circle(0.2mm);
	\draw[fill](0.2,4.6)circle(0.2mm);	\draw[fill](0.4,4.6)circle(0.2mm);
	\draw[fill](0.6,4.6)circle(0.2mm);	\draw[fill](0.8,4.6)circle(0.2mm);
	\draw[fill](0.25,4.4)circle(0.2mm);	\draw[fill](0.5,4.4)circle(0.2mm);	\draw[fill](0.75,4.4)circle(0.2mm);	
	\draw[fill](0.25,4.8)circle(0.2mm);	\draw[fill](0.5,4.8)circle(0.2mm);	\draw[fill](0.75,4.8)circle(0.2mm);			
	\draw[fill](0.2,5.2)circle(0.2mm);	\draw[fill](0.4,5.2)circle(0.2mm);
	\draw[fill](0.6,5.2)circle(0.2mm);	\draw[fill](0.8,5.2)circle(0.2mm);
	\draw[fill](0.2,5.6)circle(0.2mm);	\draw[fill](0.4,5.6)circle(0.2mm);
	\draw[fill](0.6,5.6)circle(0.2mm);	\draw[fill](0.8,5.6)circle(0.2mm);
	\draw[fill](0.25,5.4)circle(0.2mm); \draw[fill](0.5,5.4)circle(0.2mm);	\draw[fill](0.75,5.4)circle(0.2mm);
	\draw[fill](0.25,5.8)circle(0.2mm);	\draw[fill](0.5,5.8)circle(0.2mm);	\draw[fill](0.75,5.8)circle(0.2mm);	
	
	\draw[fill](1.2,0.2)circle(0.2mm);	\draw[fill](1.4,0.2)circle(0.2mm);
	\draw[fill](1.6,0.2)circle(0.2mm);	\draw[fill](1.8,0.2)circle(0.2mm);
	\draw[fill](1.2,0.6)circle(0.2mm);	\draw[fill](1.4,0.6)circle(0.2mm);
	\draw[fill](1.6,0.6)circle(0.2mm);	\draw[fill](1.8,0.6)circle(0.2mm);
	\draw[fill](1.25,0.4)circle(0.2mm);	\draw[fill](1.5,0.4)circle(0.2mm);	\draw[fill](1.75,0.4)circle(0.2mm);	
	\draw[fill](1.25,0.8)circle(0.2mm);	\draw[fill](1.5,0.8)circle(0.2mm);	\draw[fill](1.75,0.8)circle(0.2mm);			
	\draw[fill](1.2,1.2)circle(0.2mm);	\draw[fill](1.4,1.2)circle(0.2mm);
	\draw[fill](1.6,1.2)circle(0.2mm);	\draw[fill](1.8,1.2)circle(0.2mm);
	\draw[fill](1.2,1.6)circle(0.2mm);	\draw[fill](1.4,1.6)circle(0.2mm);
	\draw[fill](1.6,1.6)circle(0.2mm);	\draw[fill](1.8,1.6)circle(0.2mm);
	\draw[fill](1.25,1.4)circle(0.2mm);	\draw[fill](1.5,1.4)circle(0.2mm);	\draw[fill](1.75,1.4)circle(0.2mm);	
	\draw[fill](1.25,1.8)circle(0.2mm);	\draw[fill](1.5,1.8)circle(0.2mm);	\draw[fill](1.75,1.8)circle(0.2mm);			
	\draw[fill](1.2,2.2)circle(0.2mm);	\draw[fill](1.4,2.2)circle(0.2mm);
	\draw[fill](1.6,2.2)circle(0.2mm);	\draw[fill](1.8,2.2)circle(0.2mm);
	\draw[fill](1.2,2.6)circle(0.2mm);	\draw[fill](1.4,2.6)circle(0.2mm);
	\draw[fill](1.6,2.6)circle(0.2mm);	\draw[fill](1.8,2.6)circle(0.2mm);
	\draw[fill](1.25,2.4)circle(0.2mm);	\draw[fill](1.5,2.4)circle(0.2mm);	\draw[fill](1.75,2.4)circle(0.2mm);	
	\draw[fill](1.25,2.8)circle(0.2mm);	\draw[fill](1.5,2.8)circle(0.2mm);	\draw[fill](1.75,2.8)circle(0.2mm);		
	\draw[fill](1.2,3.2)circle(0.2mm);	\draw[fill](1.4,3.2)circle(0.2mm);
	\draw[fill](1.6,3.2)circle(0.2mm);	\draw[fill](1.8,3.2)circle(0.2mm);
	\draw[fill](1.2,3.6)circle(0.2mm);	\draw[fill](1.4,3.6)circle(0.2mm);
	\draw[fill](1.6,3.6)circle(0.2mm);	\draw[fill](1.8,3.6)circle(0.2mm);
	\draw[fill](1.25,3.4)circle(0.2mm);	\draw[fill](1.5,3.4)circle(0.2mm);	\draw[fill](1.75,3.4)circle(0.2mm);	
	\draw[fill](1.25,3.8)circle(0.2mm);	\draw[fill](1.5,3.8)circle(0.2mm);	\draw[fill](1.75,3.8)circle(0.2mm);		
	\draw[fill](1.2,4.2)circle(0.2mm);	\draw[fill](1.4,4.2)circle(0.2mm);
	\draw[fill](1.6,4.2)circle(0.2mm);	\draw[fill](1.8,4.2)circle(0.2mm);
	\draw[fill](1.2,4.6)circle(0.2mm);	\draw[fill](1.4,4.6)circle(0.2mm);
	\draw[fill](1.6,4.6)circle(0.2mm);	\draw[fill](1.8,4.6)circle(0.2mm);
	\draw[fill](1.25,4.4)circle(0.2mm);	\draw[fill](1.5,4.4)circle(0.2mm);	\draw[fill](1.75,4.4)circle(0.2mm);	
	\draw[fill](1.25,4.8)circle(0.2mm);	\draw[fill](1.5,4.8)circle(0.2mm);	\draw[fill](1.75,4.8)circle(0.2mm);			
	\draw[fill](1.2,5.2)circle(0.2mm);	\draw[fill](1.4,5.2)circle(0.2mm);
	\draw[fill](1.6,5.2)circle(0.2mm);	\draw[fill](1.8,5.2)circle(0.2mm);
	\draw[fill](1.2,5.6)circle(0.2mm);	\draw[fill](1.4,5.6)circle(0.2mm);
	\draw[fill](1.6,5.6)circle(0.2mm);	\draw[fill](1.8,5.6)circle(0.2mm);
	\draw[fill](1.25,5.4)circle(0.2mm);	\draw[fill](1.5,5.4)circle(0.2mm);	\draw[fill](1.75,5.4)circle(0.2mm);
	\draw[fill](1.25,5.8)circle(0.2mm);	\draw[fill](1.5,5.8)circle(0.2mm);	\draw[fill](1.75,5.8)circle(0.2mm);	
	
	\draw[fill](2.2,0.2)circle(0.2mm);	\draw[fill](2.4,0.2)circle(0.2mm);
	\draw[fill](2.6,0.2)circle(0.2mm);	\draw[fill](2.8,0.2)circle(0.2mm);
	\draw[fill](2.2,0.6)circle(0.2mm);	\draw[fill](2.4,0.6)circle(0.2mm);
	\draw[fill](2.6,0.6)circle(0.2mm);	\draw[fill](2.8,0.6)circle(0.2mm);
	\draw[fill](2.25,0.4)circle(0.2mm);	\draw[fill](2.5,0.4)circle(0.2mm);	\draw[fill](2.75,0.4)circle(0.2mm);	
	\draw[fill](2.25,0.8)circle(0.2mm);	\draw[fill](2.5,0.8)circle(0.2mm);	\draw[fill](2.75,0.8)circle(0.2mm);			
	\draw[fill](2.2,1.2)circle(0.2mm);	\draw[fill](2.4,1.2)circle(0.2mm);
	\draw[fill](2.6,1.2)circle(0.2mm);	\draw[fill](2.8,1.2)circle(0.2mm);
	\draw[fill](2.2,1.6)circle(0.2mm);	\draw[fill](2.4,1.6)circle(0.2mm);
	\draw[fill](2.6,1.6)circle(0.2mm);	\draw[fill](2.8,1.6)circle(0.2mm);
	\draw[fill](2.25,1.4)circle(0.2mm);	\draw[fill](2.5,1.4)circle(0.2mm);	\draw[fill](2.75,1.4)circle(0.2mm);	
	\draw[fill](2.25,1.8)circle(0.2mm);	\draw[fill](2.5,1.8)circle(0.2mm);	\draw[fill](2.75,1.8)circle(0.2mm);			
	\draw[fill](2.2,2.2)circle(0.2mm);	\draw[fill](2.4,2.2)circle(0.2mm);
	\draw[fill](2.6,2.2)circle(0.2mm);	\draw[fill](2.8,2.2)circle(0.2mm);
	\draw[fill](2.2,2.6)circle(0.2mm);	\draw[fill](2.4,2.6)circle(0.2mm);
	\draw[fill](2.6,2.6)circle(0.2mm);	\draw[fill](2.8,2.6)circle(0.2mm);
	\draw[fill](2.25,2.4)circle(0.2mm);	\draw[fill](2.5,2.4)circle(0.2mm);	\draw[fill](2.75,2.4)circle(0.2mm);	
	\draw[fill](2.25,2.8)circle(0.2mm);	\draw[fill](2.5,2.8)circle(0.2mm);	\draw[fill](2.75,2.8)circle(0.2mm);		
	\draw[fill](2.2,3.2)circle(0.2mm);	\draw[fill](2.4,3.2)circle(0.2mm);
	\draw[fill](2.6,3.2)circle(0.2mm);	\draw[fill](2.8,3.2)circle(0.2mm);
	\draw[fill](2.2,3.6)circle(0.2mm);	\draw[fill](2.4,3.6)circle(0.2mm);
	\draw[fill](2.6,3.6)circle(0.2mm);	\draw[fill](2.8,3.6)circle(0.2mm);
	\draw[fill](2.25,3.4)circle(0.2mm);	\draw[fill](2.5,3.4)circle(0.2mm);	\draw[fill](2.75,3.4)circle(0.2mm);	
	\draw[fill](2.25,3.8)circle(0.2mm);	\draw[fill](2.5,3.8)circle(0.2mm);	\draw[fill](2.75,3.8)circle(0.2mm);		
	\draw[fill](2.2,4.2)circle(0.2mm);	\draw[fill](2.4,4.2)circle(0.2mm);
	\draw[fill](2.6,4.2)circle(0.2mm);	\draw[fill](2.8,4.2)circle(0.2mm);
	\draw[fill](2.2,4.6)circle(0.2mm);	\draw[fill](2.4,4.6)circle(0.2mm);
	\draw[fill](2.6,4.6)circle(0.2mm);	\draw[fill](2.8,4.6)circle(0.2mm);
	\draw[fill](2.25,4.4)circle(0.2mm);	\draw[fill](2.5,4.4)circle(0.2mm);	\draw[fill](2.75,4.4)circle(0.2mm);	
	\draw[fill](2.25,4.8)circle(0.2mm);	\draw[fill](2.5,4.8)circle(0.2mm);	\draw[fill](2.75,4.8)circle(0.2mm);			
	\draw[fill](2.2,5.2)circle(0.2mm);	\draw[fill](2.4,5.2)circle(0.2mm);
	\draw[fill](2.6,5.2)circle(0.2mm);	\draw[fill](2.8,5.2)circle(0.2mm);
	\draw[fill](2.2,5.6)circle(0.2mm);	\draw[fill](2.4,5.6)circle(0.2mm);
	\draw[fill](2.6,5.6)circle(0.2mm);	\draw[fill](2.8,5.6)circle(0.2mm);
	\draw[fill](2.25,5.4)circle(0.2mm);	\draw[fill](2.5,5.4)circle(0.2mm);	\draw[fill](2.75,5.4)circle(0.2mm);
	\draw[fill](2.25,5.8)circle(0.2mm);	\draw[fill](2.5,5.8)circle(0.2mm);	\draw[fill](2.75,5.8)circle(0.2mm);	
	
	\draw[fill](3.2,0.2)circle(0.2mm);	\draw[fill](3.4,0.2)circle(0.2mm);
	\draw[fill](3.6,0.2)circle(0.2mm);	\draw[fill](3.8,0.2)circle(0.2mm);
	\draw[fill](3.2,0.6)circle(0.2mm);	\draw[fill](3.4,0.6)circle(0.2mm);
	\draw[fill](3.6,0.6)circle(0.2mm);	\draw[fill](3.8,0.6)circle(0.2mm);
	\draw[fill](3.25,0.4)circle(0.2mm);	\draw[fill](3.5,0.4)circle(0.2mm);	\draw[fill](3.75,0.4)circle(0.2mm);	
	\draw[fill](3.25,0.8)circle(0.2mm);	\draw[fill](3.5,0.8)circle(0.2mm);	\draw[fill](3.75,0.8)circle(0.2mm);			
	\draw[fill](3.2,1.2)circle(0.2mm);	\draw[fill](3.4,1.2)circle(0.2mm);
	\draw[fill](3.6,1.2)circle(0.2mm);	\draw[fill](3.8,1.2)circle(0.2mm);
	\draw[fill](3.2,1.6)circle(0.2mm);	\draw[fill](3.4,1.6)circle(0.2mm);
	\draw[fill](3.6,1.6)circle(0.2mm);	\draw[fill](3.8,1.6)circle(0.2mm);
	\draw[fill](3.25,1.4)circle(0.2mm);	\draw[fill](3.5,1.4)circle(0.2mm);	\draw[fill](3.75,1.4)circle(0.2mm);	
	\draw[fill](3.25,1.8)circle(0.2mm);	\draw[fill](3.5,1.8)circle(0.2mm);	\draw[fill](3.75,1.8)circle(0.2mm);			
	\draw[fill](3.2,2.2)circle(0.2mm);	\draw[fill](3.4,2.2)circle(0.2mm);
	\draw[fill](3.6,2.2)circle(0.2mm);	\draw[fill](3.8,2.2)circle(0.2mm);
	\draw[fill](3.2,2.6)circle(0.2mm);	\draw[fill](3.4,2.6)circle(0.2mm);
	\draw[fill](3.6,2.6)circle(0.2mm);	\draw[fill](3.8,2.6)circle(0.2mm);
	\draw[fill](3.25,2.4)circle(0.2mm);	\draw[fill](3.5,2.4)circle(0.2mm);	\draw[fill](3.75,2.4)circle(0.2mm);	
	\draw[fill](3.25,2.8)circle(0.2mm);	\draw[fill](3.5,2.8)circle(0.2mm);	\draw[fill](3.75,2.8)circle(0.2mm);		
	\draw[fill](3.2,3.2)circle(0.2mm);	\draw[fill](3.4,3.2)circle(0.2mm);
	\draw[fill](3.6,3.2)circle(0.2mm);	\draw[fill](3.8,3.2)circle(0.2mm);
	\draw[fill](3.2,3.6)circle(0.2mm);	\draw[fill](3.4,3.6)circle(0.2mm);
	\draw[fill](3.6,3.6)circle(0.2mm);	\draw[fill](3.8,3.6)circle(0.2mm);
	\draw[fill](3.25,3.4)circle(0.2mm);	\draw[fill](3.5,3.4)circle(0.2mm);	\draw[fill](3.75,3.4)circle(0.2mm);	
	\draw[fill](3.25,3.8)circle(0.2mm);	\draw[fill](3.5,3.8)circle(0.2mm);	\draw[fill](3.75,3.8)circle(0.2mm);		
	\draw[fill](3.2,4.2)circle(0.2mm);	\draw[fill](3.4,4.2)circle(0.2mm);
	\draw[fill](3.6,4.2)circle(0.2mm);	\draw[fill](3.8,4.2)circle(0.2mm);
	\draw[fill](3.2,4.6)circle(0.2mm);	\draw[fill](3.4,4.6)circle(0.2mm);
	\draw[fill](3.6,4.6)circle(0.2mm);	\draw[fill](3.8,4.6)circle(0.2mm);
	\draw[fill](3.25,4.4)circle(0.2mm);	\draw[fill](3.5,4.4)circle(0.2mm);	\draw[fill](3.75,4.4)circle(0.2mm);	
	\draw[fill](3.25,4.8)circle(0.2mm);	\draw[fill](3.5,4.8)circle(0.2mm);	\draw[fill](3.75,4.8)circle(0.2mm);			
	\draw[fill](3.2,5.2)circle(0.2mm);	\draw[fill](3.4,5.2)circle(0.2mm);
	\draw[fill](3.6,5.2)circle(0.2mm);	\draw[fill](3.8,5.2)circle(0.2mm);
	\draw[fill](3.2,5.6)circle(0.2mm);	\draw[fill](3.4,5.6)circle(0.2mm);
	\draw[fill](3.6,5.6)circle(0.2mm);	\draw[fill](3.8,5.6)circle(0.2mm);
	\draw[fill](3.25,5.4)circle(0.2mm);	\draw[fill](3.5,5.4)circle(0.2mm);	\draw[fill](3.75,5.4)circle(0.2mm);
	\draw[fill](3.25,5.8)circle(0.2mm);	\draw[fill](3.5,5.8)circle(0.2mm);	\draw[fill](3.75,5.8)circle(0.2mm);	
	
	\draw[fill](4.2,0.2)circle(0.2mm);	\draw[fill](4.4,0.2)circle(0.2mm);
	\draw[fill](4.6,0.2)circle(0.2mm);	\draw[fill](4.8,0.2)circle(0.2mm);
	\draw[fill](4.2,0.6)circle(0.2mm);	\draw[fill](4.4,0.6)circle(0.2mm);
	\draw[fill](4.6,0.6)circle(0.2mm);	\draw[fill](4.8,0.6)circle(0.2mm);
	\draw[fill](4.25,0.4)circle(0.2mm);	\draw[fill](4.5,0.4)circle(0.2mm);	\draw[fill](4.75,0.4)circle(0.2mm);	
	\draw[fill](4.25,0.8)circle(0.2mm);	\draw[fill](4.5,0.8)circle(0.2mm);	\draw[fill](4.75,0.8)circle(0.2mm);			
	\draw[fill](4.2,1.2)circle(0.2mm);	\draw[fill](4.4,1.2)circle(0.2mm);
	\draw[fill](4.6,1.2)circle(0.2mm);	\draw[fill](4.8,1.2)circle(0.2mm);
	\draw[fill](4.2,1.6)circle(0.2mm);	\draw[fill](4.4,1.6)circle(0.2mm);
	\draw[fill](4.6,1.6)circle(0.2mm);	\draw[fill](4.8,1.6)circle(0.2mm);
	\draw[fill](4.25,1.4)circle(0.2mm);	\draw[fill](4.5,1.4)circle(0.2mm);	\draw[fill](4.75,1.4)circle(0.2mm);	
	\draw[fill](4.25,1.8)circle(0.2mm);	\draw[fill](4.5,1.8)circle(0.2mm);	\draw[fill](4.75,1.8)circle(0.2mm);			
	\draw[fill](4.2,2.2)circle(0.2mm);	\draw[fill](4.4,2.2)circle(0.2mm);
	\draw[fill](4.6,2.2)circle(0.2mm);	\draw[fill](4.8,2.2)circle(0.2mm);
	\draw[fill](4.2,2.6)circle(0.2mm);	\draw[fill](4.4,2.6)circle(0.2mm);
	\draw[fill](4.6,2.6)circle(0.2mm);	\draw[fill](4.8,2.6)circle(0.2mm);
	\draw[fill](4.25,2.4)circle(0.2mm);	\draw[fill](4.5,2.4)circle(0.2mm);	\draw[fill](4.75,2.4)circle(0.2mm);	
	\draw[fill](4.25,2.8)circle(0.2mm);	\draw[fill](4.5,2.8)circle(0.2mm);	\draw[fill](4.75,2.8)circle(0.2mm);		
	\draw[fill](4.2,3.2)circle(0.2mm);	\draw[fill](4.4,3.2)circle(0.2mm);
	\draw[fill](4.6,3.2)circle(0.2mm);	\draw[fill](4.8,3.2)circle(0.2mm);
	\draw[fill](4.2,3.6)circle(0.2mm);	\draw[fill](4.4,3.6)circle(0.2mm);
	\draw[fill](4.6,3.6)circle(0.2mm);	\draw[fill](4.8,3.6)circle(0.2mm);
	\draw[fill](4.25,3.4)circle(0.2mm);	\draw[fill](4.5,3.4)circle(0.2mm);	\draw[fill](4.75,3.4)circle(0.2mm);	
	\draw[fill](4.25,3.8)circle(0.2mm);	\draw[fill](4.5,3.8)circle(0.2mm);	\draw[fill](4.75,3.8)circle(0.2mm);		
	\draw[fill](4.2,4.2)circle(0.2mm);	\draw[fill](4.4,4.2)circle(0.2mm);
	\draw[fill](4.6,4.2)circle(0.2mm);	\draw[fill](4.8,4.2)circle(0.2mm);
	\draw[fill](4.2,4.6)circle(0.2mm);	\draw[fill](4.4,4.6)circle(0.2mm);
	\draw[fill](4.6,4.6)circle(0.2mm);	\draw[fill](4.8,4.6)circle(0.2mm);
	\draw[fill](4.25,4.4)circle(0.2mm);	\draw[fill](4.5,4.4)circle(0.2mm);	\draw[fill](4.75,4.4)circle(0.2mm);	
	\draw[fill](4.25,4.8)circle(0.2mm);	\draw[fill](4.5,4.8)circle(0.2mm);	\draw[fill](4.75,4.8)circle(0.2mm);			
	\draw[fill](4.2,5.2)circle(0.2mm);	\draw[fill](4.4,5.2)circle(0.2mm);
	\draw[fill](4.6,5.2)circle(0.2mm);	\draw[fill](4.8,5.2)circle(0.2mm);
	\draw[fill](4.2,5.6)circle(0.2mm);	\draw[fill](4.4,5.6)circle(0.2mm);
	\draw[fill](4.6,5.6)circle(0.2mm);	\draw[fill](4.8,5.6)circle(0.2mm);
	\draw[fill](4.25,5.4)circle(0.2mm);	\draw[fill](4.5,5.4)circle(0.2mm);	\draw[fill](4.75,5.4)circle(0.2mm);
	\draw[fill](4.25,5.8)circle(0.2mm);	\draw[fill](4.5,5.8)circle(0.2mm);	\draw[fill](4.75,5.8)circle(0.2mm);

	\draw[fill](5.2,0.2)circle(0.2mm);	\draw[fill](5.4,0.2)circle(0.2mm);
	\draw[fill](5.6,0.2)circle(0.2mm);	\draw[fill](5.8,0.2)circle(0.2mm);
	\draw[fill](5.2,0.6)circle(0.2mm);	\draw[fill](5.4,0.6)circle(0.2mm);
	\draw[fill](5.6,0.6)circle(0.2mm);	\draw[fill](5.8,0.6)circle(0.2mm);
	\draw[fill](5.25,0.4)circle(0.2mm);	\draw[fill](5.5,0.4)circle(0.2mm);	\draw[fill](5.75,0.4)circle(0.2mm);	
	\draw[fill](5.25,0.8)circle(0.2mm);	\draw[fill](5.5,0.8)circle(0.2mm);	\draw[fill](5.75,0.8)circle(0.2mm);			
	\draw[fill](5.2,1.2)circle(0.2mm);	\draw[fill](5.4,1.2)circle(0.2mm);
	\draw[fill](5.6,1.2)circle(0.2mm);	\draw[fill](5.8,1.2)circle(0.2mm);
	\draw[fill](5.2,1.6)circle(0.2mm);	\draw[fill](5.4,1.6)circle(0.2mm);
	\draw[fill](5.6,1.6)circle(0.2mm);	\draw[fill](5.8,1.6)circle(0.2mm);
	\draw[fill](5.25,1.4)circle(0.2mm);	\draw[fill](5.5,1.4)circle(0.2mm);	\draw[fill](5.75,1.4)circle(0.2mm);	
	\draw[fill](5.25,1.8)circle(0.2mm);	\draw[fill](5.5,1.8)circle(0.2mm);	\draw[fill](5.75,1.8)circle(0.2mm);			
	\draw[fill](5.2,2.2)circle(0.2mm);	\draw[fill](5.4,2.2)circle(0.2mm);
	\draw[fill](5.6,2.2)circle(0.2mm);	\draw[fill](5.8,2.2)circle(0.2mm);
	\draw[fill](5.2,2.6)circle(0.2mm);	\draw[fill](5.4,2.6)circle(0.2mm);
	\draw[fill](5.6,2.6)circle(0.2mm);	\draw[fill](5.8,2.6)circle(0.2mm);
	\draw[fill](5.25,2.4)circle(0.2mm);	\draw[fill](5.5,2.4)circle(0.2mm);	\draw[fill](5.75,2.4)circle(0.2mm);	
	\draw[fill](5.25,2.8)circle(0.2mm);	\draw[fill](5.5,2.8)circle(0.2mm);	\draw[fill](5.75,2.8)circle(0.2mm);		
	\draw[fill](5.2,3.2)circle(0.2mm);	\draw[fill](5.4,3.2)circle(0.2mm);
	\draw[fill](5.6,3.2)circle(0.2mm);	\draw[fill](5.8,3.2)circle(0.2mm);
	\draw[fill](5.2,3.6)circle(0.2mm);	\draw[fill](5.4,3.6)circle(0.2mm);
	\draw[fill](5.6,3.6)circle(0.2mm);	\draw[fill](5.8,3.6)circle(0.2mm);
	\draw[fill](5.25,3.4)circle(0.2mm);	\draw[fill](5.5,3.4)circle(0.2mm);	\draw[fill](5.75,3.4)circle(0.2mm);	
	\draw[fill](5.25,3.8)circle(0.2mm);	\draw[fill](5.5,3.8)circle(0.2mm);	\draw[fill](5.75,3.8)circle(0.2mm);		
	\draw[fill](5.2,4.2)circle(0.2mm);	\draw[fill](5.4,4.2)circle(0.2mm);
	\draw[fill](5.6,4.2)circle(0.2mm);	\draw[fill](5.8,4.2)circle(0.2mm);
	\draw[fill](5.2,4.6)circle(0.2mm);	\draw[fill](5.4,4.6)circle(0.2mm);
	\draw[fill](5.6,4.6)circle(0.2mm);	\draw[fill](5.8,4.6)circle(0.2mm);
	\draw[fill](5.25,4.4)circle(0.2mm);	\draw[fill](5.5,4.4)circle(0.2mm);	\draw[fill](5.75,4.4)circle(0.2mm);	
	\draw[fill](5.25,4.8)circle(0.2mm);	\draw[fill](5.5,4.8)circle(0.2mm);	\draw[fill](5.75,4.8)circle(0.2mm);			
	\draw[fill](5.2,5.2)circle(0.2mm);	\draw[fill](5.4,5.2)circle(0.2mm);
	\draw[fill](5.6,5.2)circle(0.2mm);	\draw[fill](5.8,5.2)circle(0.2mm);
	\draw[fill](5.2,5.6)circle(0.2mm);	\draw[fill](5.4,5.6)circle(0.2mm);
	\draw[fill](5.6,5.6)circle(0.2mm);	\draw[fill](5.8,5.6)circle(0.2mm);
	\draw[fill](5.25,5.4)circle(0.2mm);	\draw[fill](5.5,5.4)circle(0.2mm);	\draw[fill](5.75,5.4)circle(0.2mm);
	\draw[fill](5.25,5.8)circle(0.2mm);	\draw[fill](5.5,5.8)circle(0.2mm);	\draw[fill](5.75,5.8)circle(0.2mm);	
	
	\draw[fill](6.2,0.2)circle(0.2mm);	\draw[fill](6.4,0.2)circle(0.2mm);
	\draw[fill](6.6,0.2)circle(0.2mm);	\draw[fill](6.8,0.2)circle(0.2mm);
	\draw[fill](6.2,0.6)circle(0.2mm);	\draw[fill](6.4,0.6)circle(0.2mm);
	\draw[fill](6.6,0.6)circle(0.2mm);	\draw[fill](6.8,0.6)circle(0.2mm);
	\draw[fill](6.25,0.4)circle(0.2mm);	\draw[fill](6.5,0.4)circle(0.2mm);	\draw[fill](6.75,0.4)circle(0.2mm);	
	\draw[fill](6.25,0.8)circle(0.2mm);	\draw[fill](6.5,0.8)circle(0.2mm);	\draw[fill](6.75,0.8)circle(0.2mm);			
	\draw[fill](6.2,1.2)circle(0.2mm);	\draw[fill](6.4,1.2)circle(0.2mm);
	\draw[fill](6.6,1.2)circle(0.2mm);	\draw[fill](6.8,1.2)circle(0.2mm);
	\draw[fill](6.2,1.6)circle(0.2mm);	\draw[fill](6.4,1.6)circle(0.2mm);
	\draw[fill](5.6,1.6)circle(0.2mm);	\draw[fill](6.8,1.6)circle(0.2mm);
	\draw[fill](6.25,1.4)circle(0.2mm);	\draw[fill](6.5,1.4)circle(0.2mm);	\draw[fill](6.75,1.4)circle(0.2mm);	
	\draw[fill](6.25,1.8)circle(0.2mm);	\draw[fill](6.5,1.8)circle(0.2mm);	\draw[fill](6.75,1.8)circle(0.2mm);			
	\draw[fill](6.2,2.2)circle(0.2mm);	\draw[fill](6.4,2.2)circle(0.2mm);
	\draw[fill](6.6,2.2)circle(0.2mm);	\draw[fill](6.8,2.2)circle(0.2mm);
	\draw[fill](6.2,2.6)circle(0.2mm);	\draw[fill](6.4,2.6)circle(0.2mm);
	\draw[fill](6.6,2.6)circle(0.2mm);	\draw[fill](6.8,2.6)circle(0.2mm);
	\draw[fill](6.25,2.4)circle(0.2mm);	\draw[fill](6.5,2.4)circle(0.2mm);	\draw[fill](6.75,2.4)circle(0.2mm);	
	\draw[fill](6.25,2.8)circle(0.2mm);	\draw[fill](6.5,2.8)circle(0.2mm);	\draw[fill](6.75,2.8)circle(0.2mm);		
	\draw[fill](6.2,3.2)circle(0.2mm);	\draw[fill](6.4,3.2)circle(0.2mm);
	\draw[fill](6.6,3.2)circle(0.2mm);	\draw[fill](6.8,3.2)circle(0.2mm);
	\draw[fill](6.2,3.6)circle(0.2mm);	\draw[fill](6.4,3.6)circle(0.2mm);
	\draw[fill](6.6,3.6)circle(0.2mm);	\draw[fill](6.8,3.6)circle(0.2mm);
	\draw[fill](6.25,3.4)circle(0.2mm);	\draw[fill](6.5,3.4)circle(0.2mm);	\draw[fill](6.75,3.4)circle(0.2mm);	
	\draw[fill](6.25,3.8)circle(0.2mm);	\draw[fill](6.5,3.8)circle(0.2mm);	\draw[fill](6.75,3.8)circle(0.2mm);		
	\draw[fill](6.2,4.2)circle(0.2mm);	\draw[fill](6.4,4.2)circle(0.2mm);
	\draw[fill](6.6,4.2)circle(0.2mm);	\draw[fill](6.8,4.2)circle(0.2mm);
	\draw[fill](6.2,4.6)circle(0.2mm);	\draw[fill](6.4,4.6)circle(0.2mm);
	\draw[fill](6.6,4.6)circle(0.2mm);	\draw[fill](6.8,4.6)circle(0.2mm);
	\draw[fill](6.25,4.4)circle(0.2mm);	\draw[fill](6.5,4.4)circle(0.2mm);	\draw[fill](6.75,4.4)circle(0.2mm);	
	\draw[fill](6.25,4.8)circle(0.2mm);	\draw[fill](6.5,4.8)circle(0.2mm);	\draw[fill](6.75,4.8)circle(0.2mm);			
	\draw[fill](6.2,5.2)circle(0.2mm);	\draw[fill](6.4,5.2)circle(0.2mm);
	\draw[fill](6.6,5.2)circle(0.2mm);	\draw[fill](6.8,5.2)circle(0.2mm);
	\draw[fill](6.2,5.6)circle(0.2mm);	\draw[fill](6.4,5.6)circle(0.2mm);
	\draw[fill](6.6,5.6)circle(0.2mm);	\draw[fill](6.8,5.6)circle(0.2mm);
	\draw[fill](6.25,5.4)circle(0.2mm);	\draw[fill](6.5,5.4)circle(0.2mm);	\draw[fill](6.75,5.4)circle(0.2mm);
	\draw[fill](6.25,5.8)circle(0.2mm);	\draw[fill](6.5,5.8)circle(0.2mm);	\draw[fill](6.75,5.8)circle(0.2mm);	
	
	\draw[fill](7.2,0.2)circle(0.2mm);	\draw[fill](7.4,0.2)circle(0.2mm);
	\draw[fill](7.6,0.2)circle(0.2mm);	\draw[fill](7.8,0.2)circle(0.2mm);
	\draw[fill](7.2,0.6)circle(0.2mm);	\draw[fill](7.4,0.6)circle(0.2mm);
	\draw[fill](7.6,0.6)circle(0.2mm);	\draw[fill](7.8,0.6)circle(0.2mm);
	\draw[fill](7.25,0.4)circle(0.2mm);	\draw[fill](7.5,0.4)circle(0.2mm);	\draw[fill](7.75,0.4)circle(0.2mm);	
	\draw[fill](7.25,0.8)circle(0.2mm);	\draw[fill](7.5,0.8)circle(0.2mm);	\draw[fill](7.75,0.8)circle(0.2mm);			
	\draw[fill](7.2,1.2)circle(0.2mm);	\draw[fill](7.4,1.2)circle(0.2mm);
	\draw[fill](7.6,1.2)circle(0.2mm);	\draw[fill](7.8,1.2)circle(0.2mm);
	\draw[fill](7.2,1.6)circle(0.2mm);	\draw[fill](7.4,1.6)circle(0.2mm);
	\draw[fill](7.6,1.6)circle(0.2mm);	\draw[fill](7.8,1.6)circle(0.2mm);
	\draw[fill](7.25,1.4)circle(0.2mm);	\draw[fill](7.5,1.4)circle(0.2mm);	\draw[fill](7.75,1.4)circle(0.2mm);	
	\draw[fill](7.25,1.8)circle(0.2mm);	\draw[fill](7.5,1.8)circle(0.2mm);	\draw[fill](7.75,1.8)circle(0.2mm);			
	\draw[fill](7.2,2.2)circle(0.2mm);	\draw[fill](7.4,2.2)circle(0.2mm);
	\draw[fill](7.6,2.2)circle(0.2mm);	\draw[fill](7.8,2.2)circle(0.2mm);
	\draw[fill](7.2,2.6)circle(0.2mm);	\draw[fill](7.4,2.6)circle(0.2mm);
	\draw[fill](7.6,2.6)circle(0.2mm);	\draw[fill](7.8,2.6)circle(0.2mm);
	\draw[fill](7.25,2.4)circle(0.2mm);	\draw[fill](7.5,2.4)circle(0.2mm);	\draw[fill](7.75,2.4)circle(0.2mm);	
	\draw[fill](7.25,2.8)circle(0.2mm);	\draw[fill](7.5,2.8)circle(0.2mm);	\draw[fill](7.75,2.8)circle(0.2mm);		
	\draw[fill](7.2,3.2)circle(0.2mm);	\draw[fill](7.4,3.2)circle(0.2mm);
	\draw[fill](7.6,3.2)circle(0.2mm);	\draw[fill](7.8,3.2)circle(0.2mm);
	\draw[fill](7.2,3.6)circle(0.2mm);	\draw[fill](7.4,3.6)circle(0.2mm);
	\draw[fill](7.6,3.6)circle(0.2mm);	\draw[fill](7.8,3.6)circle(0.2mm);
	\draw[fill](7.25,3.4)circle(0.2mm);	\draw[fill](7.5,3.4)circle(0.2mm);	\draw[fill](7.75,3.4)circle(0.2mm);	
	\draw[fill](7.25,3.8)circle(0.2mm);	\draw[fill](7.5,3.8)circle(0.2mm);	\draw[fill](7.75,3.8)circle(0.2mm);		
	\draw[fill](7.2,4.2)circle(0.2mm);	\draw[fill](7.4,4.2)circle(0.2mm);
	\draw[fill](7.6,4.2)circle(0.2mm);	\draw[fill](7.8,4.2)circle(0.2mm);
	\draw[fill](7.2,4.6)circle(0.2mm);	\draw[fill](7.4,4.6)circle(0.2mm);
	\draw[fill](7.6,4.6)circle(0.2mm);	\draw[fill](7.8,4.6)circle(0.2mm);
	\draw[fill](7.25,4.4)circle(0.2mm);	\draw[fill](7.5,4.4)circle(0.2mm);	\draw[fill](7.75,4.4)circle(0.2mm);	
	\draw[fill](7.25,4.8)circle(0.2mm);	\draw[fill](7.5,4.8)circle(0.2mm);	\draw[fill](7.75,4.8)circle(0.2mm);			
	\draw[fill](7.2,5.2)circle(0.2mm);	\draw[fill](7.4,5.2)circle(0.2mm);
	\draw[fill](7.6,5.2)circle(0.2mm);	\draw[fill](7.8,5.2)circle(0.2mm);
	\draw[fill](7.2,5.6)circle(0.2mm);	\draw[fill](7.4,5.6)circle(0.2mm);
	\draw[fill](7.6,5.6)circle(0.2mm);	\draw[fill](7.8,5.6)circle(0.2mm);
	\draw[fill](7.25,5.4)circle(0.2mm);	\draw[fill](7.5,5.4)circle(0.2mm);	\draw[fill](7.75,5.4)circle(0.2mm);
	\draw[fill](7.25,5.8)circle(0.2mm);	\draw[fill](7.5,5.8)circle(0.2mm);	\draw[fill](7.75,5.8)circle(0.2mm);	
	
	\end{tikzpicture}
	\caption{Secure piggyback codebook with subcodebooks arranged in an array. Dots indicate codewords.} \label{fig:piggy} 
\end{center}
\end{figure}

%% file: random_keys_all.tex
\begin{figure}[H]
	\centering
	\begin{tikzpicture}
	\node[above] at (1.2,1.1) {Cache at Rx\! $1$};
	\draw[rounded corners=7pt,thick] (0,0) rectangle (2.4,1);
	\node at (1.2,0.5) {$K_1$};
	
		\node[above] at (4.2,1.1) {Cache at Rx\! $2$};
	\draw[rounded corners=7pt,thick] (3,0) rectangle (5.4,1);
	\node at (4.2,0.5) {$K_2$};
	
		\node[above] at (7.2,1.1) {Cache at Rx\! $3$};
	\draw[rounded corners=7pt,thick] (6,0) rectangle (8.4,1);
	\node at (7.2,0.5) {$K_3$};

		\node[above] at (12.2,1.1) {Cache at Rx\! $\Ka$};
	\draw[rounded corners=7pt,thick] (11,0) rectangle (13.4,1);
	\node at (12.2,0.5) {$K_{\Ka}$};
 	\end{tikzpicture}
 	
\end{figure}

%% file: cache_contents_piggyback_all.tex
 \begin{figure}[H] \centering
 	\begin{tikzpicture}	
 	\node[above] at (1.4,3.5) {Cache at Rx$1$};
 	\draw[rounded corners=7pt,thick] (-0.4,0) rectangle (3.2,3.5);
 	\node at (1.4,3) {$\left\lbrace W_{d,\{1\}}^{(A)}\right\rbrace_{\!d=1}^{\!\Da}$};	
 	\node at (1.4,2.25) {$\left\lbrace W_{d,\{1,2\}}^{(B)}, W_{d,\{1,3\}}^{(B)}\right\rbrace_{\!d=1}^{\!\Da}$};
 	\node at (1.4,1.6) {$K_{\{1,2,3\}}$};	
 	\node at (1.4,1) {$K_{\{1,2\}},K_{\{1,3\}}$};	
 	\node at (1.4,0.4) {$K_{4,\{1,2\}},K_{4,\{1,3\}}$};	
 	
 	\node[above] at (5.6,3.5) {Cache at Rx$2$};
 	\draw[rounded corners=7pt,thick] (3.8,0) rectangle (7.4,3.5);
 	\node at (5.6,3) {$\left\lbrace W_{d,\{2\}}^{(A)}\right\rbrace_{\!d=1}^{\!\Da}$};	
 	\node at (5.6,2.25) {$\left\lbrace W_{d,\{1,2\}}^{(B)}, W_{d,\{2,3\}}^{(B)}\right\rbrace_{\!d=1}^{\!\Da}$};
 	\node at (5.6,1.6) {$K_{\{1,2,3\}}$};	
 	\node at (5.6,1) {$K_{\{1,2\}},K_{\{2,3\}}$};	
 	\node at (5.6,0.4) {$K_{4,\{1,2\}},K_{4,\{2,3\}}$};	
 	
 	\node[above] at (9.8,3.5) {Cache at Rx$3$};
 	\draw[rounded corners=7pt,thick] (8,0) rectangle (11.6,3.5);
 	\node at (9.8,3) {$\left\lbrace W_{d,\{3\}}^{(A)}\right\rbrace_{\!d=1}^{\!\Da}$};	
 	\node at (9.8,2.25) {$\left\lbrace W_{d,\{1,3\}}^{(B)}, W_{d,\{2,3\}}^{(B)}\right\rbrace_{\!d=1}^{\!\Da}$};
 	\node at (9.8,1.6) {$K_{\{1,2,3\}}$};	
 	\node at (9.8,1) {$K_{\{1,3\}},K_{\{2,3\}}$};	
 	\node at (9.8,0.4) {$K_{4,\{1,3\}},K_{4,\{2,3\}}$};	
 	
 	\node[above] at (13.5,3.5) {Cache at Rx$4$};
 	\draw[rounded corners=7pt,thick] (12.5,0.9) rectangle (14.5,3.5);
 	\node at (13.5,3.1) {$K_4$};	
 	\node at (13.5,2.5) {$K_{4,\{1,2\}}$};
 	\node at (13.5,1.9) {$K_{4,\{1,3\}}$};	
 	\node at (13.5,1.3) {$K_{4,\{2,3\}}$};
 	\end{tikzpicture} 
 \end{figure}

%% file: piggyback.tex
\begin{figure}[h!] \centering 
	\begin{tikzpicture} 
	\draw[thick] (0,0) rectangle (8,6);
	
	\draw (0,5) -- (8,5);	\draw (0,4) -- (8,4);	\draw (0,3) -- (8,3);
	\draw (0,2) -- (8,2);	\draw (0,1) -- (8,1);
	
	\draw (1,0) -- (1,6);	\draw (2,0) -- (2,6);	\draw (3,0) -- (3,6);	\draw (4,0) -- (4,6);
	\draw (5,0) -- (5,6);	\draw (6,0) -- (6,6);	\draw (7,0) -- (7,6);
	
	\draw[very thick,blue] (6,0) rectangle (7,6);
	\draw[very thick,blue] (0,4) rectangle (8,5);	
	\node[left,blue] at (0,4.5) {$\mathbf{W}_{\sec,\{1,2\}}^{(A)}$}; 
	\node[above,blue] at (6.5,6) {${W}_{\sec,\{4\}}^{(B)}$}; 
\draw[ultra thick,magenta] (6,4) rectangle (7,5);
		\draw[magenta,thick,->] (7,4.65) -- (7.5,5.25) -- (8.3,5.25);
		\node[right,magenta] at (8.3,5.2) {{Codeword corresponding}};
		\node[right,magenta] at (8.3,4.65) {to $\Big(\mathbf{W}_{\sec,\{1,2\}}^{(A)}, {W}_{\sec,\{4\}}^{(B)}\Big)$};
	
	\draw[fill](0.5,0.5)circle(1mm);	\draw[fill](1.5,0.5)circle(1mm);	\draw[fill](2.5,0.5)circle(1mm);
	\draw[fill](3.5,0.5)circle(1mm);	\draw[fill](4.5,0.5)circle(1mm);	\draw[fill](5.5,0.5)circle(1mm);
	\draw[fill](6.5,0.5)circle(1mm);	\draw[fill](7.5,0.5)circle(1mm);
	
	\draw[fill](0.5,1.5)circle(1mm);	\draw[fill](1.5,1.5)circle(1mm);	\draw[fill](2.5,1.5)circle(1mm);
	\draw[fill](3.5,1.5)circle(1mm);	\draw[fill](4.5,1.5)circle(1mm);	\draw[fill](5.5,1.5)circle(1mm);
	\draw[fill](6.5,1.5)circle(1mm);	\draw[fill](7.5,1.5)circle(1mm);
	
	\draw[fill](0.5,2.5)circle(1mm);	\draw[fill](1.5,2.5)circle(1mm);	\draw[fill](2.5,2.5)circle(1mm);
	\draw[fill](3.5,2.5)circle(1mm);	\draw[fill](4.5,2.5)circle(1mm);	\draw[fill](5.5,2.5)circle(1mm);
	\draw[fill](6.5,2.5)circle(1mm);	\draw[fill](7.5,2.5)circle(1mm);
	
	\draw[fill](0.5,3.5)circle(1mm);	\draw[fill](1.5,3.5)circle(1mm);	\draw[fill](2.5,3.5)circle(1mm);
	\draw[fill](3.5,3.5)circle(1mm);	\draw[fill](4.5,3.5)circle(1mm);	\draw[fill](5.5,3.5)circle(1mm);
	\draw[fill](6.5,3.5)circle(1mm);	\draw[fill](7.5,3.5)circle(1mm);
	
	\draw[fill](0.5,4.5)circle(1mm);	\draw[fill](1.5,4.5)circle(1mm);	\draw[fill](2.5,4.5)circle(1mm);
	\draw[fill](3.5,4.5)circle(1mm);	\draw[fill](4.5,4.5)circle(1mm);	\draw[fill](5.5,4.5)circle(1mm);
	\draw[fill](6.5,4.5)circle(1mm);	\draw[fill](7.5,4.5)circle(1mm);
	
	\draw[fill](0.5,5.5)circle(1mm);	\draw[fill](1.5,5.5)circle(1mm);	\draw[fill](2.5,5.5)circle(1mm);
	\draw[fill](3.5,5.5)circle(1mm);	\draw[fill](4.5,5.5)circle(1mm);	\draw[fill](5.5,5.5)circle(1mm);
	\draw[fill](6.5,5.5)circle(1mm);	\draw[fill](7.5,5.5)circle(1mm);

	\end{tikzpicture}
	\caption{Standard piggyback codebook where only a single codeword (indicated by a single dot) is assigned to each pair of messages.} \label{fig:piggy_standard} 
\end{figure}